\documentclass[english]{article}
\usepackage[T1]{fontenc}
\usepackage[latin9]{inputenc}
\usepackage{geometry}
\usepackage{color}
\usepackage{array}
\usepackage{float}
\usepackage{mathtools}
\usepackage{bm}
\usepackage{multirow}
\usepackage{amsmath}
\usepackage{amsthm}
\usepackage{amssymb}
\usepackage{tablefootnote}
\PassOptionsToPackage{normalem}{ulem}
\usepackage{ulem}

\makeatletter

\newcommand{\noun}[1]{\textsc{#1}}
\providecommand{\tabularnewline}{\\}
\floatstyle{ruled}
\newfloat{algorithm}{tbp}{loa}
\providecommand{\algorithmname}{Algorithm}
\floatname{algorithm}{\protect\algorithmname}

\theoremstyle{plain}
\newtheorem{thm}{\protect\theoremname}
\theoremstyle{plain}
\newtheorem{prop}[thm]{\protect\propositionname}
\theoremstyle{remark}
\newtheorem{rem}[thm]{\protect\remarkname}
\theoremstyle{definition}
\newtheorem{defn}[thm]{\protect\definitionname}
\theoremstyle{plain}
\newtheorem{lem}[thm]{\protect\lemmaname}
\theoremstyle{definition}
\newtheorem{example}[thm]{\protect\examplename}
\theoremstyle{plain}
\newtheorem{fact}[thm]{\protect\factname}

\usepackage{algorithmic}
\usepackage{qcircuit}
\usepackage{array}
\usepackage{tkz-graph}
\usepackage{placeins}

\makeatother

\usepackage{babel}
\providecommand{\definitionname}{Definition}
\providecommand{\examplename}{Example}
\providecommand{\factname}{Fact}
\providecommand{\lemmaname}{Lemma}
\providecommand{\propositionname}{Proposition}
\providecommand{\remarkname}{Remark}
\providecommand{\theoremname}{Theorem}

\begin{document}
\global\long\def\R{\mathbb{R}}%

\global\long\def\C{\mathbb{C}}%

\global\long\def\N{\mathbb{N}}%

\global\long\def\e{{\mathbf{e}}}%

\global\long\def\et#1{{\e(#1)}}%

\global\long\def\ef{{\mathbf{\et{\cdot}}}}%

\global\long\def\a{{\mathbf{a}}}%

\global\long\def\x{{\mathbf{x}}}%

\global\long\def\xt#1{{\x(#1)}}%

\global\long\def\xf{{\mathbf{\xt{\cdot}}}}%

\global\long\def\d{{\mathbf{d}}}%

\global\long\def\w{{\mathbf{w}}}%

\global\long\def\b{{\mathbf{b}}}%

\global\long\def\u{{\mathbf{u}}}%

\global\long\def\y{{\mathbf{y}}}%

\global\long\def\k{{\mathbf{k}}}%

\global\long\def\yt#1{{\y(#1)}}%

\global\long\def\yf{{\mathbf{\yt{\cdot}}}}%

\global\long\def\z{{\mathbf{z}}}%

\global\long\def\v{{\mathbf{v}}}%

\global\long\def\h{{\mathbf{h}}}%

\global\long\def\s{{\mathbf{s}}}%

\global\long\def\c{{\mathbf{c}}}%

\global\long\def\p{{\mathbf{p}}}%

\global\long\def\f{{\mathbf{f}}}%

\global\long\def\t{{\mathbf{t}}}%

\global\long\def\rb{{\mathbf{r}}}%

\global\long\def\rt#1{{\rb(#1)}}%

\global\long\def\rf{{\mathbf{\rt{\cdot}}}}%

\global\long\def\mat#1{{\ensuremath{\bm{\mathrm{#1}}}}}%

\global\long\def\matN{\ensuremath{{\bm{\mathrm{N}}}}}%

\global\long\def\matX{\ensuremath{{\bm{\mathrm{X}}}}}%

\global\long\def\matK{\ensuremath{{\bm{\mathrm{K}}}}}%

\global\long\def\matA{\ensuremath{{\bm{\mathrm{A}}}}}%

\global\long\def\matB{\ensuremath{{\bm{\mathrm{B}}}}}%

\global\long\def\matC{\ensuremath{{\bm{\mathrm{C}}}}}%

\global\long\def\matD{\ensuremath{{\bm{\mathrm{D}}}}}%

\global\long\def\matE{\ensuremath{{\bm{\mathrm{E}}}}}%

\global\long\def\matF{\ensuremath{{\bm{\mathrm{F}}}}}%

\global\long\def\matO{\ensuremath{{\bm{\mathrm{O}}}}}%

\global\long\def\matQ{\ensuremath{{\bm{\mathrm{Q}}}}}%

\global\long\def\matP{\ensuremath{{\bm{\mathrm{P}}}}}%

\global\long\def\matU{\ensuremath{{\bm{\mathrm{U}}}}}%

\global\long\def\matV{\ensuremath{{\bm{\mathrm{V}}}}}%

\global\long\def\matM{\ensuremath{{\bm{\mathrm{M}}}}}%

\global\long\def\matG{\ensuremath{{\bm{\mathrm{G}}}}}%

\global\long\def\calH{{\cal H}}%

\global\long\def\calS{{\cal S}}%

\global\long\def\calT{{\cal T}}%

\global\long\def\matR{\mat R}%

\global\long\def\matS{\mat S}%

\global\long\def\matO{\mat O}%

\global\long\def\matT{\mat T}%

\global\long\def\matY{\mat Y}%

\global\long\def\matI{\mat I}%

\global\long\def\matJ{\mat J}%

\global\long\def\matZ{\mat Z}%

\global\long\def\matW{\mat W}%

\global\long\def\tmatK{\widetilde{\matK}}%

\global\long\def\matL{\mat L}%

\global\long\def\matZero{\mat 0}%

\global\long\def\S#1{{\mathbb{S}_{N}[#1]}}%

\global\long\def\IS#1{{\mathbb{S}_{N}^{-1}[#1]}}%

\global\long\def\PN{\mathbb{P}_{N}}%

\global\long\def\Norm#1{\|#1\|}%

\global\long\def\NormS#1{\|#1\|^{2}}%

\global\long\def\ONorm#1{\|#1\|_{op}}%

\global\long\def\ONormS#1{\|#1\|_{op}^{2}}%

\global\long\def\OneNorm#1{\|#1\|_{1}}%

\global\long\def\OneNormS#1{\|#1\|_{1}^{2}}%

\global\long\def\TNormS#1{\|#1\|_{2}^{2}}%

\global\long\def\TNorm#1{\|#1\|_{2}}%

\global\long\def\InfNorm#1{\|#1\|_{\infty}}%

\global\long\def\InfNormS#1{\|#1\|_{\infty}^{2}}%

\global\long\def\FNorm#1{\|#1\|_{F}}%

\global\long\def\FNormK#1#2{\|#1\|_{F}^{#2}}%

\global\long\def\FNormS#1{\|#1\|_{F}^{2}}%

\global\long\def\UNorm#1{\|#1\|_{\matU}}%

\global\long\def\NucNorm#1{\|#1\|_{\star}}%

\global\long\def\UNormS#1{\|#1\|_{\matU}^{2}}%

\global\long\def\UINormS#1{\|#1\|_{\matU^{-1}}^{2}}%

\global\long\def\ANorm#1{\|#1\|_{\matA}}%

\global\long\def\BNorm#1{\|#1\|_{\mat B}}%

\global\long\def\ANormS#1{\|#1\|_{\matA}^{2}}%

\global\long\def\AINormS#1{\|#1\|_{\matA^{-1}}^{2}}%

\global\long\def\T{\textsc{T}}%

\global\long\def\conj{\textsc{\ensuremath{*}}}%

\global\long\def\pinv{\textsc{+}}%

\global\long\def\Var#1{{\mathbb{V}}\left[#1\right]}%

\global\long\def\Expect#1{{\mathbb{E}}\left[#1\right]}%

\global\long\def\ExpectC#1#2{{\mathbb{E}}_{#1}\left[#2\right]}%

\global\long\def\dotprod#1#2#3{(#1,#2)_{#3}}%

\global\long\def\dotprodN#1#2{(#1,#2)_{{\cal N}}}%

\global\long\def\dotprodH#1#2{(#1,#2)_{{\cal {\cal H}}}}%

\global\long\def\dotprodsqr#1#2#3{(#1,#2)_{#3}^{2}}%

\global\long\def\Trace#1{{\bf Tr}\left(#1\right)}%

\global\long\def\STrace#1{{\bf Tr}(#1)}%

\global\long\def\MTrace#1#2{{\bf MTr}_{#2}\left(#1\right)}%

\global\long\def\SMTrace#1#2{{\bf MTr}_{#2}(#1)}%

\global\long\def\nnz#1{{\bf nnz}\left(#1\right)}%

\global\long\def\vec#1{{\bf vec}\left(#1\right)}%

\global\long\def\Mat#1{{\bf mat}\left(#1\right)}%

\global\long\def\rank#1{{\bf rank}\left(#1\right)}%

\global\long\def\vol#1{{\bf vol}\left(#1\right)}%

\global\long\def\cirmat#1{\matM\left(#1\right)}%

\global\long\def\range#1{{\bf range}\left(#1\right)}%

\global\long\def\sr#1{{\bf sr}\left(#1\right)}%

\global\long\def\poly#1{{\bf poly}\left(#1\right)}%

\global\long\def\amplitude#1#2{{\bf Amplitude}_{#2}\left(#1\right)}%

\global\long\def\gap#1#2{{\bf gap}_{#2}\left(#1\right)}%

\global\long\def\gapS#1#2{{\bf gap}_{#2}^{2}\left(#1\right)}%

\global\long\def\diag#1{{\bf diag}\left(#1\right)}%

\global\long\def\Re#1{{\bf Re}\left(#1\right)}%

\global\long\def\Img#1{{\bf Im}\left(#1\right)}%

\global\long\def\Gr#1#2{{\bf Gr}(#1,#2)}%

\global\long\def\Ket#1{\left\vert #1\right\rangle }%

\global\long\def\Bra#1{\left\langle #1\right\vert }%

\global\long\def\kket#1{\left\vert \left\vert #1\right\rangle \right\rangle }%

\global\long\def\bbra#1{\left\langle \left\langle #1\right\vert \right\vert }%

\global\long\def\Braket#1#2{\left\langle #1|#2\right\rangle }%

\global\long\def\Ptrace#1#2{{\bf {\bf Tr}}_{#2}\left(#1\right)}%

\global\long\def\SPtrace#1#2{{\bf {\bf Tr}}_{#2}(#1)}%

\global\long\def\vtheta{\mat{\theta}}%

\global\long\def\valpha{\mat{\alpha}}%

\global\long\def\vbeta{\mat{\beta}}%

\global\long\def\vpsi{\mat{\psi}}%

\global\long\def\vphi{\mat{\phi}}%

\global\long\def\vxi{\mat{\xi}}%

\global\long\def\abs#1{\lvert#1\rvert}%

\global\long\def\HT{\operatorname{HT}}%

\global\long\def\ST{\operatorname{ST}}%

\title{Multivariate trace estimation using quantum state space linear algebra}
\author{Liron Mor Yosef\thanks{Tel Aviv University, lm2@mail.tau.ac.il} \\
 \and Shashanka Ubaru\thanks{IBM T.J. Watson Research Center, Shashanka.Ubaru@ibm.com}
\\
 \and Lior Horesh\thanks{IBM T.J. Watson Research Center, lhoresh@us.ibm.com}
\\
 \and Haim Avron\thanks{Tel Aviv University, haimav@tauex.tau.ac.il}
\\
}
\maketitle
\begin{abstract}
In this paper, we present a quantum algorithm for approximating multivariate
traces, i.e. the traces of matrix products. Our research is motivated
by the extensive utility of multivariate traces in elucidating spectral
characteristics of matrices, as well as by recent advancements in
leveraging quantum computing for faster numerical linear algebra.
Central to our approach is a direct translation of a multivariate
trace formula into a quantum circuit, achieved through a sequence
of low-level circuit construction operations. To facilitate this translation,
we introduce \emph{quantum Matrix States Linear Algebra} (qMSLA),
a framework tailored for the efficient generation of state preparation
circuits via primitive matrix algebra operations. Our algorithm relies
on sets of state preparation circuits for input matrices as its primary
inputs and yields two state preparation circuits encoding the multivariate
trace as output. These circuits are constructed utilizing qMSLA operations,
which enact the aforementioned multivariate trace formula. We emphasize
that our algorithm's inputs consist solely of state preparation circuits,
eschewing harder to synthesize constructs such as Block Encodings.
Furthermore, our approach operates independently of the availability
of specialized hardware like QRAM, underscoring its versatility and
practicality.
\end{abstract}

\section{Introduction}

In this paper we propose a quantum algorithm for approximating multivariate
traces, i.e. traces of the products of matrices. Formally, the multivariate
trace of matrices $\matA_{1},\matA_{2},\dots,\matA_{k}$ of compatible
dimensions is defined as:
\begin{equation}
\MTrace{\matA_{1},\dots,\matA_{k}}k\coloneqq\Trace{\matA_{1}\matA_{2}\cdots\matA_{k}}.\label{eq:multivariate-trace}
\end{equation}
Approximating multivariate traces is motivated as follows. For a square
$\matA\in\R^{n\times n}$, its matrix moments $\STrace{\matA^{k}}=\MTrace{\matA,\dots,\matA}k$
for $k=1,2,\dots$, which reveal useful spectral properties of $\matA$
with applications in scientific computing~\cite{golub2009matrices}
and other fields, are multivariate traces. Furthermore, by introducing
a set of shifts we have $\MTrace{\matA-\alpha_{1}\matI,\dots,\matA-\alpha_{k}\matI}k=\Trace{p(\matA)}$
for some polynomial $p(x)$ of degree $k$ and leading coefficient
1. The right hand side $\Trace{p(\matA)}$ is equal to $\sum_{i=1}^{n}p(\lambda_{i})$
where $\lambda_{1},\dots,\lambda_{n}$ are the eigenvalues of $\matA$,
which is a spectral sum. Since many smooth functions $f(x)$ can be
approximated well using polynomials~\cite{trefethen2013approximation},
it is a common practice to approximate spectral sums of the general
form $\sum_{i=1}^{n}f(\lambda_{i})$ using restricted forms $\Trace{p(\matA)}$.
A well known example is $\log\det\matA=\sum_{i=1}^{n}\log(\lambda_{i}$).
Many machine learning techniques estimate spectral properties of various
matrices by approximating appropriate spectral sums via polynomial
spectral sums $\Trace{p(\matA)}$, e.g. Gaussian Processes~\cite{Rasmussen2006},
kernel learning~\cite{davis2007information}, Bayesian learning~\cite{mackay2003information},
matrix completion~\cite{candes2009exact}, differential privacy problems~\cite{hardt2012simple},
graph analysis~\cite{estrada2000characterization}, Hessian and neural
network property analysis~\cite{ramesh2018backpropagation,ghorbani2019investigation}
and many more~\cite{ubaru2018applications,chen2021analysis}. Other
applications of multivariate traces appear in computational physics,
e.g. entanglement estimation~\cite{horodecki2002method,leifer2004measuring,brydges2019probing,feldman-2022-entan-estim,quek2024multivariate},
quantum error mitigation~\cite{liang2023unified}, quantum distinguishability
measures~\cite{buhrman2001quantum}.

Due to the ubiquity of matrix moments and spectral sums (and more
generally multivariate traces), there is an extensive literature on
efficient algorithms for spectral sum approximation. Example of classical
(conventional) algorithms for spectral sum approximation\footnote{In the context of this paper, whenever we talk about ``classical
algorithms'' we mean algorithms that use only classical computing,
as opposed to ``quantum algorithm'' that also use quantum computing
(but may have a classical component).} are~\cite{han2017approximating,ubaru-2017-fast-estim}. Many of
these classical methods are based on (a) estimating multivariate traces
using randomized trace estimation techniques (Hutchinson's method
\cite{hutchinson1990stochastic} and related methods \cite{bai1996some,avron2011randomized,roosta2015improved,meyer2021hutchpp,cortinovis2022randomized}),
and (b) approximating the function or quadratic form using polynomial
approximations, Lanczos Gaussian quadrature, histogram or some other
rational approximation. While these methods are powerful and useful,
with the constant increase in the size of matrices encountered in
applications, and ever increasing accuracy requirements (which lead
to larger degrees in approximating polynomials), developing novel
algorithms for trace estimation is an active research topic.

Quantum computers can efficiently perform certain linear-algebraic
operations on large computational spaces (exponential in the number
of qubits), and offer the potential to achieve significant speedups
over classical computations. In recent years, a variety of quantum
numerical linear algebra (qNLA)~\cite{prakash2014quantum,luongo2020quantum,gilyen2019quantum,cade2017quantum}
and quantum machine learning (QML)~\cite{schuld2019quantum,schuld2015introduction,biamonte2017quantum}
methods have been proposed to harness the computational power of quantum
computing to achieve polynomial to exponential speedups over the best-known
classical methods. In this paper, we consider the problem of estimating
multivariate traces.

\subsection{Contributions}

Most of the previous literature on qNLA and QML propose techniques
that operate under stringent quantum input models, e.g. ones that
assume availability of speculative components such as QRAM~\cite{giovannetti2008quantum}.
However, in practice, availability of such components is non trivial,
e.g. efficient QRAMs are not currently available (and are not expected
to be available the near-term future). Moreover, recent results on
\emph{classical dequantization }(aka \emph{quantum-inspired classical
algorithms})~\cite{chepurko2022quantum,chia2020sampling,tang2019quantum}
have shown that under analogous (to QRAM) classical data-structure
assumptions, ``dequantized'' classical algorithms can achieve similar
runtimes as their quantum counterparts. Therefore, we advocate the
development of qNLA algorithms that operate under only lax\textcolor{red}{{}
}quantum input models. In particular, algorithms that make no quantum
access assumptions, i.e. start with a classical description of input
data, are especially attractive. Thus, in this work, we develop a
quantum algorithm for multivariate trace estimation that uses an input
model where both matrices and vectors are presented as quantum state
preparation circuits, rather than relying on more restrictive input
models.

In designing our quantum algorithm for trace estimation we take an
approach that differs from the aforementioned qNLA and QML works.
Our algorithm for multivariate trace estimation is based on \textit{Quantum
Matrix State Linear Algebra} (qMSLA), a novel framework for performing
numerical linear algebra computations on quantum computers. qMSLA
utilizes matrix state preparation circuits as a practical and efficient
construct for representing and manipulating matrices and vectors in
quantum algorithms. The key concept behind qMSLA is the development
of a set of basic linear algebra operations that can be executed directly
on circuits that represent matrices in probability amplitudes of quantum
states. For example, given matrix state preparation circuits for matrices
$\matA$ and $\matB$, qMSLA operation $\text{\textrm{qmsla.kronecker}}$
outputs a matrix state preparation circuit for $\matA\otimes\matB$.
These operations enable the execution of various matrix algebra tasks
entirely within the quantum domain. In Section~\ref{sec:qmsla} we
propose a limited yet useful set of qMSLA operations, and show how
they can be implemented efficiently.

Although there are foundational quantum circuits and algorithms, QLA
and otherwise, and common patterns for designing quantum circuits,
these represent high level operations and do not provide general purpose
matrix algebra subroutines, which are necessary building blocks for
computational linear algebra. qMSLA allows practitioners to overcome
this limitation by offering a matrix-level framework for doing linear
algebra on quantum computers, and thus designing quantum circuits
by directly synthesizing quantum circuits that implement matrix equations.
We believe that the idea of constructing quantum circuits of target
matrix computations (such as the trace) using a series of primitive
matrix algebra operations at the circuit level might be of independent
interest, as a new paradigm for doing matrix computations in the quantum
domain. Indeed, our approach is very much inspired by BLAS, and we
view qMSLA as a (as of yet limited) BLAS for QLA.

We leverage qMSLA to design a novel algorithm for solving multivariate
trace problems. Along with fulfilling the main goal of this study,
this also showcases the potential of the qMSLA approach. The algorithm
seamlessly executes a sequences of high-level qMSLA operation, eliminating
the need for a separate, custom-built quantum circuit. The algorithm's
input are pre-designed circuits, denoted as ${\cal U}_{\text{\ensuremath{\matA}}_{i}}$
for matrices $\matA_{i}$ ($i$ from $1$ to $2k$)\footnote{In this study, we have develop an algorithm specifically tailored
for the multiplication of an even number of matrices. Odd number of
matrices can be dealt with by introducing the identity matrix as one
of the input matrices.}. These circuits encode the matrices into quantum states, that is,
${\cal U}_{\text{\ensuremath{\matA}}_{i}}$ implements a amplitude
encoding procedure for $\matA_{i}$~\cite[eq 3.60, p.g 115]{schuld2021machine}.
Our algorithm outputs two circuits, ${\cal U}_{\vpsi}$ and ${\cal U}_{\vphi}$,
that prepare two quantum states $\Ket{\vpsi}$ and $\Ket{\vphi}$,
such that their dot product (overlap in quantum computing parlance)\textcolor{red}{{}
}is equal to the multivariate trace. This overlap can then be estimated
using various existing quantum algorithms (such as Hadamard Test and
Swap Test). Alternatively, the circuits ${\cal U}_{\vpsi}$ and ${\cal U}_{\vphi}$
can be used in the context of a larger quantum algorithm. We emphasize
that our approach is able to handle matrices that are not necessarily
Hermitian or even square. We also discuss how our algorithm can be
used to approximate multivariate traces of matrices given in classical
memory, showing that even in this case, under certain conditions,
our algorithm entertains a small computational gain over state-of-the-art
classical trace estimators. We also discuss how our algorithm can
be used to approximate spectral sums.

The proposed algorithm makes only mild assumptions on the input matrices,
namely that it has access to matrix state preparation circuits for
each input matrix. Construction of such state preparation circuits
for some matrices can be inexpensive (possibly polylogarithmic in
matrix size)~\cite[Section 4.2.2]{gleinig2021efficient,iaconis2024quantum,schuld2021machine}.
At the baseline, for a general matrix available in classical memory,
it is possible to construct a state preparation circuit for it in
time that is linear in the number of entries in the matrix and has
similar depth~\cite{shende2005synthesis}, and this is the basis
of a end-to-end classical-to-classical-via-quantum use of our algorithm.

\subsection{Related work}

For the classical setting, there is scarce literature on computing
multivariate traces per se, but there is extensive literature on stochastic
estimation of the trace of implicit matrices. The latter can be used
to estimate multivariate traces. Virtually all methods for stochastic
trace estimation trace their origin to Hutchinon's work~\cite{hutchinson1990stochastic},
and the well-known Hutchinson's estimator, though Hutchinson actually
cites Girard \cite{girard1987algorithme} in his abstract. Extensive
followup, such as \cite{bai1996some,avron2011randomized,roosta2015improved,cortinovis2022randomized},
improved on Hutchinson's results. The state-of-the-art estimator is
Hutch++~~\cite{meyer2021hutchpp}, which demonstrates improved convergence
rates through randomized low-rank approximations. Classical stochastic
trace estimators are often combined with polynomial approximations
to design stochastic estimators for spectral sums~\cite{han2017approximating,ubaru-2017-fast-estim}.

Quantum computing offers a unique opportunity for faster linear algebra
computations due to its ability to represent and manipulate vectors
in a superposition state. This, combined with the power of entanglement,
opens doors for potential speedups compared to classical matrix algorithms.
Recent advancements in this field have explored various quantum techniques
for fundamental linear algebra operations, including prinicipal component
analysis~\cite{biamonte2017quantum}, matrix multiplication~\cite{shao2018quantum},
solving linear systems of equations~\cite{harrow2009quantum}\textcolor{black}{,
and others~\cite{childs2017quantum,gilyen2019quantum,subacsi2019quantum,an2022quantum,lin2020optimal}.}
Using quantum algorithms for trace estimation and spectral sum approximation
has been considered as well.

Early work on quantum trace estimation trace their origin to the one
clean qubit model of computations~\cite{shepherd-2006-comput-with}.
Given a circuit that implements a unitary operator $\matA$, access
to a maximally entangled state, and one clean pure ancilla qubit,
they show how to estimate $\frac{1}{N}\Trace{\matA}$, where $N$
is the size of the matrix, as a state overlap measured via Hadamard
Test. However, due to the $N^{-1}$ coefficient, the statistical error
is much larger than the one commonly observed by Hutchinson's method.
Another early work on quantum trace estimation is \cite{ekert2002direct},
which introduced a direct method for multivariate trace estimation
of multiple density matrices using a controlled-SWAP gate. However
their circuit design is impractical for near-term devices. More recently,
three preprints suggested algorithms related to multivariate trace
estimation or spectral sum estimation. Luongo and Shao suggested several
quantum algorithms for spectral sum estimation~\cite{luongo2020quantum}.
Their algorithms are based on stringent quantum access assumptions,
such as block-encoding and QRAM. Such assumptions pose significant
challenges~\cite{di2020fault,phalak2023quantum}, e.g. implementing
QRAMs might require a vast number of qubits (trillions for 8GB of
RAM), and might even be infeasible. Shen et al. propose a quantum
algorithm that essentially implements Hutchison's estimator in the
quantum domain~\cite{shen2024efficient}.\textcolor{red}{{} }However,
since creating random vector samples from the Radamacher distribution
is not efficient in the quantum domain, they show how to construct
a random state which they call ``quantum Hutchinson state''. The
distribution of state vector of the quantum Hutchinson state entertains
statistical properties on par with the Radamacher vector when it comes
to stochastic trace estimation. Bravyi et al. consider approximating
partition functions, i.e. $\Trace{e^{-\beta\mat H}}$ for some inverse
temperature $\beta$ and Hamiltonian $\mat H$~\cite{bravyi-2022-quant-hamil}.
They consider both classical and quantum algorithms for local and
non-local Hamiltonians. Finally, Quek et al. suggested a constant
quantum depth circuit to estimate multivariate traces of density operators~\cite{quek2024multivariate}.
Their method assume that inputs are presented as copies of quantum
states, and it computes the multivariate trace of the density operators
associated with those states.

As explained in the previous section, our algorithm is based on a
framework (qMSLA) for operating in quantum domain on matrices via
a set of primitive operations. This approach is inspired by the well-known
BLAS library in numerical linear algebra~\cite{lawson1979basic}.
Other efforts in this vain have been proposed in the quantum domain.
Although not labeled as ``BLAS'', Gilyen et al. show how to do basic
matrix arithmetic, including matrix addition, subtraction and multiplication,
using block encodings~\cite{gilyen2019quantum}, and use these techniques
to propose their seminal QSVT algorithm. In a review paper by Biamonte
et al., the authors catalogued functionalities crucial for quantum
machine learning algorithms under the umbrella term ``qBLAS''~\cite{biamonte2017quantum}.
These functionalities, which include fast Fourier transforms and eigenvalue/eigenvector
computations, are actually present (in the classical domain) in LAPACK
and not BLAS implementations. Another effort is a Q\#-specific library
called QBLAS~\cite{cuiqblas}. QBLAS assumes the presence of QRAM
hardware and implements a simulated version. QBLAS offers functionalities
like vector inner products, HHL linear solver, matrix eigenvalue decomposition,
and quantum phase estimation. Despite the name referencing BLAS, the
QBLAS library, like qBLAS, focuses on higher-order linear algebra
operations typically found in LAPACK and not in BLAS implementations.

\section{Preliminaries}

In this section, we present the notations, definitions, and background
information necessary for describing our algorithms and their analyses.

\subsection{Linear algebra notation}

We denote scalars using Greek letters or using $x,y,\dots$. Vectors
are denoted by $\x,\y,\dots$ and matrices by $\matA,\mat B,\dots$.
The $s\times s$ identity matrix is denoted $\matI_{s}$. We assume
$0$-based indexing for vectors and matrices. This is less common
in the numerical linear algebra literature, but is more convenient
in the context of quantum computing. We use the convention that vectors
are column vectors, unless otherwise stated. For a vector $\x$ or
a matrix $\mat A$, the notation $\x^{\conj}$ or $\mat A^{\conj}$
denotes the Hermitian conjugate. We say that a matrix $\matA$ is
normalized if $\FNorm{\matA}=1$. The vectorization of $\matA\in\C^{m\times n}$,
denoted by $\vec{\matA}\in\C^{mn}$, is a column vector obtained by
stacking the columns of the matrix A on top of one another (i.e. column-major
vectorization): $\vec{\matA}=[a_{0,0},\dots,a_{m-1,0},a_{01},\dots,a_{m-1,1},\dots,a_{m-1,n-1}]^{\T}$.

Given matrices $\matA\in\C^{m\times n}$ and $\matB\in\C^{p\times q}$
, their direct sum, denoted by $\matA\oplus\matB$, is a $(m+p)\times(n+q)$
matrix obtained by putting $\matA$ and $\matB$ on the diagonal.
That is,
\[
\matA\oplus\matB\coloneqq\left[\begin{array}{cc}
\matA & 0_{m\times q}\\
0_{p\times n} & \matB
\end{array}\right]
\]
Allowing for matrices in which one of the dimensions has size zero,
we can use the $\oplus$ notation to also denote padding by columns
or rows:
\[
\matA\oplus\matZero_{p\times0}\coloneqq\left[\begin{array}{c}
\matA\\
\matZero_{p\times n}
\end{array}\right]\quad\quad\matA\oplus\matZero_{0\times q}\coloneqq\left[\begin{array}{cc}
\matA & \matZero_{m\times q}\end{array}\right]
\]
Although we can view vectors as matrices in which either the column
dimension or the row dimension is $1$, we introduce specialized $\oplus$
notations for vectors:
\[
\x\oplus\y\coloneqq\left[\begin{array}{c}
\x\\
\y
\end{array}\right]\quad\quad\x^{\T}\oplus\y^{\T}\coloneqq\left[\begin{array}{cc}
\x^{\T} & \y^{\T}\end{array}\right]
\]
We rely on context to resolve any ambiguity in use of $\oplus$.

Given matrices $\matA\in\C^{m\times n}$ and $\matB\in\C^{p\times q}$,
their Kronecker product $\matA\otimes\matB$ is a $(mp)\times(nq)$
matrix with elements defined by $a_{i,j}b_{k,l}$. A useful representation
is:
\begin{align*}
\matA\otimes\matB & =\sum_{i=0}^{m-1}\sum_{j=0}^{n-1}\sum_{k=0}^{p-1}\sum_{l=0}^{q-1}a_{i,j}b_{k,l}\mat E_{pi+k,qj+l}^{mp\times nq}
\end{align*}
where $\mat E_{i,j}^{m\times n}$ denotes $m\times n$ matrix with
1 in the $(i,j)$th entry and 0 otherwise.

\subsection{\label{subsec:quantum}Matrices and quantum states: Basic notations
and definitions}

We use the state vector based formulation of quantum computing, i.e.
the system's state is represented by a unit vector in an Hilbert space.
We assume without loss of generality that the Hilbert space is $\C^{n}$,
where $n=2^{q}$ for $q$ qubits. We let $\Ket 0_{q},\Ket 1_{q},\Ket 2_{q},\dots$
denote the computational basis in a $q$-qubit system, and $\Ket{\phi}_{q},\Ket{\psi}_{q},\dots$
some abstract states in a $q$-qubit system. Given two states $\Ket{\phi}_{q}$
and $\Ket{\psi}_{p}$ on $q$ and $p$ qubits (respectively), we denote
by $\Ket{\psi}_{q}\Ket{\phi}_{p}$ the $q+p$ qubit state obtained
by tensoring the two states to obtain a state in $\C^{mn}$ where
$n=2^{q}$ and $m=2^{p}$. Similar to many physics textbooks, we assume
that the MSB is in the lowest index in binary expansions, e.g. $\Ket i_{q}=\Ket{b_{0}}_{1}\Ket{b_{1}}_{1}\cdots\Ket{b_{q-1}}_{1}\eqqcolon\Ket{b_{0}\cdots b_{q-1}}_{q}$
where $b_{0}\cdots b_{q-1}$ is the binary expansion of $i$.

Given an amplitude vector $\valpha\in\C^{n}$ , where $n$ is a power
of 2, we use the ket $\Ket{\valpha}$ to denote the $\log_{2}n$-qubit
system's state whose amplitudes are given by $\valpha$ after normalization.
That is, $\Ket{\valpha}\coloneqq\frac{1}{\TNorm{\valpha}}\sum_{i=0}^{n-1}\alpha_{i}\Ket i_{\log_{2}n}$
where $\alpha_{i}$ is the $i$th entry of $\valpha$. Note that under
these conventions we have $\Ket i_{\log_{2}n}=\Ket{\e_{i}^{n}}$ where
$\e_{0}^{n},\e_{1}^{n},\dots$ are the $n$-dimensional identity vectors
(when the dimension of the vectors is clear from the context we omit
it). We also define $\Ket{\overline{\valpha}}\coloneqq\frac{1}{\TNorm{\overline{\valpha}}}\sum_{i}\overline{\alpha_{i}}\Ket i_{\log_{2}n}$,
where $\overline{\alpha_{i}}$ is the complex conjugation of $\alpha_{i}\in\C$.
Given two amplitude vectors $\valpha\in\C^{n}$ and $\vbeta\in\C^{m}$,
we have $\Ket{\valpha}\Ket{\vbeta}=\Ket{\valpha\otimes\vbeta}$ and
this extends naturally to higher number of multiplicands.  

We use calligraphic letters to denote both quantum circuits and operators
on the state's Hilbert space, e.g. ${\cal U}$, ${\cal S}$ and $\calT$,
also using the same letter for a circuit and the operator it induces.
For a circuit ${\cal U}$ on $q$ qubits, we use $\matM({\cal U})\in\C^{n\times n}$
($n=2^{q}$), to denote the unique unitary matrix such that for every
$\valpha\in\C^{n}$ applying ${\cal U}$ on the state $\Ket{\valpha}$
results in the state $\Ket{\matM({\cal U})\valpha}$. We say that
a circuit ${\cal U}^{\conj}$ is the inverse (or adjoint) of circuit
${\cal U}$ if $\matM({\cal U})=\matM({\cal U})^{\conj}$. Given
a quantum system, we denote the application of circuits sequentially
with $\cdot$ or omit totally.

A quantum register is defined as contiguous subset of qubits of a
multi-qubit system. The size of the register is determined by the
number of qubits it encompasses. Given a system with several quantum
registers, we use the standard tensor product $\otimes$ operator
to denote the application of each circuit on the corresponding register,
but also use $\cdotp_{i}$ to denote the action of a circuit ${\cal U}$
to only register $i$ of the system, e.g. ${\cal U}\cdotp_{i}\Ket{\valpha}$
(where the operation on the other register is the identity). We use
$0$-based indexing for register numbering within a system.

Given a system with two quantum registers, with qubit sizes $q^{(0)}=\log_{2}n$
and $q^{(1)}=\log_{2}m$, and a matrix $\matA\in\C^{m\times n}$,
we use $\kket{\matA}$ to denote the state whose amplitudes vector
are $\frac{\vec{\matA}}{\FNorm{\matA}}$. That is,
\[
\kket{\matA}\coloneqq\frac{1}{\FNorm{\matA}}\sum_{j=0}^{n-1}\sum_{i=0}^{m-1}a_{i,j}\Ket{mj+i}_{q^{(0)}+q^{(1)}}=\frac{1}{\FNorm{\matA}}\sum_{i=0}^{m-1}\sum_{j=0}^{n-1}a_{i,j}\kket{\mat E_{i,j}^{m\times n}}
\]
where $a_{ij}$ denotes the $(i,j)$th entry of $\matA$. Since $\mat E_{i,j}^{m\times n}=\e_{i}^{m}(\e_{j}^{n})^{\T}$
we have $\kket{\mat E_{i,j}^{m\times n}}\coloneqq\Ket{\e_{j}^{n}}\Ket{\e_{i}^{m}}=\Ket j_{q^{(0)}}\Ket i_{q^{(1)}}$.
We use the notation $\kket{\matA}$ (instead of $\Ket{\vec{\matA}}$)
to help the reader distinguish state descriptions that are based on
matrices as opposed to vectors. We remark that the use of double kets
to denote matrix states can be traced to \cite{d2001quantum}.

For a matrix $\matA\in\C^{m\times n}$ we denote by $q_{\matA}$ the
number of qubits we need to hold the state $\kket{\matA}$. That is,
$q_{\matA}=\log_{2}mn$. For a circuit ${\cal U}$, we denote $g_{{\cal U}}$
the number of gates in ${\cal U},$ and by $d_{{\cal U}}$ the depth
(critical path) of ${\cal U}$, and by $q(\mathcal{U})$ the number
of qubits of $\mathcal{U}$.

\subsection{Overlap between two quantum states}

Two foundational quantum algorithms for measuring the overlap (i.e.,
dot product or similarity) between two quantum states are the Hadamard
and Swap Tests. For completeness, we now recall both.

For our purposes, it is useful to describe the Hadamard Test as an
algorithm that accepts a description of a circuit ${\cal U}$ and
construct a new circuit that implements the test.
\begin{prop}
[Hadamard Test] Given a classical description of the circuit ${\cal U}$
on $q$ qubit, there exist an algorithm that constructs a circuit
$\HT_{{\cal U}}$ on $q+1$ qubits (see Figure~\ref{fig:hadamard-test-circuit}
for the circuit) such that for any $q$-qubit state $\Ket{\nu}_{q}$
we have 
\begin{equation}
\HT_{{\cal U}}\left(\Ket 0_{1}\Ket{\nu}_{q}\right)=\frac{1}{2}\Ket 0_{1}(\Ket{\nu}_{q}+{\cal U}\Ket{\nu}_{q})+\frac{1}{2}\Ket 1_{1}(\Ket{\nu}_{q}-{\cal U}\Ket{\nu}_{q})\label{eq:HT}
\end{equation}
 The circuit $\HT_{{\cal U}}$ serves as an implementation of the
Hadamard Test. Upon applying the Hadamard Test to $\Ket 0_{1}\Ket{\nu}_{q}$,
the probability of obtaining $0$ after measurement of the first qubit
is 
\[
\frac{1}{2}(1+\Re{\Bra{\nu}{\cal U}\Ket{\nu}}.
\]
The cost of constructing $\HT_{{\cal U}}$ is $O(g_{{\cal U}})$ and
the depth is $d_{{\cal U}}+2$. To obtain the imaginary part, the
phase gate $\calS=\left[\begin{array}{cc}
1 & 0\\
0 & i
\end{array}\right]$ can be employed. In this complex version, we simply apply $\calS^{\dagger}$
after the initial Hadamard gate, and the probability of obtaining
$0$ after measurement of the first qubit is 
\[
\frac{1}{2}\left(1+\Img{\Bra{\nu}{\cal U}\Ket{\nu}}\right)
\]
\end{prop}

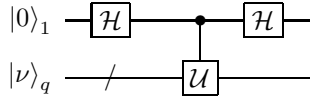
\begin{figure}[tph]
\begin{centering}
\begin{tabular}{c}
\Qcircuit @C=1em @R=.9em  {    &\lstick{\Ket 0_{1}}	 & \gate{{\cal H}}			& 	\qw  &  \ctrl{1}	     & \gate{{\cal H}}		& \qw & \rstick{} 	\\ &\lstick{\Ket{\nu}_{q}} &  {/} \qw  & \qw	       & \gate{{\cal U}}		  & \qw					& \qw	     } \tabularnewline
\end{tabular}
\par\end{centering}
\centering{}\caption{\label{fig:hadamard-test-circuit}Circuit Implementing of Hadamard
test}
\end{figure}

\begin{rem}
If we have access to a circuit ${\cal U}_{\nu}$ such that ${\cal U}_{\nu}\Ket 0_{q}=\Ket{\nu}_{q}$,
and a circuit ${\cal U_{\eta}}$ such that ${\cal U}_{\eta}\Ket 0_{q}=\Ket{\eta}_{q}$,
then one can use Hadamard Test for overlap estimation (with relative
phase information as well) by the following identity: $\Braket{\nu}{\eta}=\Bra 0{\cal U}_{\nu}^{\conj}{\cal U}_{\eta}\Ket 0$,
e.g. in that case we have that 
\begin{equation}
p(0)=\frac{1}{2}\left(1+\Re{\Bra 0{\cal U}_{\nu}^{\conj}{\cal U}_{\eta}\Ket 0}\right)=\frac{1}{2}\left(1+\Re{\Braket{\nu}{\eta}}\right)\label{eq:hadamard-overlap}
\end{equation}
for the output of $\HT_{{\cal U}_{\nu}^{\conj}{\cal U}_{\eta}}$,
where $p(0)$ denotes the probability of measuring $0$ in the first
qubit (when only that qubit is measured).
\end{rem}

\begin{prop}
[Swap Test] For $q$ qubits, there exist an a circuit $\ST_{q}$
on $2q+1$ qubits (see Figure~\ref{fig:swap-test-circuit} for the
circuit) such that for any two $q$-qubit states $\Ket{\psi}_{q}$
and $\Ket{\phi}_{q}$ we have 
\[
\ST_{q}\left(\Ket 0_{1}\Ket{\psi}_{q}\Ket{\phi}_{q}\right)=\frac{1}{2}\Ket 0_{1}\left(\Ket{\psi}_{q}\Ket{\phi}_{q}+\Ket{\phi}_{q}\Ket{\psi}_{q}\right)+\frac{1}{2}\Ket 1_{1}\left(\Ket{\psi}_{q}\Ket{\phi}_{q}-\Ket{\phi}_{q}\Ket{\psi}_{q}\right)
\]
 The circuit $\ST_{q}$ serves as an implementation of the Swap Test.
Upon applying the Swap Test to $\Ket 0_{1}\Ket{\psi}_{q}\Ket{\phi}_{q}$,
the probability of obtaining $0$ after measurement of the first qubit
is 
\[
\frac{1}{2}(1+\abs{\Braket{\phi}{\psi}}^{2}).
\]
The cost of constructing $\ST_{q}$ is $q+2$ and the depth is $3$\textcolor{red}{.}
\end{prop}

\begin{figure}[tph]
\begin{centering}
\begin{tabular}{c}
\Qcircuit @C=1em @R=.9em  {    &\lstick{\Ket{0}_{1}}	 & \gate{{\cal H}}			& 	\qw  &  \ctrl{2}	  & \gate{{\cal H}}		& \qw & \rstick{} 	\\ &\lstick{\Ket{\psi}_{q}} &  \qw  & \qw	 & \qswap		  & \qw					& \qw	\\ &\lstick{\Ket{\phi}_{q}} & \qw   & \qw	 & \qswap \qwx	  & \qw					& \qw  } \tabularnewline
\end{tabular}
\par\end{centering}
\centering{}\caption{\label{fig:swap-test-circuit}Circuit Implementing of Swap Test}
\end{figure}
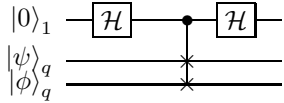

\section{\label{sec:qmsla}quantum Matrix State Linear Algebra (qMSLA)}

Since quantum circuits essentially evolve quantum states by applying
unitary transformations on them, one can envision performing matrix
computations using quantum circuits. Quantum speedups on numerical
linear algebra tasks will immediately lead to faster algorithms in
many scientific computing and machine learning tasks, since at the
core, such algorithms rely heavily on matrix computations~\cite{biamonte2017quantum}.
\textcolor{black}{Indeed, quite a few quantum algorithms for scientific
computing and machine learning have been proposed recently, many of
them relying on QLA breakthroughs such as the HHL algorithm~\cite{harrow2009quantum},
and others~\cite{childs2017quantum,gilyen2019quantum,subacsi2019quantum,an2022quantum,lin2020optimal}.
}\textcolor{red}{}

All the aforementioned works make assumptions on the mechanism in
which the input is presented to the quantum algorithm, and that mechanism
is of crucial importance \cite{chakraborty2018power}. We refer to
the mechanism in which input is presented to the quantum algorithm
as the \emph{input model} of the quantum algorithm. Several input
models have been presented in the literature. Two prominent examples
are the sparse-data access model~\cite{ambainis2012variable,childs2017quantum}
and various quantum data structure based input models~\cite{kerenidis2016quantum,kerenidis2020quantum}.
Recently, Chakraborty et al.~\cite{chakraborty2018power} showed
that a variety of widely used input models can be reduced to an input
model in which matrices are inputed using \emph{block encodings }and
vectors are inputed as \emph{state preparation circuits}: \textcolor{red}{}
\begin{defn}
[State preparation Circuit]\label{def:state-prep-circ} We say that
a $\log_{2}n$-qubit circuit ${\cal U}$ is a \emph{state preparation
circuit} for a vector $\x\in\C^{n}$ if applying ${\cal U}$ to the
state $\Ket 0_{\log_{2}n}$ results in the state $\Ket{\x}$.
\end{defn}

\begin{defn}
[Block encoding of a matrix \cite{gilyen2019quantum}] For $\alpha\geq\FNorm{\matA}$,
a circuit ${\cal U}$ is a $\alpha$\emph{-Block Encoding} of $\matA\in\C^{m\times n}$
if 
\[
\alpha\matM({\cal U})=\left[\begin{array}{cc}
\matA & *\\*
* & *
\end{array}\right]
\]
where $*$ denotes arbitrary entries. We refer to $\alpha$ as the
\emph{scale}.
\end{defn}

We refer to the input model in which matrices are accessed using block
encodings and vectors are accessed as state preparation circuits as
the \emph{block encoding input model}. There are powerful algorithms
that operate under the block encoding model. In particular, in the
block encoding model we can perform Quantum Singular Value Transformation
(QSVT)~\cite{gilyen2019quantum}, a powerful technique that leads
to efficient algorithms for solving linear equations, amplitude amplification,
quantum simulation, and more~\cite{martyn2021grand}.

In this paper we do not assume matrices are given as block encodings,
and instead consider an input model in which \emph{both} matrices
and vectors are inputed as state preparation circuits; for matrices
we use a \emph{matrix state preparation circuit.}
\begin{defn}
[Matrix state preparation circuit] \label{def:matrix-circ}We say
that a $(\log_{2}n+\log_{2}m)$-qubit circuit ${\cal U}$ is a \emph{matrix
state preparation circuit} for a matrix $\matA\in\C^{m\times n}$
if applying ${\cal U}$ to the state $\Ket 0_{\log_{2}mn}$ results
in the state $\kket{\matA}$. Equivalently, the first column of $\matM({\cal U})$
is $\vec{\matA}$. For convenience, where appropriate, we add the
matrix as sub-index when denoting state preparation circuits, e.g.
${\cal U}_{\matA}$. In such cases, with an abuse of notation, the
number of gates in ${\cal U}_{\matA}$ is denoted by $g_{\matA}$,
and the depth by $d_{\matA}$.
\end{defn}

\begin{rem}
A matrix state preparation circuit ${\cal U}_{\matA}$ should be viewed
as a data structure bundling the actual circuit, along with meta-data
that describe the size of the resulting matrix (i.e., $m$ and $n$).
The meta-data in essence describes the partition of the qubits on
which the circuit operates into two registers, one corresponding to
columns indices (MSB qubits, ``top'' qubits in circuit visualizations)
and the other corresponding to row indices (LSB qubits, ``bottom''
qubits in circuit visualizations). We refer to the qubits that correspond
to the column indices as \emph{column qubits}, and the register that
corresponds to these indices as the \emph{column register}. Likewise
for \emph{row qubits} and \emph{row register}.
\end{rem}

\begin{rem}
It is also possible to track classically the Frobinieus norm of the
matrix (and include it in the meta-data), though this is not necessary
for the purpose of multivariate trace estimation.
\end{rem}

\begin{rem}
A vector is also matrix where one of the dimensions has size $1$.
So, our definition of matrix state preparation circuit (Definition~\ref{def:matrix-circ})
subsumes vector state preparation circuits (Definition~\ref{def:state-prep-circ}).
In vector state preparation circuits one of the two registers has
zero size: for a column vector we have only row qubits, and for a
row vector we have only column qubits.

We refer to the input model in which both matrices and vectors are
accessed via state preparation circuits as the \emph{state preparation
input model}. Are the block encoding model and the state preparation
model equivalent?\emph{ }We do not have a complete answer to this
question, though we do have some partial observations. First, given
a matrix state preparation circuit for $\matA$ and a state preparation
circuit for a vector $\w$ whose entries are the row norms of $\matA$,
it is possible to construct a block encoding of $\matA$~~\cite[Section I.D]{clader2022quantum}.
We are unaware of any efficient algorithm that given \emph{only} a
matrix state preparation circuit for $\matA$ constructs a block encoding
of $\matA$. Second, in Section~\ref{subsec:mu-state-preparation}
we show that given\textcolor{red}{{} }a circuit ${\cal U}$ we can
construct a matrix state preparation of $\matM({\cal U})$. Thus,
given a block encoding of $\matA$ we can construct a matrix state
preparation circuit for a matrix that contains $\matA$. Both observations
together lead us to conjecture that the state preparation model uses
less stringent assumptions\textcolor{red}{{} }than the block encoding
model, though we leave formalizing this conjecture and proving its
correctness to future work.

There is one additional way in which the state preparation model has
an advantage over the block encoding model. The value of the scale
parameter $\alpha$ of the block encoding is important to downstream
efficiency. Roughly speaking, the smaller $\alpha$ is, shallower
circuits suffice or less shots are required. Recalling that the minimum
value for $\alpha$ for a block encoding of $\matA$ is $\FNorm{\matA}$,
we note that it is impossible to construct a block encoding with the
minimum $\alpha$ unless $\matA$ is unitary. In contrast, for \emph{any}
$m\times n$ matrix it is possible to construct a $O(mn)$ depth matrix
state preparation circuit in $O(mn)$ (classical) time, without any
loss~\cite{shende2005synthesis}.\textcolor{red}{}

In this paper we show how to estimate multivariate traces and spectral
sums in the state preparation model. The main challenge is that in
the state preparation model we do not have access to powerful tools
like QSVT. Our approach is based on showing that given classical descriptions
of state preparation circuits, we can construct more complex state
preparation circuits that implement several important matrix algebra
operations, thus providing a toolbox for performing matrix algebra
in the state preparation model. We refer to our framework as ``quantum
Matrix State Linear Algebra'' (qMSLA). Using this toolbox, we can
encode multivariate traces in state preparation circuit via simple
matrix identities. In this section we outline qMSLA, while in the
next section (Section~\ref{sec:Encoding-Matrix-Moments}) we use
qMSLA to propose a a quantum algorithm for multivariate trace estimation.
\end{rem}

In the following subsections we present various qMSLA operations,
along with discussion of their complexity. With few exceptions, each
operations receives a variable number of classical descriptions of
matrix state preparation circuit, and outputs a classical description
of a new state preparation circuit that implements some matrix algebra
operation between the input matrices. Thus, the algorithms are classical-to-classical,
but provide circuits to be executed on a quantum computer. For each
operation, we provide visualization of the output circuit, analyze
the gate complexity and depth of the output circuit, and the (classical)
cost of forming the output circuit.

We split qMSLA to two levels. Level 1 includes primitive operations
that create or operate on state preparation circuits using low level
circuit building procedures, while level 2 includes composite operations
built using level 1 primitive operations. Both levels are summarized
in Tables~\ref{tab:qmsla-table-l1} and \ref{tab:qmsla-composite-l2}.

In this section we present a high level description on how to implement
various qMSLA operations. Low-level implementation details are deferred
to the appendix. We also introduce in the appendix another qMSLA level:
level 0. Level 0 operates on the circuit level, and does not deal
with matrix states at all. In essence, these are the basic circuit
building operations we require from an underlying quantum computing
framework (such as QISKIT) so that we can build qMSLA on top of that
framework. Most of the operations are implemented in various frameworks,
but some are not, so we explain how they can be implemented using
only the ability to add and remove gates. We then show how to implement
level 1 operations using only level 0 operations.

Henceforth in this section, unless otherwise specified, we assume
that $\matA\in\C^{m\times n}$, ${\cal U},{\cal W},{\cal Q}$ are
quantum circuits, $\sigma$ is a permutation function on all qubits,
${\cal S}_{\sigma}$ denotes an implementation of a permutation $\sigma$
using SWAP gates, and $\vxi\in\C^{m(n-1)}$ represents auxiliary information
(garbage), which may be arbitrary.

\begin{table}[tph]
\centering{}{\footnotesize{}}%
\begin{tabular}{|c|c|c|c|c|c|c|}
\hline 
{\footnotesize{}Input} & {\footnotesize{}Output} & {\footnotesize{}$q(\text{output})$} & {\footnotesize{}$g(\text{output})$} & {\footnotesize{}$d(\text{output})$} & {\footnotesize{}Operation Name} & {\footnotesize{}Subsection}\tabularnewline
\hline 
\hline 
{\footnotesize{}$n=2^{q}$} & {\footnotesize{}${\cal U}_{\matI_{n}}$} & {\footnotesize{}$2q$} & {\footnotesize{}$2q$} & {\footnotesize{}$2$} & {\footnotesize{}$\textrm{qmsla.identity}$} & {\footnotesize{}\ref{subsec:identity}}\tabularnewline
\hline 
{\footnotesize{}${\cal U}$} & {\footnotesize{}${\cal U}_{\matM({\cal U})}$} & {\footnotesize{}$2q({\cal U})$} & {\footnotesize{}$g({\cal U})+2q$} & {\footnotesize{}$d({\cal U})+2$} & {\footnotesize{}$\textrm{qmsla.matrix}$} & {\footnotesize{}\ref{subsec:mu-state-preparation}}\tabularnewline
\hline 
{\footnotesize{}${\cal U}_{\matA}$} & {\footnotesize{}${\cal U}_{\overline{\matA}}$} & {\footnotesize{}$q_{\matA}$} & {\footnotesize{}$g_{\matA}$} & {\footnotesize{}$d_{\matA}$} & {\footnotesize{}$\textrm{qmsla.conjugate}$} & {\footnotesize{}\ref{subsec:transpose-conjugate-inverse}}\tabularnewline
\hline 
{\footnotesize{}${\cal U}_{\matA}$} & {\footnotesize{}${\cal U}_{\matA^{\T}}$} & {\footnotesize{}$q_{\matA}$} & {\footnotesize{}$g_{\matA}$} & {\footnotesize{}$d_{\matA}$} & {\footnotesize{}$\textrm{qmsla.transpose}$} & {\footnotesize{}\ref{subsec:matrix-transpose}}\tabularnewline
\hline 
{\footnotesize{}${\cal U}_{\matA}$} & {\footnotesize{}${\cal U}_{\vec{\matA}}$} & {\footnotesize{}$q_{\matA}$} & {\footnotesize{}$g_{\matA}$} & {\footnotesize{}$d_{\matA}$} & {\footnotesize{}$\textrm{qmsla.vec}$} & {\footnotesize{}\ref{subsec:vectorize-matrix}}\tabularnewline
\hline 
{\footnotesize{}${\cal U}_{\matA},k$} & {\footnotesize{}${\cal U}_{\matA\oplus0_{0\times(2^{k}-1)n}}$} & {\footnotesize{}$q_{\matA}+k$} & {\footnotesize{}$g_{\matA}$} & {\footnotesize{}$d_{\matA}$} & {\footnotesize{}$\textrm{qmsla.pad\_zero\_columns}$} & {\footnotesize{}\ref{subsec:pad-columns}}\tabularnewline
\hline 
{\footnotesize{}${\cal U_{\matA}},{\cal U}_{\b}$} & {\footnotesize{}${\cal U}_{\frac{\matA\b}{\FNorm{\matA}\TNorm{\b}}\oplus\vxi}$} & {\footnotesize{}$q_{\matA}$} & {\footnotesize{}$g_{\matA}+g_{\matB}$} & {\footnotesize{}$d_{\matA}+d_{\matB}$} & {\footnotesize{}$\textrm{qmsla.matrix\_vec}$} & {\footnotesize{}\ref{subsec:matrix-vector-product}}\tabularnewline
\hline 
{\footnotesize{}${\cal U}_{\matA},{\cal U}_{\matB}$} & {\footnotesize{}${\cal U}_{\matA\otimes\matB}$} & {\footnotesize{}$q_{\matA}+q_{\matB}$} & {\footnotesize{}$g_{\matA}$+$g_{\matB}$} & {\footnotesize{}$\max(d_{\matA},d_{\matB})$} & {\footnotesize{}$\textrm{qmsla.kronecker}$} & {\footnotesize{}\ref{subsec:kronecker-product}}\tabularnewline
\hline 
{\footnotesize{}${\cal U}_{\matA_{1}},\dots,{\cal U}_{\matA_{k}}$} & {\footnotesize{}${\cal U}_{\matA_{1}\otimes\matA_{2}\otimes\cdots\otimes\matA_{k}}$} & {\footnotesize{}$\sum_{i}q_{\matA_{i}}$} & {\footnotesize{}$\sum_{i}g_{\matA_{i}}$} & {\footnotesize{}$\max(d_{\matA_{1}},\cdots,d_{\matA_{k}})$} & {\footnotesize{}$\textrm{qmsla.kronecker}$} & {\footnotesize{}\ref{subsec:multiple-kronecker}}\tabularnewline
\hline 
\end{tabular}\caption{\label{tab:qmsla-table-l1}Level 1 qMSLA operations. These are primitive
operations that create or operate on state preparation circuits using
low level circuit building tools. The classical cost for these operations
is consistently $O(g(output))$.}
\end{table}

\begin{table}[tph]
\begin{centering}
{\footnotesize{}}%
\begin{tabular}{|c|c|c|c|c|c|c|}
\hline 
{\footnotesize{}Input} & {\footnotesize{}Output} & {\footnotesize{}$q(\text{output})$} & {\footnotesize{}$g(\text{output})$} & {\footnotesize{}$d(\text{output})$} & {\footnotesize{}Operation Name} & {\footnotesize{}Subsection}\tabularnewline
\hline 
\hline 
{\footnotesize{}${\cal U}_{\matA},r$} & {\footnotesize{}${\cal U}_{\matA\oplus0_{(2^{r}-1)m\times0}}$} & {\footnotesize{}$q_{\matA}+r$} & {\footnotesize{}$g_{\matA}$} & {\footnotesize{}$d_{\matA}$} & {\footnotesize{}$\textrm{qmsla.pad\_zero\_rows}$} & {\footnotesize{}\ref{subsec:pad-rows}}\tabularnewline
\hline 
{\footnotesize{}${\cal U}_{\matA},k,r$} & {\footnotesize{}${\cal U}_{\matA\oplus0_{(2^{r}-1)m\times(2^{k}-1)n}}$} & {\footnotesize{}$q_{\matA}+k+r$} & {\footnotesize{}$g_{\matA}$} & {\footnotesize{}$d_{\matA}$} & {\footnotesize{}$\textrm{qmsla.pad}$} & {\footnotesize{}\ref{subsec:pad-rows}}\tabularnewline
\hline 
{\footnotesize{}${\cal U}_{\matA}$} & {\footnotesize{}${\cal U}_{\matA^{\conj}}$} & {\footnotesize{}$q_{\matA}$} & {\footnotesize{}$g_{\matA}$} & {\footnotesize{}$d_{\matA}$} & {\footnotesize{}$\textrm{qmsla.adjoint}$} & {\footnotesize{}\ref{subsec:adjoint}}\tabularnewline
\hline 
{\footnotesize{}${\cal U}_{\vpsi},{\cal U}_{\vphi}$} & {\footnotesize{}${\cal U}_{\frac{\vpsi^{\conj}\vphi}{\TNorm{\vpsi}\TNorm{\vphi}}\oplus\vxi}$} & {\footnotesize{}$q_{\vpsi}$} & {\footnotesize{}$g_{\vpsi}+g_{\vphi}$} & {\footnotesize{}$d_{\vpsi}+d_{\vphi}$} & {\footnotesize{}$\textrm{qmsla.overlap}$} & {\footnotesize{}\ref{subsec:overlap}}\tabularnewline
\hline 
\end{tabular}{\footnotesize\par}
\par\end{centering}
\caption{\label{tab:qmsla-composite-l2}Level 2 qMSLA operations. These are
composite operations built using level 1 operations. The classical
cost for these operations is consistently $O(g(output))$.}
\end{table}

\subsection{\label{subsec:prelude}Prelude: Operating on a subset of the qubits
and qubit permutations}

In this subsection, we investigate quantum circuit operations targeting
specific subsets of qubits (i.e., registers), qubit permutation procedures,
circuit composition, and an efficient method for eliminating SWAP
gates in state preparation circuits. Understanding each of these operations
is crucial for the implementation of various qMSLA operations.

\subsubsection{Operating on single registers of a matrix state}

Consider a matrix state $\kket{\matA}$ for $\matA\in\C^{m\times n}$.
qMSLA is based on three types of operations: operating on the column
register alone, operating on the row register alone, and rearranging
the qubits between registers (using SWAP gates). In this subsection
we seek to understand how the first two evolve matrix states.

First, lets consider how the state evolves when we apply a circuit
only on one of the two registers. Looking at the product states of
the two registers, assuming that $\valpha\in\C^{n}$ and $\vbeta\in\C^{m}$,
we have the following immediate identities:
\begin{enumerate}
\item ${\cal U}\cdot\Ket{\valpha}\Ket{\vbeta}$ results in the state $\Ket{\matM({\cal U})(\valpha\otimes\vbeta)}$.
\item ${\cal Q}\cdotp_{0}\Ket{\valpha}\Ket{\vbeta}=({\cal Q}\otimes{\cal I}_{m})\Ket{\valpha}\Ket{\vbeta}$
results in the state $\Ket{\matM({\cal Q})\valpha}\Ket{\vbeta}$.
\item ${\cal W}\cdotp_{1}\Ket{\valpha}\Ket{\vbeta}=({\cal I}_{n}\otimes{\cal W})\Ket{\valpha}\Ket{\vbeta}$
results in the state $\Ket{\valpha}\Ket{\matM({\cal W})\vbeta}$.
\end{enumerate}
Using these identities, we prove the following lemma.

\begin{lem}
\label{prop:upper-lower-opperations} Consider the state $\kket{\matA}$
for $\matA\in\C^{m\times n}$, and circuits ${\cal U}$ and ${\cal Q}$
where ${\cal U}$ operates on $\log_{2}n$ qubits and ${\cal Q}$
operates on $\log_{2}m$ qubits. The following holds:
\begin{enumerate}
\item ${\cal U}\cdotp_{0}\kket{\matA}$ results in the state $\kket{\matA\matM({\cal U})^{\T}}$.
\item ${\cal Q}\cdotp_{1}\kket{\matA}$ results in the state $\kket{\matM({\cal Q})\matA}$.
\end{enumerate}
\end{lem}

\begin{proof}
We use the following identity: $\vec{\matA\matB\matC}=(\matC^{\T}\otimes\matA)\vec{\matB})$
\cite[Proposition 7.1.9]{bernstein2009matrix}. We have,

\begin{align*}
{\cal U}\cdotp_{0}\kket{\matA}=({\cal U}\otimes{\cal I})\Ket{\vec{\matA}} & =\Ket{\matM({\cal U}\otimes{\cal I})\vec{\matA}}\\
 & =\Ket{(\matM({\cal U})\otimes\matI)\vec{\matA}}\\
 & =\Ket{\vec{\matI\matA\matM({\cal U})^{\T}}}\\
 & =\kket{\matA\matM({\cal U})^{\T}}
\end{align*}
and,
\begin{align*}
{\cal Q}\cdotp_{1}\kket{\matA}=({\cal I\otimes}{\cal Q})\Ket{\vec{\matA}} & =\Ket{\matM({\cal I}{\cal \otimes}{\cal Q})\vec{\matA}}\\
 & =\Ket{(\matI\otimes\matM({\cal Q}))\vec{\matA}}\\
 & =\Ket{\vec{\matM({\cal Q})\matA\matI}}\\
 & =\kket{\matM({\cal Q})\matA}
\end{align*}
\end{proof}

\subsubsection{Permutation of qubits and registers}

Let $q$ be the number of qubits in a system, and consider a permutation
$\sigma:\{0,\ldots,q-1\}\to\{0,\ldots,q-1\}$. Denote by ${\cal S}_{\sigma}$
be the unitary that permutes the qubits according to $\sigma$. In
the computational basis we have
\[
\calS_{\sigma}\Ket{b_{0}\cdots b_{q-1}}_{q}=\Ket{b_{\sigma(0)}\cdots b_{\sigma(q-1)}}_{q}
\]
where $b_{0}\cdots b_{q-1}$ is the binary expansion of the index
of a basis state. A circuit implementing ${\cal S}_{\sigma}$ can
be built using SWAP gates. A simple algorithm for forming ${\cal S}_{\sigma}$
is the following. Given the sequence $(\sigma(0),\dots,\sigma(q-1))$,
find a sequence of swaps that sort the sequence (e.g., by tracking
the swaps of a sorting algorithm). That is find $(i_{1},j_{1}),\dots,(i_{T},j_{T})$
such that $\sigma_{(i_{T},j_{T})}\circ\cdots\circ\sigma_{(i_{1},j_{1})}\circ\sigma=\sigma_{I}$
where $\sigma_{(i,j)}$ is the permutation that swaps the entry $i$
and $j$, and $\sigma_{I}$ is the identity permutation. Let $\calS_{(i,j)}$
denote the circuit which swaps qubits $i$ and $j$ (via a single
SWAP gate), but keeps other qubits in place. Now, ${\cal S}_{(i_{T},j_{T})}\cdots{\cal S}_{(i_{1},j_{1})}{\cal S}_{\sigma}={\cal I}$.
By applying the inverse of these swap operations from the left, we
obtain a implementation of $\calS_{\sigma}$ as $\calS_{(i_{1},j_{1})}\cdots\calS_{(i_{T},j_{T})}$.
The set of swaps can be found using sorting algorithm in $O(q\log q)$
operations. This is also the bound on the number of gates in ${\cal S}_{\sigma}$,
and on the depth.

However, our algorithms typically permute full registers rather than
individual bits. In such cases, to distinguish register permutations
from qubit permutations, we use the bar notation $\bar{\sigma}$ to
indicate that the permutation is to be applied to registers. It is
assumed to have signature $\bar{\sigma}:\{0,\ldots,r-1\}\to\{0,\ldots,r-1\}$
for $r$ registers. In particular, suppose we have $r$ registers,
of sizes $q_{0},\dots,q_{r-1}$, and consider a permutation $\bar{\sigma}:\{0,\ldots,r-1\}\to\{0,\ldots,r-1\}$,
then define the circuit ${\cal S}_{\bar{\sigma}}^{R}$ as the circuit
on $\sum_{i=0}^{m-1}q_{i}$ qubits for which for all states $\Ket{\psi_{0}}_{q_{0}},\dots,\Ket{\psi_{r-1}}_{q_{r-1}}$we
have 
\[
{\cal S}_{\bar{\sigma}}^{R}\cdot\Ket{\psi_{0}}_{q_{0}}\cdots\Ket{\psi_{r-1}}_{q_{r-1}}=\Ket{\psi_{\bar{\sigma}(0)}}_{q_{\bar{\sigma}(0)}}\cdots\Ket{\psi_{\bar{\sigma}(r-1)}}_{q_{\bar{\sigma}(r-1)}}
\]

\begin{example}
Given $(\log_{2}m+\log_{2}n)$-qubit system and a permutation $\bar{\sigma}_{\T}=(1,0)$
(i.e., swap the order of the two registers) we have that,
\[
\calS_{\bar{\sigma}_{\T}}^{R}\Ket{\phi}_{\log_{2}n}\Ket{\psi}_{\log_{2}m}=\Ket{\psi}_{\log_{2}m}\Ket{\phi}_{\log_{2}n}
\]
We show in Section~\ref{subsec:matrix-transpose} that $\calS_{\bar{\sigma}_{\T}}^{R}$
applied to $\kket{\matA}$ results in $\kket{\matA^{\T}}$.
\end{example}

\subsubsection{Eliminating permutations in state preparation circuits}

In the subsequent sections, various qMSLA operations are described,
and from these descriptions, a common structure emerges in the output
of straightforward implementations. Suppose the output is a state
preparation circuit ${\cal U}_{\matX}$ for some matrix $\matX$.
Typically, ${\cal U}_{\matX}$ starts with a series of SWAP gates
that effectively implements a permutation on the qubits, next some
circuit ${\cal Q}$ is executed, and finally another series of SWAP
operations are executed effectively implementing another permutation
on the qubits. In short, the output can be written as $\mathcal{U}_{\matX}={\cal S}_{\sigma_{1}}{\cal Q}{\cal S}_{\sigma_{2}}$.
However, conceptually, since qMSLA concerns only with the matrix its
outputs prepare (i.e., their operation on the ground state $\Ket 0$),
any circuit ${\cal U}_{\matX}^{'}$ will do equally well as long as
it is a state preparation circuit for $\matX$ as well, even if $\matM({\cal U}_{\matX})\neq\matM({\cal U}_{\matX}^{'})$.
In this subsection, we show how given $\sigma_{1}$, $\sigma_{2}$
and ${\cal Q}$, we can construct a circuit ${\cal U}_{\matX}^{'}$
that prepares the same matrix as $\mathcal{U}_{\matX}={\cal S}_{\sigma_{1}}{\cal Q}{\cal S}_{\sigma_{2}}$
but does not use any SWAPs. We call this process \noun{EliminatePremutations.}

As stated earlier, in the appendix we give low-level implementation
details of qMSLA operations based on the high-level descriptions in
this section. In translating the high-level descriptions of Sections~\ref{subsec:lev1}
and \ref{subsec:lev2} to low-level pseudocode descriptions, we implicitly
use \noun{EliminatePremutations }as a meta-algorithm that transforms
the output circuit to a more efficient ones. The high-level descriptions
of Sections~\ref{subsec:lev1} and \ref{subsec:lev2} do not use
\noun{EliminatePremutations. }However, the circuit diagrams do show
the boundaries on which \noun{EliminatePremutations} is applied, and
various complexity statements assume that SWAPS have been eliminated
using \noun{EliminatePremutations.}

\noun{EliminatePremutations} proceeds as followed. First, it eliminates
any leading SWAP gates, as swaps applied to the ground state have
no effect, i.e. ${\cal S}_{\sigma_{2}}\Ket 0=\Ket 0$. Next\noun{,
EliminatePremutations} calculates the inverse permutation $\sigma_{1}^{-1}$,
inverts the circuit ${\cal Q}$, and composes the empty circuit with
${\cal Q}^{-1}$. However, in the composition, wire $i$ is connected
to wire $\sigma_{1}(i)$ in ${\cal Q}^{-1}$ input. Finally, \noun{EliminatePremutations}
inverts the output circuit, resulting in $\mathcal{U}_{\matX}^{'}$.
 More details, including a proof of correctness, are provided in
the appendix.

\subsection{\label{subsec:lev1}Level-1 qMSLA operations}

Level 1 operations create or operate on state preparation circuits
using low-level circuit building operations. For more details on implementation,
including pseudo-code for all level 1 operations, see Appendix \ref{sec:Pseudo-Code-of-level1}.

\subsubsection{\label{subsec:identity}Preparing an identity matrix:\protect \\
 ${\cal U}_{\protect\matI_{n}}\gets\textrm{qmsla.identity}(n)$}

Given $2q$-qubits, split between two $q$-qubit registers, consider
the maximally entangled state 
\[
\Ket{\psi}_{2q}=\frac{1}{\sqrt{n}}\sum_{k=0}^{n-1}\Ket k_{q}\otimes\Ket k_{q}
\]
where $n=2^{q}$. We have 
\[
\Ket{\psi}_{2q}=\frac{1}{\sqrt{n}}\sum_{k=0}^{n-1}\Ket k_{q}\otimes\Ket k_{q}=\frac{1}{\sqrt{n}}\sum_{n=0}^{n-1}\kket{\e_{k}\e_{k}^{\T}}=\kket{\matI_{n}}
\]
Thus, to implement a state preparation circuit ${\cal U}_{\mat I_{n}}$
for $\matI_{n}$ we need to find a circuit that given $\Ket 0_{2q}$
creates the maximally entangled state $\Ket{\psi}_{2q}$. It is well
known that the circuit depicted in Figure~\ref{fig:circuit-for-psi}
prepares this state. Thus, $g_{\matI_{n}}=2q$ and $d_{\matI_{n}}=2$.

\begin{figure}[tph]
\centering    
\begin{tabular}{c  c  c}
\Qcircuit @C=1em @R=.9em  {    &\lstick{} 	&  \qw 		& \multigate{5}{{\cal U}_{\matI_{n}}}    &   		& 					& &	& \\ &\lstick{} 	&  \qw 		& \ghost{{\cal U}_{\matI_{n}}} 			 & {/} \qw	&  \qw   		\inputgrouph{1}{6}{4.0em}{\Ket 0_{2\log_{2}n}}{1.0em} & \dstick{c:\log_{2}n} & &  \\ &\lstick{} 	&  \qw 		& \ghost{{\cal U}_{\matI_{n}}}  		 &   		& 					& &	& \rstick{\kket{\matI}}  \\ &\lstick{} 	&  \qw 		& \ghost{{\cal U}_{\matI_{n}}}  		 &   		& 					& &	&  \\ &\lstick{} 	&  \vdots 	& \nghost{{\cal U}_{\matI_{n}}} 		 &			{/} \qw & \qw    & \dstick{r:\log_{2}n} & & \\		 &\lstick{} 	&  \qw 		& \ghost{{\cal U}_{\matI_{n}}} 			 & 			 &    & &  &  \\		   } 
& 
\raisebox{-4.0em}{\hspace{1cm}\scalebox{1.5}{\quad = \quad}}    
& 
\hspace{1.25cm}  
\Qcircuit @C=1em @R=.9em  {    &\lstick{} 	&  \qw 	& \gate{\calH}  			&  \ctrl{3} & \qw		& \qw \inputgrouph{1}{6}{3.75em}{\Ket 0_{2\log_{2}n}}{1.0em}	& \qw 		&  \\ &\lstick{} 	&   	& \raisebox{1.5ex}{\vdots}	& 	    	& 			&  \raisebox{1.5ex}{\vdots}	 									& 			&   \rstick{\hspace{-0.5em}c:\log_{2}n}  		\\ &\lstick{} 	&  \qw 	& \gate{\calH} 			 	& \qw 		& \ctrl{3}	& \qw															& \qw 		& 	 \rstick{\raisebox{3.0em}{\hspace{+4em}}\kket{\matI}} 		\\ &\lstick{} 	&  \qw 	& \qw  						& \targ		& \qw 		& \qw															& \qw 		&   											\\ &\lstick{} 	&  		& \raisebox{1.5ex}{\vdots}	& 		 	& 			&  \raisebox{1.5ex}{\vdots}										& 			& \rstick{\hspace{-0.5em}r:\log_{2}n}	\\		 &\lstick{} 	&  \qw 	& \qw	  					& \qw		& \targ 	&  \qw															& \qw 		& 																												 \\		   } 
\end{tabular}

\caption{\label{fig:circuit-for-psi}Graphical description of $\text{qmsla.identity}$
(Input: $n$, Output: ${\cal U}_{\protect\matI_{n}}$).}
\end{figure}
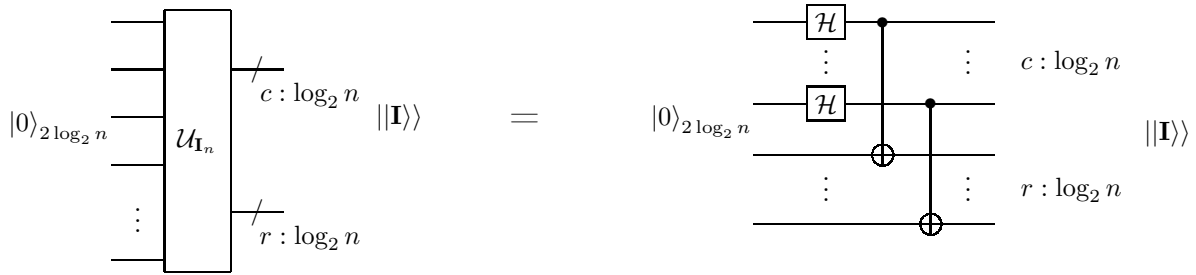

\subsubsection{\label{subsec:mu-state-preparation}State preparation for $\protect\matM({\cal U})$:
\protect \\
${\cal U}_{\protect\matM({\cal U})}\gets\textrm{qmsla.matrix}({\cal U})$}

Consider a circuit ${\cal U}$ on $q$ qubits and denote $n=2^{q}$.
Due to Lemma~\ref{prop:upper-lower-opperations} we have, 
\begin{align*}
{\cal U}\cdotp_{1}{\cal U}_{\matI_{n}}\Ket 0_{2q} & ={\cal U}\cdotp_{1}\kket{\matI_{n}}\\
 & =\kket{\matM({\cal U})\matI_{n}}\\
 & =\kket{\matM({\cal U})}
\end{align*}
Thus, ${\cal U}\cdotp_{1}{\cal U}_{\matI_{n}}$ is a state preparation
circuit for $\matM({\cal U})$. See Figure~\ref{fig:circuit-for-u_mu}
for a graphical description. We have $g_{\matM({\cal U})}=g({\cal U})+2q$
and $d_{\matM({\cal U})}=d({\cal U})+2$. The complexity of forming
${\cal U}_{\matM({\cal U})}$ is $O(g_{\matM({\cal U})})$.

\newcommand{\nmultigate}[2]{*+<1em,.9em>{\hphantom{#2}} \POS [0,0]="i",[0,0].[#1,0]="e",!C *{#2},"e"+UR;"e"+UL **\dir{-};"e"+DL **\dir{-};"e"+DR **\dir{-};"e"+UR **\dir{-},"i" }     

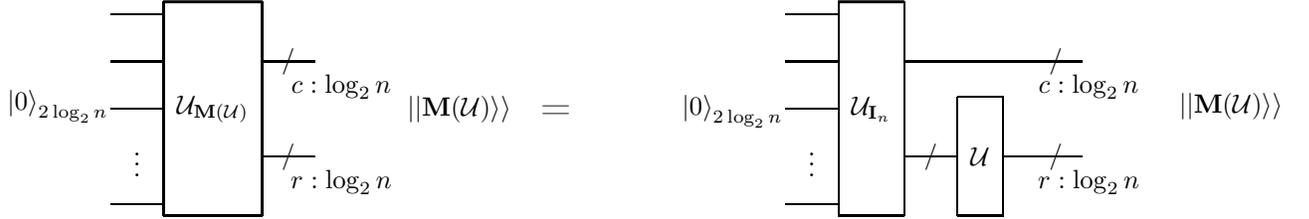
\begin{figure}[tph]
\noindent \begin{centering}
\centering     
\begin{tabular}{c  c  c} 
\Qcircuit @C=1em @R=.9em  {    &\lstick{} 	&  \qw 		& \multigate{4}{{\cal U}_{\matM({\cal U})}}    &   		& 					& &	& \\ &\lstick{} 	&  \qw 		& \ghost{{\cal U}_{\matM({\cal U})}} 		   & {/} \qw	&  \qw   		\inputgrouph{1}{5}{3.5em}{\Ket 0_{2\log_{2}n}}{1.0em} & \dstick{c:\log_{2}n} & &  \\ &\lstick{} 	&  \qw 		& \ghost{{\cal U}_{\matM({\cal U})}}  		   &   		& 					& &	& \rstick{\kket{\matM({\cal U})}}  \\ &\lstick{} 	&  \vdots 	& \nghost{{\cal U}_{\matM({\cal U})}} 		   &			{/} \qw & \qw    & \dstick{r:\log_{2}n} & & \\		 &\lstick{} 	&  \qw 		& \ghost{{\cal U}_{\matM({\cal U})}} 		   & 			 &    & &  &  \\		   } 

& 
\raisebox{-4.0em}{\hspace{1cm}\scalebox{1.5}{\quad = \quad}}     
& 
\hspace{1.25cm}   
\Qcircuit @C=1em @R=.9em  {                                                      &\lstick{} 	&  \qw 		& \multigate{4}{{\cal U}_{\matI_{n}}}   &   						 &   					&   	&			&   & \\ &\lstick{} 	&  \qw 		& \ghost{{\cal U}_{\matI_{n}}} 			& \qw    					 &  \qw       					&  \qw \inputgrouph{1}{5}{3.7em}{\Ket 0_{2\log_{2}n}}{1.0em}	& {/} \qw  	& \dstick{\raisebox{-1.0em}{\hspace{0.5em}}c:\log_{2}n}  \qw  &  \\ &\lstick{} 	&  \qw 		& \ghost{{\cal U}_{\matI_{n}}}  		& 							 &  \nmultigate{2}{{\cal U}} &		&			&   &  \rstick{\raisebox{-0.5em}{\hspace{+2em}}\kket{\matM({\cal U})} }  \\ &\lstick{} 	&  \vdots 	&\nghost{{\cal U}_{\matI_{n}}} 			& {/} \qw					 &  \ghost{{\cal U}}       &  \qw  & {/} \qw 	& \dstick{\raisebox{-1.0em}{\hspace{0.5em}}r:\log_{2}n} \qw  & \\		 &\lstick{} 	&  \qw 		& \ghost{{\cal U}_{\matI_{n}}} 			& 							 &  \nghost{{\cal U}}  		& 		&  			&   &  \\		   }
\end{tabular}
\par\end{centering}
\caption{\label{fig:circuit-for-u_mu}Graphical description of $\text{qmsla.matrix}$
(Input: ${\cal U}$, Output: ${\cal U}_{\protect\cirmat{{\cal U}}}$).}
\end{figure}

In a sense, the ability to build ${\cal U}_{\matM({\cal U})}$ from
a circuit ${\cal U}$ hints that usage of the state preparation model
constitutes making weaker assumptions than the ones made when using
the block encoding model, since block encodings can be converted to
state preparation circuits.

\subsubsection{Matrix conjugate:\protect \\
${\cal U}_{\overline{\protect\matA}}\gets\textrm{qmsla.conjugate}({\cal U}_{\protect\matA})$}

Given a state preparation circuit for a matrix $\matA\in\C^{m\times n}$,
we can construct a state preparation circuit for $\overline{\matA}$
by conjugating the circuit itself. This is shown in the following
proposition.
\begin{prop}
\emph{For a matrix state preparation circuit ${\cal U}_{\matA}$,
we have 
\[
\overline{{\cal U}_{\matA}}={\cal U}_{\overline{\matA}}
\]
Thus, given ${\cal U}_{\matA}$ we can compute a description of ${\cal U}_{\overline{\matA}}$
with depth $d_{\matA}$ in cost $O(g_{\matA})$.}
\end{prop}

\begin{proof}
Recall that the first column of $\matM({\cal U}_{\matA})$ is $\vec{\matA}$,
so 
\begin{align*}
\matM(\overline{{\cal U}_{\matA}}) & =\overline{\matM({\cal U}_{\matA})}\\
 & =\overline{\left[\begin{array}{cccc}
| & * & \cdots & *\\
\vec{\matA} & \vdots &  & \vdots\\
| & * & \cdots & *
\end{array}\right]}\\
 & =\left[\begin{array}{cccc}
| & * & \cdots & *\\
\vec{\overline{\matA}} & \vdots &  & \vdots\\
| & * & \cdots & *
\end{array}\right]\\
 & =\matM({\cal U}_{\overline{\matA}}).
\end{align*}
\end{proof}

\subsubsection{\label{subsec:matrix-transpose}Matrix transpose:\protect \\
${\cal U}_{\protect\matA^{\protect\T}}\gets\textrm{qmsla.transpose}({\cal U}_{\protect\matA})$}

Given a matrix state preparation circuit for a matrix $\matA\in\C^{m\times n}$,
we can construct a matrix state preparation circuit for $\matA^{\T}$
based on the following proposition. See Figure~\ref{fig:matrix-transpose}
for a graphical description. It is easy to see then that $g_{\matA^{\T}}=g_{\matA}$
and $d_{\matA^{\T}}=d_{\matA}$, if \noun{EliminatePermutations} is
used to reduce the gate complexity and depth. The cost of $\text{qmsla.transpose}$
is $O(g_{\matA^{\T}})$.
\begin{prop}
\label{prop:matrix-transpose}Given $(\log_{2}n+\log_{2}m)$-qubit
\emph{state preparation circuit} ${\cal U}_{\matA}$, the circuit
$\calS_{(1,0)}^{R}\cdot{\cal U}_{\matA}$ is a state preparation circuit
for $\matA^{\T}$.
\end{prop}

\begin{proof}
Denote $\bar{\sigma}_{\T}=(1,0)$ . We have that,
\begin{align*}
\calS_{\bar{\sigma}_{\T}}^{R}\cdot{\cal U}_{\matA}\Ket 0_{\log_{2}mn} & =\calS_{\bar{\sigma}_{\T}}^{R}\kket{\matA}\\
 & =\calS_{\bar{\sigma}_{\T}}^{R}\left(\frac{1}{\FNorm{\matA}}\sum_{j=0}^{n-1}\sum_{i=0}^{m-1}a_{i,j}\Ket{mj+i}_{\log_{2}mn}\right)\\
 & =\calS_{\bar{\sigma}_{\T}}^{R}\left(\frac{1}{\FNorm{\matA}}\sum_{j=0}^{n-1}\sum_{i=0}^{m-1}a_{i,j}\Ket j_{\log_{2}n}\Ket i_{\log_{2}m}\right)\\
 & =\frac{1}{\FNorm{\matA^{\T}}}\sum_{j=0}^{n-1}\sum_{i=0}^{m-1}a_{i,j}\Ket i_{\log_{2}m}\Ket j_{\log_{2}n}\\
 & =\frac{1}{\FNorm{\matA^{\T}}}\sum_{j=0}^{n-1}\sum_{i=0}^{m-1}a_{i,j}\Ket{ni+j}_{\log_{2}mn}\\
 & =\kket{\matA^{\T}}
\end{align*}
So, we see that $\calS_{(1,0)}^{R}\cdot{\cal U}_{\matA}$ is a state
preparation circuit for $\matA^{\T}$.
\end{proof}

\begin{figure}[tph]
\noindent \begin{centering}
\centering     
\begin{tabular}{c  c  c} 
\Qcircuit @C=1em @R=.9em  {    & 	&   		&     &   		& 					& &	& \\ &\lstick{} 	&  \qw 		& \multigate{4}{{\cal U}_{\matA^{\T}}}     &   		& 					& &	& \\ &\lstick{} 	&  \qw 		& \ghost{{\cal U}_{\matA^{\T}}} 		   & {/} \qw	&  \qw   		\inputgrouph{2}{6}{3.5em}{\Ket 0_{\log_{2}mn}}{1.0em} & \dstick{c:\log_{2}n} & &  \\ &\lstick{} 	&  \qw 		& \ghost{{\cal U}_{\matA^{\T}}}  		   &   		& 					& &	& \rstick{\kket{\matA^{\T}}}  \\ &\lstick{} 	&  \vdots 	&\nghost{{\cal U}_{\matA^{\T}}} 		   &			{/} \qw & \qw    & \dstick{r:\log_{2}m} & & \\		 &\lstick{} 	&  \qw 		& \ghost{{\cal U}_{\matA^{\T}}} 		   & 			 &    & &  &  \\		   }  
& 
\raisebox{-5.5em}{\hspace{1cm}\scalebox{1.5}{\quad = \quad}}     
& 
\hspace{1.25cm}   
\Qcircuit @C=1em @R=.9em  {        & 	&   		& &   	{\hspace{-1.0em}\tiny\textsc{EliminatePermutations}}					&   			&   		 	&   		&			&   & \\ &\lstick{} 	&  \qw 		& \multigate{4}{{\cal U}_{\matA}} &   						&   			&  \nmultigate{4}{\calS_{\bar{\sigma}_{\T}}^{R}} 													 	&   		&			&   & \\ &\lstick{} 	&  \qw 		& \ghost{{\cal U}_{\matA}} 		  & \dstick{{\hspace{0.85em}}\log_{2}m}  \qw & {/} \qw 	&  \ghost{\calS_{\bar{\sigma}_{\T}}^{R}}  \inputgrouph{2}{6}{3.7em}{\Ket 0_{\log_{2}mn}}{1.0em}     &  \qw 	& {/} \qw  	& \dstick{\raisebox{-1.0em}{\hspace{0.5em}}c:\log_{2}n}  \qw  &  \\ &\lstick{} 	&  \qw 		& \ghost{{\cal U}_{\matA}}  	  &  						&  			&  \nghost{\calS_{\bar{\sigma}_{\T}}^{R}}  															&		&			&   &  \rstick{\raisebox{-0.5em}{\hspace{+2em}}\kket{\matA^{\T}} }  \\ &\lstick{} 	&  \vdots 	&\nghost{{\cal U}_{\matA}} 		  & \dstick{{\hspace{0.7em}}\log_{2}n}  \qw & {/} \qw 	&  \ghost{\calS_{\bar{\sigma}_{\T}}^{R}}																	  	&  \qw   & {/} \qw 	& \dstick{\raisebox{-1.0em}{\hspace{0.5em}}r:\log_{2}m} \qw  & \\		 &\lstick{} 	&  \qw 	\gategroup{2}{4}{6}{7}{.7em}{--}	& \ghost{{\cal U}_{\matA}} 		  & 						& 			&  \nghost{\calS_{\bar{\sigma}_{\T}}^{R}}																		  	& 			&  			&   &  \\		   } 
\end{tabular}
\par\end{centering}
\caption{\label{fig:matrix-transpose}Graphical description of $\text{qmsla.transpose}$
(Input: ${\cal U}_{\protect\matA}$, Output: ${\cal U}_{\protect\matA^{\protect\T}}$).}
\end{figure}
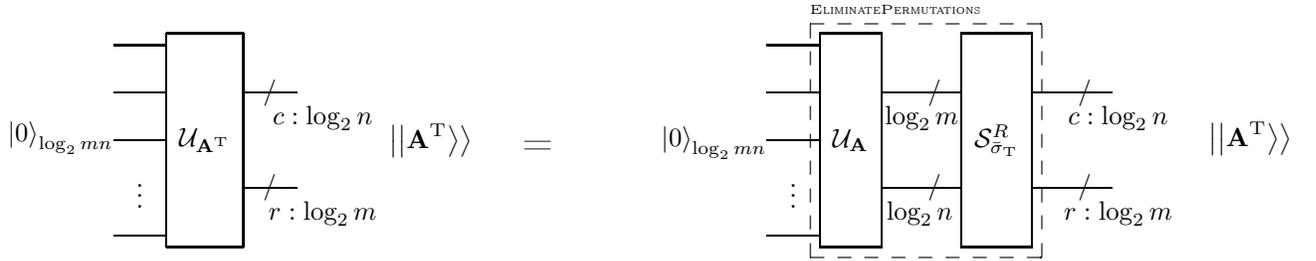

\subsubsection{\label{subsec:vectorize-matrix}Vectorize:\protect \\
${\cal U}_{\protect\vec{\protect\matA}}\gets\textrm{qmsla.vec}({\cal U}_{\protect\matA})$}

Notice that $\kket{\matA}=\kket{\vec{\matA}}$ when considering only
probability amplitudes. Indeed, each state is obtained by a different
logical split of the qubits between the two registers. In other words,
the matrix state preparation circuit ${\cal U}_{\matA}$ can be interpreted
also as ${\cal U}_{\vec{\matA}}$, where the circuit itself is the
same, but the meta-data (partition of the qubits to registers) is
different. Thus, $\textrm{qmsla.vec}$ can be implemented via meta-data
transformation. Gate complexity and depth do not change, and the cost
of the operation is $O(g_{\matA})$.

\subsubsection{\label{subsec:pad-columns}Pad with zero columns: \protect \\
${\cal U}_{\protect\matA\oplus0_{0\times(2^{k}-1)n}}\gets\textrm{qmsla.pad\_zero\_columns}({\cal U}_{\protect\matA},k)$
($k>0$)}

Consider a state preparation circuits ${\cal U}_{\matA}$, where $\matA\in\C^{m\times n}$.
Given an integer $k>0$, the following identity shows how to construct
a matrix state preparation circuit for $\left[\begin{array}{cc}
\matA & \mat 0_{m\times(2^{k}-1)n}\end{array}\right]$:
\begin{align*}
{\cal I}_{k}\otimes{\cal U}_{\matA}\Ket 0_{k}\Ket 0_{\log_{2}mn} & =\Ket 0_{k}\kket{\matA}\\
 & =\kket{\left[\begin{array}{cc}
\matA & \mat 0\end{array}\right]}
\end{align*}
where ${\cal I}_{k}$ is the $k$-qubit empty (identity) circuit.
So, ${\cal U}_{\matA\oplus0_{0\times(2^{k}-1)n}}={\cal I}_{k}\otimes{\cal U}_{\matA}$.\textcolor{red}{{}
}See Figure~\ref{fig:matrix-extened} for a graphical description
of $\textrm{qmsla.pad\_zero\_columns}$. It is easy to see then that
$g_{\matA\oplus0_{0\times(2^{k}-1)n}}=g_{\matA}$ and $d_{\matA\oplus0_{0\times(2^{k}-1)n}}=d_{\matA}$
and the cost of the algorithm is $O(g_{\matA})$.

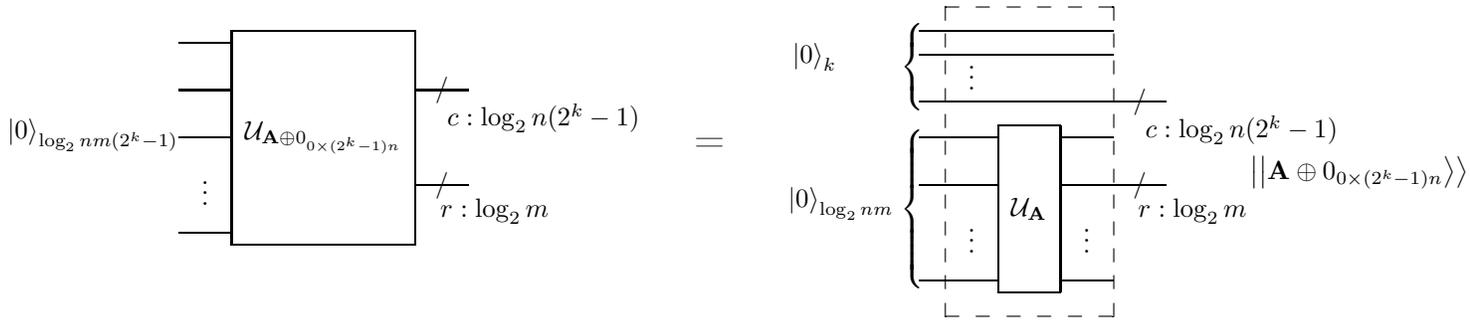
\begin{figure}[tph]
\noindent \begin{centering}
\centering     
\begin{tabular}{c  c  c} 
\Qcircuit @C=1em @R=.9em  {    & 	&   		&     &   		& 					& &	& \\ &\lstick{} 	&  \qw 		& \multigate{4}{{\cal U}_{\matA\oplus0_{0\times(2^{k}-1)n}}}     &   		& 					& &	& \\ &\lstick{} 	&  \qw 		& \ghost{{\cal U}_{\matA\oplus0_{0\times(2^{k}-1)n}}} 		   & {/} \qw	&  \qw   		\inputgrouph{2}{6}{3.5em}{\hspace{-2.5em}\Ket 0_{\log_{2}nm(2^{k}-1)}}{1.0em} & \dstick{{\hspace{+3.7em}}c:\log_{2}n(2^{k}-1)} & &  \\ &\lstick{} 	&  \qw 		& \ghost{{\cal U}_{\matA\oplus0_{0\times(2^{k}-1)n}}}  		   &   		& 					& &	&   \\ &\lstick{} 	&  \vdots 	&\nghost{{\cal U}_{\matA\oplus0_{0\times(2^{k}-1)n}}} 		   &			{/} \qw & \qw    & \dstick{r:\log_{2}m} & & \\		 &\lstick{} 	&  \qw 		& \ghost{{\cal U}_{\matA\oplus0_{0\times(2^{k}-1)n}}} 		   & 			 &    & &  &  \\	 & 	&   		&     &   		& 					& &	& \\ & 	&   		&     &   		& 					& &	& \\ & 	&   		&     &   		& 					& &	& \\ & 	&   		&     &   		& 					& &	& \\	   } 
& 
\raisebox{-5.5em}{\hspace{1cm}\scalebox{1.5}{\quad = \quad}}     
& 
\hspace{1.25cm}   
\Qcircuit @C=1em @R=.9em  {                                                                                               														 &		&   						& 									&		   &   	                                   	& 		 		&&\\ &\qw	&  \qw 						&\qw 								& \qw 	   & \qw                                    & 				&& \\ &\qw	&  \qw 						&\qw 								& \qw 	   & \qw                                    &   			&&  \\ &		&  \raisebox{1.5ex}{\vdots} &	 								&  	  	   &  	                                    & 		&& \\ &\qw 	&  \qw 						&\qw 								& \qw 	   & \qw \gategroup{1}{2}{10}{6}{0.0em}{--} & 	{/} \qw		&	\qw  \inputgroupv{2}{5}{0.5em}{1.0em}{\hspace{-4.5em} \Ket 0_{k}}			&   \dstick{{\hspace{+3.7em}}c:\log_{2}n(2^{k}-1)} 										\\ &\qw 	&  \qw 						& \multigate{3}{{\cal U}_{\matA}} 	&	\qw	   & \qw  	                                & 	&& \rstick{\raisebox{+3.5em}{\hspace{1.5em}}\kket{\matA\oplus0_{0\times(2^{k}-1)n}}} \\ &\qw 	&  \qw 						& \ghost{{\cal U}_{\matA}}  	  	&  \qw     & \qw                                    & 	 {/} \qw & 	\qw \inputgroupv{6}{9}{0.5em}{2.5em}{\hspace{-2.5em} \Ket 0_{\log_{2}nm}}	& \dstick{r:\log_{2}m} 		\\ & 		&  \vdots 					& \nghost{{\cal U}_{\matA}} 		&  \vdots  & 	                                    &   &&  \\ &\qw	&  \qw 						& \ghost{{\cal U}_{\matA}} 			&  	\qw    & \qw 	                                &		   		&& 	 \\ &		&   						& 									&		   &   	                                   	& 		 		& & 	\\ }
\end{tabular}
\par\end{centering}
\caption{\label{fig:matrix-extened}Graphical description of $\text{qmsla.pad\_zero\_columns}$
(Input: ${\cal U}_{\protect\matA}$ and $k$, Output: ${\cal U}_{\protect\matA\oplus0_{0\times(2^{k}-1)n}}$)}
\end{figure}

\subsubsection{\label{subsec:matrix-vector-product}Matrix-vector product:\protect \\
${\cal U}_{\protect\y}\gets\textrm{qmsla.matrix\_vec}({\cal U}_{\protect\matA},{\cal U}_{\protect\b})$
where $\protect\y=\frac{\protect\matA\protect\b}{\protect\FNorm{\protect\matA}\protect\TNorm{\protect\b}}\oplus\protect\vxi\in\protect\C^{mn}$
for some garbage $\protect\vxi\in\protect\C^{m(n-1)}$ and $\protect\TNorm{\protect\y}=1$}

Assume that $\matA\in\C^{m\times n}$ and $\b\in\C^{n}$. Consider
applying ${\cal U}_{\b}^{\T}$ to the MSB register (i.e. column register)
of the state $\kket{\matA}$. We have (Lemma~\ref{prop:upper-lower-opperations}):
\begin{align*}
{\cal U}_{\b}^{\T}\cdot_{0}\kket{\matA} & =\kket{\matA\matM({\cal U}_{\b}^{\T})^{\T}}\\
 & =\kket{\matA\matM({\cal U}_{\b})}\\
 & =\kket{\matA\left[\begin{array}{cccc}
| & * & \cdots & *\\
\frac{\b}{\TNorm{\b}} & \vdots &  & \vdots\\
| & * & \cdots & *
\end{array}\right]}\\
 & =\kket{\left[\begin{array}{cccc}
| & * & \cdots & *\\
\frac{\matA\b}{\TNorm{\b}} & \vdots &  & \vdots\\
| & * & \cdots & *
\end{array}\right]}\\
 & =\Ket{\left[\begin{array}{c}
\frac{\matA\b}{\TNorm{\b}}\\*
\vdots\\
\vdots\\*
\end{array}\right]}=\Ket{\left[\begin{array}{c}
\frac{\matA\b}{\FNorm{\matA}\TNorm{\b}}\\*
\vdots\\
\vdots\\*
\end{array}\right]}
\end{align*}
where $*$ denotes arbitrary amplitudes. The third equality follows
from the fact that first column of $\matM({\cal U}_{\b})$ is equal
to $\b/\TNorm{\b}$. The equality follows since the ket notation $\Ket{\cdot}$
we use is oblivious to a global scaling by a positive real number
(i.e. $\Ket{\alpha\x}=\Ket{\x}$ for any real $\alpha>0$), and so
we can rescale the first element and the additional arbitrary amplitudes
below it. We now denote the values in the lower part of the vector
by $\vxi$. The normalization we introduced by dividing by $\FNorm{\matA}$
ensures that $\TNorm{\y}=1$. 

See Figure~\ref{fig:matrix-vector-product} for a graphical description
of $\textrm{qmsla.matrix\_vec}$. The gate complexity is $g_{\matA}+g_{\b}$
and the depth is bounded by $d_{\matA}+d_{\b}$. The number of qubits
is $q_{\matA}$. The cost of the algorithm is $O(g_{\matA}+g_{\b})=O(g_{\y})$.

\begin{figure}[tph]
\noindent \begin{centering}
\centering    
\begin{tabular}{c  c  c}
\Qcircuit @C=1em @R=.9em  {     &\lstick{} 	&  		&    &   		& 					& &	& \\	  &\lstick{} 	&  		&    &   		& 					& &	& \\	    &\lstick{} 	&  \qw 		& \multigate{4}{{\cal U}_{\frac{\matA\b}{\FNorm{\matA}\TNorm{\b}}\oplus\vxi}}    &   		& 					& &	& \\  &\lstick{} 	&  \qw 		& \ghost{{\cal U}_{\frac{\matA\b}{\FNorm{\matA}\TNorm{\b}}\oplus\vxi}} 		   & 	&   & &  \\  &\lstick{} 	&  \qw 		& \ghost{{\cal U}_{\frac{\matA\b}{\FNorm{\matA}\TNorm{\b}}\oplus\vxi}}  		   & {/} \qw  		& 	\qw   		\inputgrouph{2}{6}{3.5em}{\Ket 0_{\log_{2}mn}}{1.0em} & \dstick{r:\log_{2}mn}				& &	&   \\  &\lstick{} 	&  \vdots 	&\nghost{{\cal U}_{\frac{\matA\b}{\FNorm{\matA}\TNorm{\b}}\oplus\vxi}} 		   &			 &  &  & & \\		  &\lstick{} 	&  \qw 		& \ghost{{\cal U}_{\frac{\matA\b}{\FNorm{\matA}\TNorm{\b}}\oplus\vxi}} 		   & 			 &    & &  &  \\	  &\lstick{} 	&  		&    &   		& 					& &	& \\	   }  
& 
\raisebox{-5.0em}{\scalebox{1.5}{\quad = \quad}}    
& 
\hspace{1.25cm}  
\Qcircuit @C=1em @R=.9em  {         &			&  		&   {\hspace{-1.0em}\tiny\textsc{Vec}}  &  & & &   	&			&   & \\    &			&  		&    									&  & & &   	&			&   & \\                                              &\lstick{} 	&  \qw & \multigate{4}{{\cal U}_{\matA}}   	&  & &		 		&   	&			&   & \\ &\lstick{} 	&  \qw & \ghost{{\cal U}_{\matA}} 				&  {/} \qw	& \dstick{{\hspace{-0.7em}}{\tiny\textrm{r:\ensuremath{\log_{2}n}}}}	\qw		&	\qw			&  \qw 	&   	&  &  \\ &\lstick{} 	&  \qw 									& \ghost{{\cal U}_{\matA}}  			& 			&	&	  \nmultigate{2}{{\cal U}_{\b}^{\T}} 		&		& {/} \qw			&  \dstick{\raisebox{-1.0em}{\hspace{1.0em}}r:\log_{2}mn}  \qw  &  \rstick{\raisebox{-3.5em}{\hspace{+2em}}\Ket{\frac{\matA\b}{\FNorm{\matA}\TNorm{\b}}\oplus\vxi} }  \\ &\lstick{} 	&  {\hspace{0.5em}\vdots} 				&\nghost{{\cal U}_{\matA}} 				&  {/} \qw	&  \dstick{{\hspace{-0.7em}}{\tiny\textrm{c:\ensuremath{\log_{2}m}}}} \qw		&     \ghost{{\cal U}_{\b}^{\T}}  \inputgrouph{2}{6}{3.7em}{\Ket 0_{2\log_{2}mn}}{1.0em} 																							&  \qw  																					&  	&   & \\		 &\lstick{} 	&  \qw \gategroup{2}{3}{8}{8}{0.0em}{--}& \ghost{{\cal U}_{\matA}} 				& 			&&    \nghost{{\cal U}_{\b}^{\T}}  	& 		&  			&   &  \\		   &	&  			&   &   		&& 	&   	&			&   & \\                                              }
\end{tabular}
\par\end{centering}
\caption{\label{fig:matrix-vector-product}Graphical description of $\text{qmsla.matrix\_vec}$
(Input: ${\cal U}_{\protect\matA}$ and ${\cal U}_{\protect\b}$,
Output: ${\cal U}_{\frac{\protect\matA\protect\b}{\protect\FNorm{\protect\matA}\protect\TNorm{\protect\b}}\oplus\protect\vxi}$).}
\end{figure}

\subsubsection{\label{subsec:kronecker-product}Kronecker product:\protect \\
${\cal U}_{\protect\mat A\otimes\protect\matB}\gets\textrm{qmsla.kronecker}({\cal U}_{\protect\matA},{\cal U}_{\protect\matB})$}

Consider two state preparation circuits ${\cal U}_{\matA}$ and ${\cal U}_{\matB}$.
We have that $({\cal U}_{\matA}\otimes{\cal U}_{\matB})\cdot\Ket 0_{q_{\matA}+q_{\matB}}=\kket{\matA}\kket{\matB}$.
The state $\kket{\matA}\kket{\matB}=\kket{\matA}\otimes\kket{\matB}$
is generally not equal to $\kket{\matA\otimes\matB}$, and ${\cal U}_{\matA}\otimes{\cal U}_{\matB}$
is not a state preparation circuit for $\matA\otimes\matB$. Nevertheless,
the following proposition shows that by permuting the order of the
registers of $\kket{\matA}\kket{\matB}$ we can obtain the state $\kket{\matA\otimes\matB}$.
This allows us to construct a matrix state preparation circuit ${\cal U}_{\matA\otimes\matB}$
for $\matA\otimes\matB$.
\begin{prop}
Let 
\[
\overline{\sigma}_{\otimes}\coloneqq(0,2,1,3)
\]
Then 
\[
\calS_{\bar{\sigma}_{\otimes}}^{R}\cdot\kket{\matA}\kket{\matB}=\kket{\matA\otimes\matB}.
\]
\end{prop}

\begin{proof}
Suppose that $\matA$ is $m\times n$ and $\matB$ is $p\times q$.
First, we show that 
\[
\calS_{\bar{\sigma}_{\otimes}}^{R}\kket{\mat E_{i,j}^{m\times n}}\kket{\mat E_{k,l}^{p\times q}}=\kket{\mat E_{pi+k,qj+l}^{mp\times nq}}
\]
for all $i,j,k$ and $l$. Indeed, 
\begin{align*}
\calS_{\bar{\sigma}_{\otimes}}^{R}\kket{\mat E_{i,j}^{m\times n}}\kket{\mat E_{k,l}^{p\times q}} & =\calS_{\bar{\sigma}_{\otimes}}^{R}\Ket j_{\log_{2}n}\Ket i_{\log_{2}m}\Ket l_{\log_{2}q}\Ket k_{\log_{2}p}\\
 & =\Ket j_{\log_{2}n}\Ket l_{\log_{2}q}\Ket i_{\log_{2}m}\Ket k_{\log_{2}p}\\
 & =\kket{\mat E_{l,j}^{q\times n}}\kket{\mat E_{k,i}^{p\times m}}\\
 & =\Ket{\e_{qj+l}^{nq}}\Ket{\e_{pi+k}^{mp}}\\
 & =\kket{\mat E_{pi+k,qj+l}^{mp\times nq}}
\end{align*}
Now,

\begin{align*}
\calS_{\bar{\sigma}_{\otimes}}^{R}\kket{\matA}\kket{\matB} & =\calS_{\bar{\sigma}_{\otimes}}^{R}\left(\frac{1}{\FNorm{\matA}}\sum_{i=0}^{m-1}\sum_{j=0}^{n-1}a_{ij}\kket{\mat E_{i,j}^{m\times n}}\right)\left(\frac{1}{\FNorm{\matB}}\sum_{k=0}^{p-1}\sum_{l=0}^{q-1}b_{kl}\kket{\mat E_{k,l}^{p\times q}}\right)\\
 & =\calS_{\bar{\sigma}_{\otimes}}^{R}\left(\frac{1}{\FNorm{\matA}\FNorm{\matB}}\sum_{i=0}^{m-1}\sum_{j=0}^{n-1}\sum_{k=0}^{p-1}\sum_{l=0}^{q-1}a_{ij}b_{kl}\kket{\mat E_{i,j}^{m\times n}}\kket{\mat E_{k,l}^{p\times q}}\right)\\
 & =\frac{1}{\FNorm{\matA}\FNorm{\matB}}\sum_{i=0}^{m-1}\sum_{j=0}^{n-1}\sum_{k=0}^{p-1}\sum_{l=0}^{q-1}a_{ij}b_{kl}\calS_{\bar{\sigma}_{\otimes}}\kket{\mat E_{i,j}^{m\times n}}\kket{\mat E_{k,l}^{p\times q}}\\
 & =\frac{1}{\FNorm{\matA\otimes\matB}}\sum_{i=0}^{m-1}\sum_{j=0}^{n-1}\sum_{k=0}^{p-1}\sum_{l=0}^{q-1}a_{ij}b_{kl}\kket{\mat E_{pi+k,qj+l}^{mp\times nq}}\\
 & =\kket{\matA\otimes\matB}
\end{align*}
\end{proof}
Based on the last proposition, given matrix state preparation circuits
${\cal U}_{\matA}$ and ${\cal U}_{\matB}$ we can create a matrix
state preparation circuit for ${\cal U}_{\matA\otimes\matB}$. See
Figure~\ref{fig:circuit-for-u_kron} for a graphical description.
Thus, we have $g_{\matA\otimes\matB}=g_{\matA}+g_{\matB}$ and $d_{\matA\otimes\matB}=\max(d_{\matA},d_{\matB})$,
once we apply \noun{EliminatePermutations} to reduce the gate complexity
and depth. The cost of $\textrm{qmsla.kronecker}$ is $O(g_{\matA\otimes\matB})$.

\begin{figure}[tph]
\noindent \begin{centering}
\centering    
\begin{tabular}{c  c  c}
\Qcircuit @C=1em @R=.9em  {    & 			&   		&    											&   						& 					& &	& \\ &\lstick{} 	&  \qw 		& \multigate{4}{{\cal U}_{\mat A\otimes\matB}}  &   						& 					& &	& \\ &\lstick{} 	&  \qw 		& \ghost{{\cal U}_{\mat A\otimes\matB}} 									    & {/} \qw						&  \qw   		\inputgrouph{2}{6}{3.5em}{\hspace{-2.5em}\Ket 0_{\log_{2}mnpq}}{1.0em} & \dstick{r:\log_{2}pm} & &  \\ &\lstick{} 	&  \qw 		& \ghost{{\cal U}_{\mat A\otimes\matB}}  								    &   							& 					& &	&   \\ &\lstick{} 	&  \vdots 	&\nghost{{\cal U}_{\mat A\otimes\matB}} 									    &			{/} \qw 			& \qw    	& \dstick{c:\log_{2}qn} & & \\		 &\lstick{} 	&  \qw 		& \ghost{{\cal U}_{\mat A\otimes\matB}} 									    & 							&  			 & &  &  \\	 & 			&   		&   										    &   							& 					& &	& \\ & 			&   		&   										    &   							& 					& &	& \\ & 			&   		&   										    &   							& 					& &	& \\ & 			&   		&   										    &   							& 					& &	& \\  } 
& 
\raisebox{-5.0em}{\scalebox{1.5}{\quad = \quad}}    
& 
\hspace{1.25cm}  
\Qcircuit @C=1em @R=.9em  {   & \lstick{} &							& 									& {\raisebox{-1.0em}{\hspace{-2.0em}}\tiny\textsc{EliminatePermutations}} 				& 		&&&	&&&\\ & \lstick{} & \qw 						& \multigate{3}{{\cal U}_{\matB}}	&  			& & \nmultigate{7}{\calS_{\bar{\sigma}_{\otimes}}^{R}}  &  		&    && &  \\ & \lstick{} & \qw 						& \ghost{{\cal U}_{\matB}}			& {/} \qw 	& \dstick{{\hspace{-0.7em}}{\tiny\textrm{r:\ensuremath{\log_{2}p}}}} \qw \inputgroupv{2}{5}{0.5em}{2.5em}{\hspace{-2.5em} \Ket 0_{\log_{2}pq}}	 & \ghost{\calS_{\bar{\sigma}_{\otimes}}^{R}}		& {/} \qw & \dstick{{\hspace{-0.em}}{\tiny\textrm{\ensuremath{r_{\matB}:\log_{2}p}}}} \qw	&       \\ & \lstick{} & {\hspace{0.5em}\vdots}	&\nghost{{\cal U}_{\matB}}			& {/} \qw 	& \dstick{{\hspace{-0.7em}}{\tiny\textrm{c:\ensuremath{\log_{2}q}}}} \qw & \ghost{\calS_{\bar{\sigma}_{\otimes}}^{R}}		 & {/} \qw 	& \dstick{{\hspace{-0.0em}}{\tiny\textrm{\ensuremath{r_{\matA}:\log_{2}m}}}} \qw &  \qw & {/} \qw & \dstick{{\hspace{+3.7em}}r:\log_{2}pm}  \qw   \\ & \lstick{} & \qw 						& \ghost{{\cal U}_{\matB}}			&  			& & \nghost{\calS_{\bar{\sigma}_{\otimes}}^{R}}																								  & 	&&	&    & \rstick{\raisebox{+3.5em}{\hspace{1.5em}}\kket{\matA\otimes\matB}}  \\ & \lstick{} & \qw 						& \multigate{3}{{\cal U}_{\matA}}	&  			& & \nghost{\calS_{\bar{\sigma}_{\otimes}}^{R}}																								  & 	&&	&   &    \\ & \lstick{} & \qw 						& \ghost{{\cal U}_{\matA}}			& {/}  \qw	& \dstick{{\hspace{-0.7em}}{\tiny\textrm{r:\ensuremath{\log_{2}m}}}} \qw \inputgroupv{6}{9}{0.5em}{2.5em}{\hspace{-2.5em}\Ket 0_{\log_{2}mn}} & \ghost{\calS_{\bar{\sigma}_{\otimes}}^{R}}		 & {/}  \qw	&  \dstick{{\hspace{0.0em}}{\tiny\textrm{\ensuremath{c_{\matB}:\log_{2}q}}}} \qw & \qw & {/} \qw &\dstick{{\hspace{+3.7em}}c:\log_{2}qn} \qw  \\ & \lstick{} & {\hspace{0.5em}\vdots}	&\nghost{{\cal U}_{\matA}}			& {/}  \qw	& \dstick{{\hspace{-0.7em}}{\tiny\textrm{c:\ensuremath{\log_{2}n}}}} \qw & \ghost{\calS_{\bar{\sigma}_{\otimes}}^{R}}		 & {/} \qw & \dstick{{\hspace{.0em}}{\tiny\textrm{\ensuremath{c_{\matA}:\log_{2}n}}}} \qw     &   \\ & \lstick{} & \qw 	\gategroup{1}{3}{10}{9}{.0em}{--}						& \ghost{{\cal U}_{\matA}}												        	 &  			& & \nghost{\calS_{\bar{\sigma}_{\otimes}}^{R}}		 &  		&   &&   & \\ & \lstick{} &  							& 									&  																		        	 & & 		 &			& \\ }
\end{tabular}
\par\end{centering}
\caption{\label{fig:circuit-for-u_kron}Graphical description of $\text{qmsla.kronecker}$
(Input: ${\cal U}_{\protect\matA}$ and ${\cal U}_{\protect\matB}$,
Output: ${\cal U}_{\protect\matA\otimes\protect\matB}$).}
\end{figure}
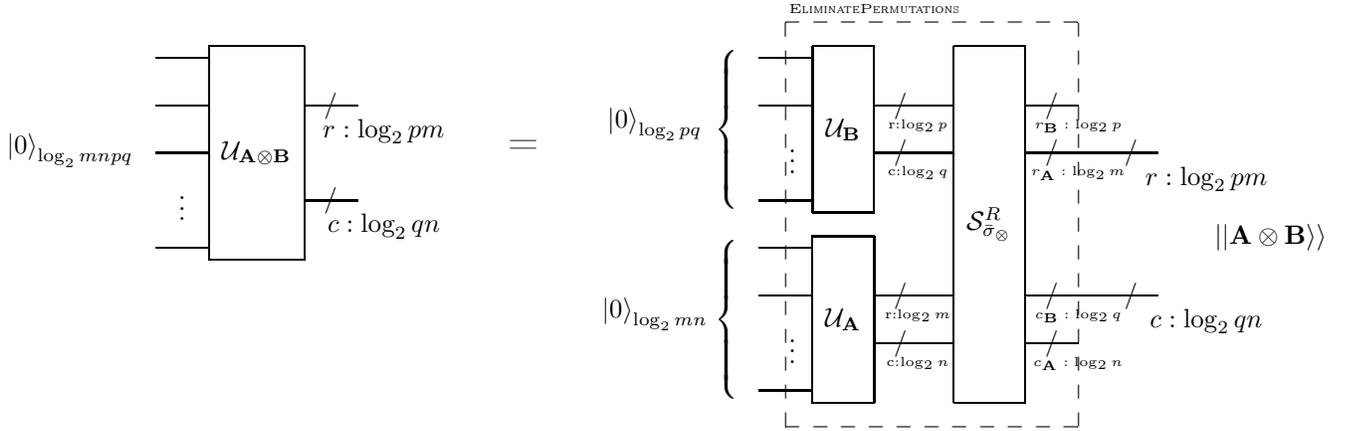

\subsubsection{\label{subsec:multiple-kronecker}Multiple kronecker products:\protect \\
${\cal U}_{\protect\matA_{1}\otimes\protect\matA_{2}\otimes\cdots\otimes\protect\matA_{k}}\gets\textrm{qmsla.kronecker}({\cal U}_{\protect\matA_{1}},\dots,{\cal U}_{\protect\matA_{k}})$}

Given a list of state preparation circuits ${\cal U}_{\matA_{1}},\dots{\cal U}_{\matA_{k}}$
we can compute ${\cal U}_{\matA_{1}\otimes\matA_{2}\otimes\cdots\otimes\matA_{k}}$
using recursion, with the algorithm from the previous subsection used
to combine results. That is, if $k=1$ we simply return the circuit
itself. For $k\geq2$ we recursively form ${\cal U}_{\matA_{1}\otimes\cdots\otimes\matA_{\left\lfloor k/2\right\rfloor }}$
and ${\cal U}_{\matA_{\left\lfloor k/2\right\rfloor +1}\otimes\cdots\otimes\matA_{k}}$,
and the combine the results to form ${\cal U}_{\matA_{1}\otimes\matA_{2}\otimes\cdots\otimes\matA_{k}}$
using the algorithm from the previous subsection. Assuming $\matA_{i}\in\C^{m_{i}\times n_{i}}$
for $i=1,\dots,k$, we have, $g_{\matA_{1}\otimes\cdots\otimes\matA_{k}}=g_{\matA_{1}}+\cdots+g_{\matA_{k}}$
and $d_{\matA_{1}\otimes\cdots\otimes\matA_{k}}=\max(d_{\matA_{1}},\dots,d_{\matA_{k}})$,
once we apply \noun{EliminatePermutations} to reduce the gate complexity
and depth.

The recursive algorithm has total complexity of $O(\log k\sum_{i=2}^{2k}g_{\matA_{i}})$.
This arises from the fact that each matrix appears in $O(\log k)$
\noun{EliminatePermutations} operations, with a cost of $O(g(\text{output}))$
each. Thus, it is better to us a non-recursive algorithm: In particular,
we could compute the permutation associated with the Kronecker products
upfront, and apply the \noun{EliminatePermutations }once\noun{, }to
reduce cost to $O(g_{\matA_{1}\otimes\cdots\otimes\matA_{k}})$\noun{.}

\subsection{\label{subsec:lev2}Level 2 qMSLA operations}

Having defined a basic set of qMSLA primitives, we can start composing
them into more complex operations. Here, we describe a few such operations
that prove useful for estimating multivariate traces.

\subsubsection{\label{subsec:pad-rows}Pad with zero rows and padding diagonally:\protect \\
${\cal U}_{\protect\matA\oplus0_{(2^{r}-1)m\times0}}\gets\textrm{qmsla.pad\_zero\_rows}({\cal U}_{\protect\matA},r)$
($r>0$)\protect \\
${\cal U}_{\protect\matA\oplus0_{(2^{r}-1)m\times(2^{k}-1)n}}\gets\text{qmsla.pad}(\protect\matA,r,k)$
($k,r>0$)}

Assuming we already have $\textrm{qmla.pad\_zero\_columns(\ensuremath{{\cal U}_{\matA}},r)}$,
then

\[
\text{qmlsa.pad\_zero\_rows}({\cal U}_{\matA},r)\coloneqq\textrm{qmsla.transpose}(\textrm{qmsla.pad\_zero\_columns}(\textrm{qmsla.transpose}({\cal U}_{\matA}),r))
\]
However, it is easy to see that this amounts to adding an empty register
between the column and row register. Thus, $g_{\matA\oplus0_{(2^{r}-1)m\times0}}=g_{\matA}$
and $d_{\matA\oplus0_{(2^{r}-1)m\times0}}=d_{\matA}$.

Both padding operations can be combined to a unified single padding
operation:
\[
\text{qmlsa.pad}(\matA,r,k)\coloneqq\textrm{qmsla.pad\_zero\_rows}(\textrm{qmsla.pad\_zero\_columns}({\cal U}_{\matA},k),r))
\]
This amounts to adding two empty registers: one before the column
register and one between the two registers. Again, $g_{\matA\oplus0_{(2^{r}-1)m\times(2^{k}-1)n}}=g_{\matA}$
and $d_{\matA\oplus0_{(2^{r}-1)m\times(2^{k}-1)n}}=d_{\matA}$, and
the cost is $O(g_{\matA})$.

\subsubsection{\label{subsec:adjoint}Matrix adjoint:\protect \\
${\cal U}_{\protect\matA^{\protect\conj}}\gets\textrm{qmsla.adjoint}({\cal U}_{\protect\matA})$}

Since $\matA^{\conj}=(\overline{\matA})^{\T}$ we can compose $\textrm{qmsla.conjugate}$
and $\textrm{qmsla.transpose}$ (order does not matter), to compute
${\cal U}_{\matA^{\conj}}$ given ${\cal U}_{\matA}$. Thus, $g_{\matA^{\conj}}=g_{\matA}$
and $d_{\matA^{\star}}=d_{\matA}$. The cost is $O(g_{\matA})$.

\subsubsection{\label{subsec:overlap}Overlap:\protect \\
${\cal U}_{\frac{\protect\vpsi^{\protect\conj}\protect\vphi}{\protect\TNorm{\protect\vpsi}\protect\TNorm{\protect\vphi}}\oplus\protect\vxi}\gets\text{qmsla.overlap}({\cal U}_{\protect\vpsi},{\cal U}_{\protect\vphi})$}

Given two state preparation circuits ${\cal U}_{\vpsi}$ and ${\cal U}_{\vphi}$
of vectors $\vpsi$ and $\vphi$ of the same size, we can build a
state preparation circuit for a unit vector of the form
\[
\y=\left[\begin{array}{c}
\frac{\vpsi^{\conj}\vphi}{\TNorm{\vpsi}\TNorm{\vphi}}\\*
'\vdots\\
\vdots\\*
\end{array}\right]
\]
Thus, the overlap (dot product) appears in the first probability amplitude.
The idea is to first form ${\cal U}_{\vpsi^{\conj}}$ from ${\cal U}_{\vpsi}$
using $\textrm{qmsla.adjoint}$, and then apply $\textrm{qmsla.matrix\_vector}$
to obtain ${\cal U}_{\y}$ (where we normalize the vector by dividing
by $\TNorm{\vpsi}$). A summary appears in Algorithm~\ref{alg:overlap}.
Thus, the number of gates in the resulting circuit is $g_{\vpsi}+g_{\vphi}$
and the\textcolor{red}{{} }\textcolor{black}{depth is $d_{\vpsi}+d_{\vphi}$,
and the cost of the algorithm is $O(g_{\vpsi}+g_{\vphi})$.}\textcolor{red}{{}
}Note that the resulting circuit does not use any additional control
operation beyond the ones in ${\cal U}_{\psi}$ and ${\cal U}_{\phi}$,
and no additional qubits.

\begin{algorithm}[tph]
\begin{algorithmic}[1]

\STATE \textbf{Input: }Classical description of the circuits ${\cal U}_{\vpsi},{\cal U}_{\vphi}$

\STATE 

\STATE ${\cal U}_{\vpsi^{\conj}}\gets\textrm{qmsla.adjoint}({\cal U}_{\vpsi})$

\STATE ${\cal U}_{\y}\gets\textrm{qmsla.matrix\_vec}({\cal U_{\vpsi^{\conj}}},{\cal U}_{\vphi})$

\RETURN ${\cal U}_{\y}$

\end{algorithmic}

\caption{\label{alg:overlap}\noun{$\text{qmsla.overlap}$}}
\end{algorithm}

\section{\label{sec:Encoding-Matrix-Moments}Encoding multivariate traces
in quantum states}

In this section, we leverage qMSLA to introduce our main algorithm:
\textsc{\noun{MVTracePrep}}. Given state preparation circuits for
$2k$ matrices, \textsc{\noun{MVTracePrep}} produces two (vector)
state preparation circuits, for which the overlap between the vectors
is equal to the multivariate trace of the matrices up to normalization.
More formally, assuming the inputs to \textsc{\noun{MVTracePrep}}\noun{
}are ${\cal U}_{\matA_{1}},\dots,{\cal U}_{\matA_{2k}}$, where ${\cal U}_{\matA_{i}}$
is a matrix state preparation circuit for $\matA_{i}$, \textsc{\noun{MVTracePrep}}\noun{
}outputs two state preparations circuits ${\cal U}_{\vpsi}$ and ${\cal U}_{\vphi}$
for vectors $\vpsi$ and $\vphi$ (respectively) such that $\vpsi^{\conj}\vphi=\frac{\Trace{\matA_{1}\matA_{2}\cdots\matA_{2k}}}{\FNorm{\matA_{1}}\cdots\FNorm{\matA_{2k}}}$.

A pseudocode description of \textsc{\noun{MVTracePrep}} appears in
Algorithm~\ref{alg:main}. For illustration purposes, a high-level
block circuit diagram of ${\cal U}_{\vphi}({\cal U}_{\matA_{2}},\cdots,{\cal U}_{\matA_{6}})$
(i.e., $k=3$) is shown in Figure~\ref{fig:main-alg} (${\cal U}_{\vpsi}$,
which depends only on ${\cal U}_{\matA_{1}}$, is very simple, and
is based on $\textrm{qmsla.pad}$), with the SWAPs written explicitly
(they can be eliminated using \noun{EliminatePermutations}). In the
following subsections, we explain \textsc{\noun{MVTracePrep}} and
it rationale, and finally state and prove Theorem~\ref{thm:main},
which summarizes the main properties of \textsc{\noun{MVTracePrep}}.

\begin{algorithm}[tph]
\begin{algorithmic}[1]

\STATE \textbf{Input: }Classical description of the circuits ${\cal U}_{\matA_{1}},\cdots{\cal U}_{\matA_{2k}}$

\STATE 

\STATE \COMMENT{Constructing ${\cal U}_{\psi}$} 

\STATE ${\cal U}_{\overline{\matA}_{1}}\gets\textrm{qmsla.conjugate}({\cal U}_{\matA_{1}})$

\STATE ${\cal U}_{\psi}\gets\textrm{qmsla.pad\_zero\_rows}(\textrm{qmsla.vec}({\cal U}_{\overline{\matA}_{1}}),\log_{2}n_{3}n_{4}\cdots n_{2k})$

\STATE 

\STATE \COMMENT{Constructing ${\cal U}_{\phi}$} 

\STATE $p\gets k+1$

\FOR{$i=p$ {\bf to} $i=2k$} 

\STATE ${\cal U}_{\matA_{i}^{\T}}\gets\textrm{qmsla.transpose}({\cal U}_{\matA_{i}})$

\ENDFOR 

\IF{$k$ is odd} 

\STATE $l_{\textrm{even}}\gets(k+1)/2$ and $l_{\textrm{odd}}\gets(k-1)/2$

\STATE ${\cal U}_{\matF_{\textrm{even}}^{(1)}}\gets{\cal U}_{\matA_{p}^{\T}}$
\COMMENT{State prep for $\matE_{(k+1)/2,k}=\matA_{p}^{\T}$}

\STATE ${\cal U}_{\matF_{\textrm{odd}}^{(1)}}\gets\textrm{qmsla.kronecker}({\cal U}_{\matA_{p-1}},{\cal U}_{\matA_{p+1}^{\T}})$
\COMMENT{State prep for $\matO_{i,k}=\matA_{p+1}\otimes\matA_{p-1}^{\T}$}

\ELSE 

\STATE $l_{\textrm{even}}\gets k/2$ and $l_{\textrm{odd}}\gets k/2$

\STATE ${\cal U}_{\matF_{\textrm{even}}^{(1)}}\gets\textrm{qmsla.kronecker}({\cal U}_{\matA_{p-1}},{\cal U}_{\matA_{p+1}^{\T}})$
\COMMENT{State prep for $\matE_{i,k}=\matA_{p+1}\otimes\matA_{p-1}^{\T}$}

\STATE ${\cal U}_{\matF_{\textrm{odd}}^{(1)}}\gets{\cal U}_{\matA_{p}^{\T}}$
\COMMENT{State prep for $\matO_{k/2,k}=\matA_{p}^{\T}$}

\ENDIF 

\FOR{$i=2$ {\bf to} $i=2k-p$} 

\IF{$p+i$ is odd} 

\STATE ${\cal U}_{\matO_{\textrm{i,k}}}\gets\textrm{qmsla.kronecker}({\cal U}_{\matA_{p-i}},{\cal U}_{\matA_{p+i}^{\T}})$

\STATE ${\cal U}_{\matF_{\textrm{odd}}^{(i)}}\gets\textrm{qmsla.kronecker}({\cal U}_{\matO_{\textrm{i,k}}},\textrm{qmsla.rvec}({\cal U}_{\matF_{\textrm{odd}}^{(i-1)}}))$

\ELSE 

\STATE ${\cal U}_{\matE_{\textrm{i,k}}}\gets\textrm{qmsla.kronecker}({\cal U}_{\matA_{p-i}},{\cal U}_{\matA_{p+i}^{\T}})$

\STATE ${\cal U}_{\matF_{\textrm{even}}^{(i)}}\gets\textrm{qmsla.kronecker}({\cal U}_{\matE_{\textrm{i,k}}},\textrm{qmsla.rvec}({\cal U}_{\matF_{\textrm{even}}^{(i-1)}}))$

\ENDIF 

\ENDFOR 

\STATE ${\cal U}_{\phi}\gets\textrm{qmsla.matrix\_vec}({\cal U}_{\matF_{\textrm{even}}^{(l_{\textrm{odd}})}},\textrm{qmsla.vec}({\cal U}_{\matF_{\textrm{odd}}^{(l_{\textrm{odd}})}}))$

\STATE 

\RETURN ${\cal U}_{\psi},{\cal U}_{\phi}$

\end{algorithmic}

\caption{\label{alg:main}\textsc{\noun{MVTracePrep}}}
\end{algorithm}

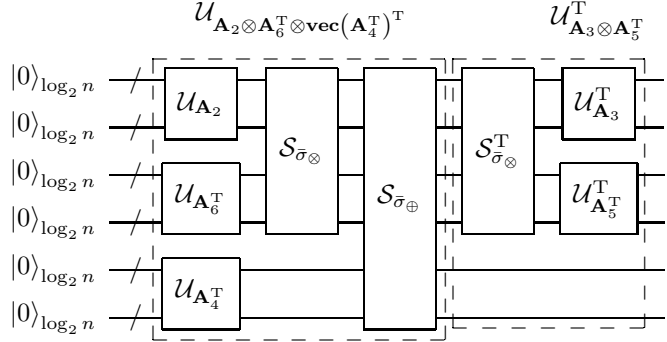
\begin{figure}[tph]
\noindent \begin{centering}
\Qcircuit @C=1em @R=.9em  {    & & & & & & & &  \\ & & & & {{\cal U}_{\mat A_{2}\otimes\mat A_{6}^{\text{T}}\otimes\vec{\mat A_{4}^{\text{T}}}^{\text{T}}}}   & & & {{\cal U}_{\mat A_{3}\otimes\mat A_{5}^{\text{T}}}^{\T}}&  \\ & & &  & & &  &&  \\ &\lstick{\Ket 0_{\log_{2}n}} & {/} \qw & \multigate{1}{{\cal U}_{{\mat A}_{2}}}	 	& \multigate{3}{{\cal S}_{\bar{\sigma}_{\otimes}}} 		& \multigate{5}{{\cal S}_{\bar{\sigma}_{\oplus}}} 	   & \multigate{3}{{\cal S}_{\bar{\sigma}_{\otimes}}^{\T}}		& \multigate{1}{{\cal U}_{{\mat A}_{3}}^{\T}}	& \qw \\ &\lstick{\Ket 0_{\log_{2}n}} & {/} \qw & \ghost{{\cal U}_{{\mat A}_{2}}}			    & \ghost{{\cal S}_{\bar{\sigma}_{\otimes}}}			  	& \ghost{{\cal S}_{\bar{\sigma}_{\oplus}}}		 	   & \ghost{{\cal S}_{\bar{\sigma}_{\otimes}}^{\T}}				& \ghost{{\cal U}_{{\mat A}_{3}}^{\T}} 	 & \qw \\ &\lstick{\Ket 0_{\log_{2}n}} & {/} \qw & \multigate{1}{{\cal U}_{{\mat A}_{6}^{\T}}} 	& \ghost{{\cal S}_{\bar{\sigma}_{\otimes}}}	   		    & \ghost{{\cal S}_{\bar{\sigma}_{\oplus}}}		 	   & \ghost{{\cal S}_{\bar{\sigma}_{\otimes}}^{\T}}				& \multigate{1}{{\cal U}_{{\mat A}_{5}^{\T}}^{\T}}			 & \qw \\ &\lstick{\Ket 0_{\log_{2}n}} & {/} \qw & \ghost{{\cal U}_{{\mat A}_{6}^{\T}}}			& \ghost{{\cal S}_{\bar{\sigma}_{\otimes}}}			    & \ghost{{\cal S}_{\bar{\sigma}_{\oplus}}}		 	   & \ghost{{\cal S}_{\bar{\sigma}_{\otimes}}^{\T}}				& \ghost{{\cal U}_{{\mat A}_{5}^{\T}}^{\T}}					 & \qw \\ &\lstick{\Ket 0_{\log_{2}n}} & {/} \qw & \multigate{1}{{\cal U}_{{\mat A}_{4}^{\T}}} 	&   \qw   											   	& \ghost{{\cal S}_{\bar{\sigma}_{\oplus}}}			   & \qw	& \qw									 				    & \qw \\ &\lstick{\Ket 0_{\log_{2}n}} & {/} \qw & \ghost{{\cal U}_{{\mat A}_{4}^{\T}}}			&  \qw    											  	& \ghost{{\cal S}_{\bar{\sigma}_{\oplus}}}			   & \qw	& \qw									 				    & \qw \gategroup{4}{4}{9}{6}{0.7em}{--} \gategroup{4}{7}{9}{8}{0.7em}{--}        } 
\par\end{centering}
\caption{\label{fig:main-alg}Visualization of the unitary transform\noun{
${\cal U}_{\protect\vphi}({\cal U}_{\protect\matA_{2}},\cdots,{\cal U}_{\protect\matA_{6}})$
}from the \textsc{\noun{MVTracePrep}} algorithm, where $\protect\matA_{i}\in\protect\C^{n\times n}$
.}
\end{figure}

\subsection{Warm-up I: $k=2$}

As a warmup, let us first consider the case of $k=2$, i.e. the trace
of the product of four matrices. To keep notation simple, we denote
the four matrices by $\matA,\matB,\matC$ and $\matD$, and we want
to compute $\Trace{\matA\matB\matC\matD}$. The proposed algorithm
is based on the following fact.
\begin{fact}
[Fact 7.4.9 from \cite{bernstein2009matrix}] Let $\matA\in\C^{n\times m}$,
$\matB\in\C^{m\times l}$, $\matC\in\C^{l\times k}$, $\matD\in\C^{k\times n}$.
Then,
\begin{equation}
\Trace{\matA\matB\matC\matD}=\vec{\matA}^{\T}(\matB\otimes\matD^{\T})\vec{\matC^{\T}}\label{eq:trace_of_matrix_multiplication}
\end{equation}
\end{fact}

This fact can easily be used to implement, via qMSLA operations, an
algorithm that takes ${\cal U}_{\matA},{\cal U}_{\matB},{\cal U}_{\matC},\text{ and }{\cal U}_{\matD}$
and outputs two state preparation circuits whose overlap is $\Trace{\matA\matB\matC\matD}$:
(in the parentheses we show which qMSLA operation is used in each
step)
\begin{description}
\item [{Step~1:}] Compute circuits for $\matC^{\T}$ and $\matD^{\T}$,
i.e. ${\cal U}_{\matC^{\T}},{\cal U}_{\matD^{\T}}$ ($\textrm{qmsla.tranpose}$
twice).
\item [{Step~2:}] Compute a circuit for $\matB\otimes\matD^{\T}$, i.e.
${\cal U}_{\matB\otimes\matD^{\T}}$ ($\textrm{qmsla.kronecker}$).
\item [{Step~3:}] Compute a circuit for $\vec{\matC^{\T}}$, i.e. ${\cal U}_{\vec{\matC^{\T}}}$
($\textrm{qmsla.vec}$).
\item [{Step~4:}] Compute a circuit for $\frac{(\matB\otimes\matD^{\T})\vec{\matC^{\T}}}{\FNorm{\matB}\FNorm{\matD}\FNorm{\matC}}\oplus\vxi'$
for some garbage $\vxi'\in\C^{kl(mn-1)}$ ($\textrm{qmsla.matrix\_vec})$.
Note that we use the fact that $\TNorm{\vec{\matC^{\T}}}=\FNorm{\matC}$
and $\FNorm{\matB\otimes\matD^{\T}}=\FNorm{\matB}\FNorm{\matD}.$
This is output ${\cal U}_{\vphi}$.
\item [{Step~5:}] Compute a circuit for $\overline{\matA}$, i.e. ${\cal U}_{\overline{\matA}}$
($\textrm{qmsla.conjugate}$).
\item [{Step~6:}] Compute a circuit for $\vec{\overline{\mat A}}$, i.e.
${\cal U}_{\vec{\overline{\matA}}}$ ($\textrm{qmsla.vec}$).
\item [{Step~7:}] Compute a circuit for $\vec{\overline{\mat A}}\oplus0_{kl(mn-1)}$
($\textrm{qmsla.pad\_zero\_rows}$). This is output ${\cal U}_{\vpsi}$.
\end{description}

Once we have ${\cal U}_{\vphi}$ and ${\cal U}_{\vpsi}$, we can estimate
the overlap in order to estimate $\Trace{\matA\matB\matC\matD}$.
For example, we can use Algorithm~\ref{alg:overlap}.
\begin{prop}
Consider circuits ${\cal U}_{\vpsi}$ and ${\cal U}_{\vphi}$ which
are the output of the steps described above. Consider the circuit
which the result of Algorithm~\ref{alg:overlap} applied to ${\cal U}_{\vpsi}$
and ${\cal U}_{\vphi}$. It is a state preparation circuit for the
vector 
\[
\frac{\Trace{\matA\matB\matC\matD}}{\FNorm{\matA}\FNorm{\matB}\FNorm{\matC}\FNorm{\matD}}\oplus\vxi
\]
for some vector $\vxi\in\C^{mnkl-1}$.
\end{prop}

\begin{proof}
Clearly ${\cal U}_{\vphi}$ is a state preparation circuit for $\vphi=\frac{(\matB\otimes\matD^{\T})\vec{\matC^{\T}}}{\FNorm{\matB}\FNorm{\matD}\FNorm{\matC}}\oplus\vxi'$
for some garbage $\vxi'\in\C^{kl(mn-1)}$, and ${\cal U}_{\vpsi}$
is a state preparation circuit for $\vpsi=\vec{\overline{\mat A}}\oplus0_{kl(mn-1)}$.
We have that,
\begin{align*}
\vpsi^{\conj}\vphi & =\left(\vec{\overline{\mat A}}\oplus0_{lk(mn-1)}\right)^{\conj}\left(\frac{(\matB\otimes\matD^{\T})\vec{\matC^{\T}}}{\FNorm{\matB}\FNorm{\matD}\FNorm{\matC}}\oplus\vxi'\right)\\
 & =\left(\vec{\matA}\oplus0_{lk(mn-1)}\right)^{\T}\left(\frac{(\matB\otimes\matD^{\T})\vec{\matC^{\T}}}{\FNorm{\matB}\FNorm{\matD}\FNorm{\matC}}\oplus\vxi'\right)\\
 & =\vec{\matA}^{\T}\frac{(\matB\otimes\matD^{\T})\vec{\matC^{\T}}}{\FNorm{\matB}\FNorm{\matD}\FNorm{\matC}}\\
 & =\frac{\Trace{\matA\matB\matC\matD}}{\FNorm{\matB}\FNorm{\matD}\FNorm{\matC}}
\end{align*}
Algorithm~~\ref{alg:overlap} on ${\cal U}_{\vpsi}$ and ${\cal U}_{\vphi}$
computes a state preparation circuit for the vector
\[
\frac{\vpsi^{\conj}\vphi}{\TNorm{\vpsi}\TNorm{\vphi}}\oplus\vxi=\frac{\Trace{\matA\matB\matC\matD}}{\FNorm{\matA}\FNorm{\matB}\FNorm{\matC}\FNorm{\matD}}\oplus\vxi
\]
for some garbage $\vxi\in\C^{mnkl-1}$. This is because $\TNorm{\vpsi}=1/\FNorm{\matA}$
and $\TNorm{\vphi}=1$ (since it is the result of $\text{qmsla.matrix\_vec}$).
\end{proof}

\subsection{\label{subsec:Warm-Up-II}Warm-up II: $k=3$}

For six matrices ($k=3$) there is no simple multivariate trace formula.
However, we can twice clamp together two of the matrices to reduce
the number of matrices to four, and apply Eq.~(\ref{eq:trace_of_matrix_multiplication}).
Further algebra reduces the trace to a formula that can be encoded
using qMSLA operations. More concretely:

\begin{align}
\Trace{\matA_{1}\matA_{2}\matA_{3}\matA_{4}\matA_{5}\matA_{6}} & =\Trace{\matA_{1}(\matA_{2}\matA_{3})\matA_{4}(\matA_{5}\matA_{6})}\nonumber \\
 & =\vec{\matA_{1}}^{\T}(\matA_{2}\matA_{3}\otimes\matA_{6}^{\T}\matA_{5}^{\T})\vec{\matA_{4}^{\T}}\nonumber \\
 & =\vec{\matA_{1}}^{\T}(\matA_{2}\otimes\matA_{6}^{\T})(\matA_{3}\otimes\matA_{5}^{\T})\vec{\matA_{4}^{\T}}\nonumber \\
 & =\Trace{\vec{\matA_{1}}^{\T}(\matA_{2}\otimes\matA_{6}^{\T})(\matA_{3}\otimes\matA_{5}^{\T})\vec{\matA_{4}^{\T}}}\nonumber \\
 & =\vec{\vec{\matA_{1}}^{\T}}^{\T}\left(\matA_{2}\otimes\matA_{6}^{\T}\otimes\vec{\matA_{4}^{\T}}^{\T}\right)\vec{\matA_{3}\otimes\matA_{5}^{\T}}\nonumber \\
 & =\vec{\matA_{1}}^{\T}\left(\matA_{2}\otimes\matA_{6}^{\T}\otimes\vec{\matA_{4}^{\T}}^{\T}\right)\vec{\matA_{3}\otimes\matA_{5}^{\T}}\label{eq:trace-of-6-matrices}
\end{align}
In the above, we applied Eq.~(\ref{eq:trace_of_matrix_multiplication})
to matrices $\matA_{1},\matA_{2}\matA_{3},\matA_{4}$ and \textbf{$\matA_{5}\matA_{6}$}.
In the third equality, we used the Kronecker mixed product property.
In the fourth equality we use the fact that $\Trace c=c$ for scalars.
Finally, we apply Eq.~(\ref{eq:trace_of_matrix_multiplication})
again, this time on the matrices $\vec{\matA_{1}}^{\T},\matA_{2}\otimes\matA_{6}^{\T},\matA_{3}\otimes\matA_{5}^{\T}$
and $\vec{\matA_{4}^{\T}}$.

Eq.~(\ref{eq:trace-of-6-matrices}) can be used to design, via qMSLA
operations, circuits ${\cal U}_{\vphi}$ and ${\cal U}_{\vpsi}$ for
\[
\vphi=\frac{\left(\matA_{2}\otimes\matA_{6}^{\T}\otimes\vec{\matA_{4}^{\T}}^{\T}\right)\vec{\matA_{3}\otimes\matA_{5}^{\T}}}{\FNorm{\matA_{2}}\cdots\FNorm{\matA_{6}}}\oplus\vxi'
\]
and 
\[
\vpsi=\vec{\overline{\matA_{1}}}\oplus0_{w}
\]
for some garbage $\vxi'$ and appropriate size $w$. Computing overlap
using Algorithm~~\ref{alg:overlap} produces the state
\[
\frac{\Trace{\matA_{1}\matA_{2}\cdots\matA_{6}}}{\FNorm{\matA_{1}}\cdots\FNorm{\matA_{6}}}\oplus\vxi
\]
for some garbage $\vxi$. We omit the details, and move directly to
arbitrary $k>2$, as the details are very similar to the case of $k=2$.

\subsection{Multivariate trace formula for $k>2$}

We now consider the trace of the product of $\matA_{1},\dots,\matA_{2k}$.
We clamp $\matA_{2}$ with $\matA_{3}$ and $\matA_{2k-1}$ with $\matA_{2k}$,
finding that we need to compute the trace of the product of $\matA_{1},\matA_{2}\matA_{3},\matA_{4},\dots\matA_{2k-2},\matA_{2k-1}\matA_{2k}$.
These are $2k-2$ matrices. So, we recursively apply the formula for
$2k-2$ inputs, and after using again the Kronecker mixed product
property and the fact that $\Trace c=c$ for scalars, we obtain a
multivariate trace formula. The details are somewhat technical, so
we first give the final result, and then prove it inductively.

First, we need a few additional notations. Given matrices $\matA_{1},\ldots\matA_{2k}$,
we define the following set of matrices
\[
\matE_{i,k}\coloneqq\matA_{2i}\otimes\matA_{2(k+1-i)}^{\T}\quad i=1,\dots,\left\lfloor \frac{k}{2}\right\rfloor 
\]
\[
\matE_{(k+1)/2,k}\coloneqq\matA_{k+1}^{\T}\quad(k\text{ is odd})
\]
\[
\matO_{i,k}\coloneqq\matA_{2i+1}\otimes\matA_{2(k-i)+1}^{\T}\quad i=1,\dots,\left\lfloor \frac{k-1}{2}\right\rfloor 
\]
\[
\matO_{k/2,k}\coloneqq\matA_{k+1}^{\T}\quad(k\text{ is even})
\]
Each of the matrices $\{\matE_{i,k}\}$ and $\{\matO_{i,k}\}$ pair
two matrices (except for the pivotal matrix, which is not paired)
of the sequence $\matA_{2},\dots,\matA_{2k}$ via the Kronecker product.
Each matrix appears exactly once. The $\matE$ matrices pair even
indexed matrices, and the $\matO$ matrices pair odd indexed matrices.
Next, we define a series of functions, $F^{(1)},F^{(2)},\dots$, where
$F^{(p)}$ is $p$-ary and given by: 
\begin{align*}
F^{(1)}(\matX) & \coloneqq\matX\\
F^{(p)}(\matX_{1},\dots,\matX_{p}) & \coloneqq\matX_{1}\otimes\vec{F^{(p-1)}(\matX_{2},\dots,\matX_{p})}
\end{align*}

\begin{thm}
\label{thm:main-trace}With the above notations, if $k$ is even 
\[
\MTrace{\matA_{1},\dots,\matA_{2k}}{2k}=\vec{\matA_{1}}^{\T}F^{(k/2)}(\matE_{1,k},\ldots,\matE_{k/2,k})\vec{F^{(k/2)}(\matO_{1,k},\dots,\matO_{k/2,k})}
\]
and if $k$ is odd 
\[
\MTrace{\matA_{1},\dots,\matA_{2k}}{2k}=\vec{\matA_{1}}^{\T}F^{((k+1)/2)}(\matE_{1,k},\ldots,\matE_{(k+1)/2,k})\vec{F^{((k-1)/2)}(\matO_{1,k},\dots,\matO_{(k-1)/2,k})}
\]
\end{thm}

\begin{proof}
We prove the theorem by induction on $k$. For the base case, note
that $k=2$ is exactly Eq. (\ref{eq:trace_of_matrix_multiplication}):
\begin{align*}
\Trace{\matA_{1}\cdots\matA_{4}} & =\vec{\matA_{1}}^{\T}\matA_{2}\otimes\matA_{4}^{\T}\vec{\matA_{3}^{\T}}\\
 & =\vec{\matA_{1}}^{\T}F^{(1)}(\matE_{1,k})\vec{F^{(1)}(\matO_{1,k})}
\end{align*}
The inductive step is slightly different if $k$ is even or odd.
However, the difference is very minor, and sums up with slightly modified
indices. Thus, for conciseness, we show the inductive step only for
even $k$ to odd $k+1$. That is, assuming the formula hold for an
even $k$, we prove that the following formula (which is the formula
for $k+1$) holds:

{\footnotesize{}
\[
\Trace{\matA_{1}\matA_{2}\cdots\matA_{2k}\matA_{2k+1}\matA_{2k+2}}=\vec{\matA_{1}}^{\T}F^{((k+2)/2)}(\matE_{1,k+1},\ldots,\matE_{(k+2)/2,k+1})\vec{F^{(k/2)}(\matO_{1,k+1},\ldots,\matO_{k/2,k+1})}
\]
}To see that, we first note that{\scriptsize{}
\begin{align*}
\Trace{\matA_{1}\cdots\matA_{2k+2}} & =\Trace{\matA_{1}(\matA_{2}\matA_{3})\cdots\matA_{2k}(\matA_{2k+1}\matA_{2k+2})}\\
 & =\vec{\matA_{1}}^{\T}F^{\left(\frac{k}{2}\right)}(\matA_{2}\matA_{3}\otimes\matA_{2k+2}^{\T}\matA_{2k+1}^{\T},\matO_{2,k+1},\ldots,\matO_{\frac{k}{2},k+1})\vec{F^{\left(\frac{k}{2}\right)}(\matE_{2,k+1},\ldots,\matE_{\frac{k+1}{2},k+1})}\\
 & =\vec{\matA_{1}}^{\T}\left(\matA_{2}\matA_{3}\otimes\matA_{2k+2}^{\T}\matA_{2k+1}^{\T}\otimes\vec{F^{\left(\frac{k-2}{2}\right)}(\matO_{2,k+1},\ldots,\matO_{\frac{k}{2},k+1})}^{\T}\right)\vec{F^{\left(\frac{k}{2}\right)}(\matE_{2,k+1},\ldots,\matE_{\frac{k+1}{2},k+1})}\\
 & =\vec{\matA_{1}}^{\T}(\matA_{2}\otimes\matA_{2k+2}^{\T})(\matA_{3}\otimes\matA_{2k+1}^{\T})\otimes\vec{F^{\left(\frac{k-2}{2}\right)}(\matO_{2,k+1},\ldots,\matO_{\frac{k}{2},k+1})}^{\T}\vec{F^{\left(\frac{k}{2}\right)}(\matE_{2,k+1},\ldots,\matE_{\frac{k+1}{2},k+1})}\\
 & =\vec{\matA_{1}}^{\T}(\matA_{2}\otimes\matA_{2k+2}^{\T})\left(\matA_{3}\otimes\matA_{2k+1}^{\T}\otimes\vec{F^{\left(\frac{k-2}{2}\right)}(\matO_{2,k+1},\ldots,\matO_{\frac{k}{2},k+1})}^{\T}\right)\vec{F^{\left(\frac{k}{2}\right)}(\matE_{2,k+1},\ldots,\matE_{\frac{k+1}{2},k+1})}
\end{align*}
}In the second equality, we applied the inductive assumption. However,
note that due to index shift the even labeled pairs ($\matE$ matrices)
become odd labeled pairs $(\matO$ matrices), except for the first
index. In third equality we apply the recursive definition of $F^{(k/2)}$,
and in the fourth equality we use the Kronecker mixed product property.
The fifth equality follows by the fact that for any two matrices $\matA,\matB$
and row vector $\x$, we have $(\matA\matB)\otimes\x=\matA(\matB\otimes\x)$~\cite[Fact 7.4.20]{bernstein2009matrix}.
Finally, we can always put a trace around scalars and utilize Eq.~(\ref{eq:trace_of_matrix_multiplication})
once again, to find that:{\scriptsize{}
\begin{align*}
\Trace{\matA_{1}\cdots\matA_{2k+2}} & =\vec{\matA_{1}}^{\T}\matA_{2}\otimes\matA_{2k+2}^{\T}\otimes\vec{F^{\left(\frac{k}{2}\right)}(\matE_{2,k+1},\ldots,\matE_{\frac{k+1}{2},k+1})}^{\T}\vec{\matA_{3}\otimes\matA_{2k+1}^{\T}\otimes\vec{F^{\left(\frac{k-2}{2}\right)}(\matO_{2,k+1},\ldots,\matO_{\frac{k}{2},k+1})}^{\T}}\\
 & =\vec{\matA_{1}}^{\T}F^{\left(\frac{k+2}{2}\right)}(\matE_{1,k+1},\ldots,\matE_{\frac{k+2}{2},k+1})\vec{F^{\left(\frac{k}{2}\right)}(\matO_{1,k+1},\ldots,\matO_{\frac{k}{2},k+1})}
\end{align*}
}where in the second equality we applied the recursive definition
of $F^{((k+2)/2)}$ and $F^{(k/2)}$ in reverse.
\end{proof}

\subsection{\label{subsec:proof-main}\textsc{\noun{MVTracePrep}}}

Algorithm \textsc{\noun{MVTracePrep}} implements the general trace
formula described in the previous subsection using qMSLA operations.
\begin{thm}
\label{thm:main}Given classical descriptions of matrix state preparation
circuits ${\cal U}_{\matA_{1}},\dots{\cal U}_{\matA_{2k}}$ of matrices
$\matA_{1},\dots,\matA_{2k}$, where the size of $\matA_{i}$ is $n_{i-1}\times n_{i}$
for $n_{0},\dots,n_{2k}$, such that $n_{0}=n_{2k}$, all of which
are powers of 2, Algorithm \textsc{\noun{MVTracePrep}}\noun{ }(Algorithm~\ref{alg:main})
outputs two (vector) state preparation circuits, ${\cal U}_{\vpsi}$
for a vector $\vpsi$ and ${\cal U}_{\vphi}$ for a vector $\vphi$,
on $q=2\sum_{i=1}^{2k}\log_{2}n_{i}$ qubits, such that 
\[
\Braket{\vpsi}{\vphi}=\frac{\Trace{\matA_{1}\cdots\matA_{2k}}}{\FNorm{\matA_{1}}\cdots\FNorm{\matA_{2k}}}.
\]
${\cal U}_{\vpsi}$ depends only on $\matA_{1}$, while ${\cal U}_{\vphi}$
depends only on $\matA_{2},\dots,\matA_{2k}$. The classical cost
of the algorithm is $O(\sum_{i=1}^{2k}g_{\matA_{i}})$. The depth
of ${\cal U_{\vpsi}}$ is $d_{\matA_{1}}$ and the gate complexity
is $g_{\matA_{1}}$. The depth of ${\cal U_{\vphi}}$ is $\max_{i\in[2,4,\cdots,2k]}d_{\matA_{i}}+\max_{i\in[3,5,\cdots,2k-1]}d_{\matA_{i}}$
and the gate complexity is $\sum_{i=2}^{2k}g_{\matA_{i}}$.
\end{thm}

\begin{proof}
Correctness of the algorithm follow Theorem~\ref{thm:main-trace},
since it simply implements the various parts of formula using qMSLA
operations. Indeed, ${\cal U}_{\matF_{\textrm{even}}^{(i)}}$ and
${\cal U}_{\matF_{\textrm{odd}}^{(i)}}$ are state preparation circuits
for 
\[
\matF_{\textrm{even}}^{(i)}\coloneqq F^{(i)}(\matE_{1,k},\ldots,\matE_{i,k})
\]
\[
\matF_{\textrm{odd}}^{(i)}\coloneqq F^{(i)}(\matO_{1,k},\ldots,\matO_{i,k})
\]
respectively, and ${\cal U}_{\matO_{i,k}}$ and ${\cal U}_{\matE_{i,k}}$
are state preparation for $\matO_{i,k}$ and $\matE_{i,k}$ (respectively).
However, further explanations are needed for lines 4 and 29. Since
the output of $\textrm{qmsla.matrix\_vec}$ is a state preparation
for the product with garbage, line 29 creates a state preparation
circuit for 
\[
\vphi=\begin{cases}
\frac{F^{(k/2)}(\matE_{1,k},\ldots,\matE_{k/2,k})\vec{F^{(k/2)}(\matO_{1,k},\dots,\matO_{k/2,k})}}{\FNorm{\matA_{2}}\cdots\FNorm{\matA_{2k}}}\oplus\vxi' & k\,\textrm{is even}\\
\frac{F^{((k+1)/2)}(\matE_{1,k},\ldots,\matE_{(k+1)/2,k})\vec{F^{((k-1)/2)}(\matO_{1,k},\dots,\matO_{(k-1)/2,k})}}{\FNorm{\matA_{2}}\cdots\FNorm{\matA_{2k}}}\oplus\vxi' & k\,\textrm{is odd}
\end{cases}
\]
To avoid the garbage contaminating the overlap, the algorithm set
$\vpsi$ to be equal to $\vec{\overline{\mat A}_{1}}$ padded with
zeros. Thus, line 4 creates a state preparation circuit for 
\[
\vpsi=\vec{\overline{\mat A}_{1}}\oplus0^{0\times\log_{2}n_{3}n_{4}\cdots n_{2k}}
\]
According to Theorem~\ref{thm:main-trace}, we have that $\vpsi^{\conj}\vphi=\frac{\Trace{\matA_{1}\cdots\matA_{2k}}}{\FNorm{\matA_{1}}\cdots\FNorm{\matA_{2k}}}$.

We are left with analyzing the gate and depth complexities, along
with the running time, considering qubit sizes. For ${\cal U}_{\vpsi}$,
following the conjugation and padding of ${\cal U}_{\matA_{1}}$,
the gate and depth complexities are $g_{\matA_{1}}$ and $d_{\matA_{1}}$,
respectively, resulting a running time of $O(g_{\matA_{1}})$. Subsequently,
Kronecker products and matrix vector product introduce no additional
complexity in terms of depth and gate complexity. Consequently, the
depth $d_{{\cal U}_{\phi}}$ is $\max_{i\in[2,4,\cdots,2k]}d_{\matA_{i}}+\max_{i\in[3,5,\cdots,2k-1]}d_{\matA_{i}}$,
and the gate complexity is $\sum_{i=2}^{2k}g_{\matA_{i}}$. The qubit
count of both circuits is determined by the sizes of matrices with
even indices, given that ${\cal U}_{\matA_{i}}$ has dimensions $n_{i-1}\times n_{i}$
for $n_{0},\dots,n_{2k}$, where $n_{0}=n_{2k}$. Thus, the qubit
count is $q=2\sum_{i=1}^{2k}\log_{2}n_{i}$.

Now consider the running time. In the description in Algorithm~\ref{alg:main}
we use qMSLA operation, each having complexity of $O(g(\text{output}))$.
Unfortunately, this does not imply that the final output of a sequence
of qMSLA operations has complexity of $O(g(\text{output}))$. Specifically,
the description of \textsc{\noun{MVTracePrep}} results in complexity
which is $O(\log k\sum_{i=2}^{2k}g_{\matA_{i}}$) (this is because,
each matrix appears in $O(\log k)$ qMSLA operations, and for each
one you pay for its gate complexity). However, a careful implementation
of the algorithm that writes parts of ${\cal U}_{\matA_{1}},\dots,{\cal U}_{\matA_{2k}}$
to their final place, and delays \noun{EliminatePermutations} to the
end, allows us to shave off the $\log k$ factor. We defer the details
to Appendix~\ref{sec:optimized-mvtp}.
\end{proof}

\section{\label{sec:est-prod-trace}Estimating multivariate traces}

Algorithm \textsc{\noun{MVTracePrep}} encodes $\SMTrace{\matA_{1},\dots,\matA_{2k}}{2k}$
using two state preparation circuits ${\cal U}_{\vphi}$ and ${\cal U}_{\vpsi}$,
where to overlap between $\Ket{\vphi}$ and $\Ket{\vpsi}$ is equal
to $\SMTrace{\matA_{1},\dots,\matA_{2k}}{2k}$ up to normalization.
These two circuits can be used in the context of a larger quantum
algorithm, or used to directly estimate $\SMTrace{\matA_{1},\dots,\matA_{2k}}{2k}$
using various algorithms for estimating overlap (e.g., Hadamard Test,
Swap Test and Algorithm~\ref{alg:overlap}; each can be combined
with amplitude estimation acceleration \cite[Example 4.5]{lin2022lecture}).
Here we analyze the overall cost needed for estimating $\SMTrace{\matA_{1},\dots,\matA_{2k}}{2k}$
up to additive $\epsilon$ error, and discuss implications in the
context of an end-to-end algorithm that start with matrices $\matA_{1},\dots,\matA_{2k}$
explicitly stored in classical memory. We focus on using the Hadamard
Test, and for simplicity assume the matrices are real. For complex
matrices we need to employ the Hadamard T twice (the real and imaginary
versions), but the analysis is essentially the same.

Following Eq.~(\ref{eq:HT}), we have 
\[
\HT_{{\cal U_{\vpsi}^{\conj}{\cal U}_{\vphi}}}\Ket 0_{1}\Ket 0_{q}=\frac{1}{2}(\Ket 0_{1}({\cal I}+{\cal U_{\vpsi}^{\conj}{\cal U}_{\vphi}})\Ket 0_{q}+\Ket 1_{1}({\cal I}-{\cal U_{\vpsi}^{\conj}{\cal U}_{\vphi}})\Ket 0_{q})
\]
for $q=2\sum_{i=1}^{2k}\log_{2}n_{i}$. Due to Eq~(\ref{eq:hadamard-overlap})
and the fact that the matrices are real we get that 
\[
\p(0)=\frac{1}{2}\left(1+\Trace{\frac{\matA_{1}\cdots\matA_{2k}}{\FNorm{\matA_{1}}\cdots\FNorm{\matA_{2k}}}}\right)
\]
where $\p(0)$ denotes the probability of measuring $0$ in the first
qubit (when only that qubit is measured). Next, we compute an approximation
$\tilde{\p}$ of $\p\coloneqq\p(0)$, and then we can build an approximation
$\tilde{t}\approx t\coloneqq\Trace{\frac{\matA_{1}\cdots\matA_{2k}}{\FNorm{\matA_{1}}\cdots\FNorm{\matA_{2k}}}}$
via 
\[
\tilde{\t}\coloneqq2\tilde{\p}-1
\]
Using Quantum Phase Estimation (QPE) we can estimate $\tilde{\p}$
s.t $\left|\p-\tilde{\p}\right|\leq\epsilon'$ with total complexity
(depth times shots) of $O(\epsilon'^{-1}\max_{i}(d_{{\cal \matA}_{i}}))$
(see~\cite[Section 4.2]{lin2022lecture}) and we have that, 
\begin{align*}
\left|\t-\tilde{\t}\right| & =\left|(2\p-1)-(2\tilde{\p}-1)\right|\\
 & =2\left|\p-\tilde{\p}\right|\\
 & \le2\epsilon'
\end{align*}
Thus, the total complexity is $O(\epsilon^{-1}\max_{i}(d_{{\cal \matA}_{i}}))$
for approximating $\Trace{\frac{\matA_{1}\cdots\matA_{2k}}{\FNorm{\matA_{1}}\cdots\FNorm{\matA_{2k}}}}$
to additive error $\epsilon$. Since $\frac{\Trace{\matA_{1}\cdots\matA_{2k}}}{\FNorm{\matA_{1}}\cdots\FNorm{\matA_{2k}}}=\Trace{\frac{\matA_{1}\cdots\matA_{2k}}{\FNorm{\matA_{1}}\cdots\FNorm{\matA_{2k}}}}$
, the total complexity is $O\left(\epsilon^{-1}\FNorm{\matA_{1}}\cdots\FNorm{\matA_{2k}}\max_{i}(d_{{\cal \matA}_{i}})\right)$
for approximating $\Trace{\matA_{1}\cdots\matA_{2k}}$ to $\epsilon$
additive error.

\textbf{}

\subsection{\label{sec:Classical-Approaches}Quantum multivariate trace estimation
with classical inputs}

Consider the case that the matrices $\matA_{1},\dotsm,\matA_{2k}\in\R^{n\times n}$
are given in classical memory, and our goal is to approximate $\SMTrace{\matA_{1},\dots,\matA_{2k}}{2k}$
to additive error $\epsilon$. For simplicity of analysis, we assume
that all matrices are $n\times n$. We do not assume sparsity of $\matA_{i}$
(i.e., for each $i$ we treat $\matA_{i}$ as dense), or \emph{any}
other special property, like positive definiteness. To employ \textsc{\noun{MVTracePrep}},
we need state preparation circuits ${\cal U}_{\matA_{1}},\dots,{\cal U}_{\matA_{2k}}$.
Given a classically stored matrix $\matA\in\C^{m\times n}$, a state
preparation circuit for $\matA$ can be built using the algorithm
described in~\cite{shende2005synthesis} in $O(mn)$. So the total
combined cost (classical, one time) of building ${\cal U}_{\matA_{1}},\dots,{\cal U}_{\matA_{2k}}$
is $O(kn^{2})$, and the depth of each circuit is $O(n^{2})$. The
costs of the output of the \textsc{\noun{MVTracePrep}} are
\[
q=4k\log_{2}n+1\quad\quad g=O(kn^{2})\quad\quad d=O(n^{2})
\]
The classical cost of \textsc{\noun{MVTracePrep}} is also $O(kn^{2})$.

Using Hadamard test and QPE, as in the previous subsection, we get
total quantum complexity for $\epsilon$ additive error of $O\left(\epsilon^{-1}\FNorm{\matA_{1}}\cdots\FNorm{\matA_{2k}}n^{2}\right)$.
Thus, the total cost, classical and quantum, of approximating $\SMTrace{\matA_{1},\dots,\matA_{2k}}{2k}$
to additive $\epsilon$ error with constant probability is 
\[
O\left(n^{2}(k+\epsilon^{-1}\FNorm{\matA_{1}}\cdots\FNorm{\matA_{2k}})\right).
\]
By taking the median of $O(\log(1/\delta))$ different executions
of the algorithm, we can boost the success probability to at least
$1-\delta$, so the running time is $O\left(n^{2}(k+\epsilon^{-1}\FNorm{\matA_{1}}\cdots\FNorm{\matA_{2k}}\log(1/\delta))\right)$
for an additive $(\epsilon,\delta)$ estimator.

\paragraph{Comparison to Classical Methods.}

We now compare our algorithm to multivariate trace estimators that
use only classical computation. The simplest classical approach is
to compute $\matA_{1}\matA_{2}\cdots\matA_{2k}$ and then sum the
diagonal entries. This algorithm requires $O(n^{3}k)$ arithmetic
operations, and provides an exact value for the matrix product trace.
Our algorithm reduces the dependence on $n$ from $n^{3}$ to $n^{2}$,
but provides a stochastic estimate.

More appropriate baselines are classical stochastic trace estimators.
The state-of-the-art algorithm is Hutch++~\cite{meyer2021hutchpp},
which requires $O(\epsilon^{-1}\log(1/\delta))$ matrix vector products
for a relative $(\epsilon,\delta)$ estimator for a symmetric positive
definite $\matA$. Computing the product of $\matA=\matA_{1}\cdots\matA_{2k}$
with a vector can be accomplished in $O(n^{2}k)$ using repeated matrix-vector
products. Overall, the cost of Hutch++ for the case we consider in
this section is $O(n^{2}k\epsilon^{-1}\log(1/\delta))$. In comparison,
in our algorithm the factor $n^{2}k$ does not appear with $\epsilon$
and $\delta$. So, as long as $\FNorm{\matA_{1}}\cdots\FNorm{\matA_{2k}}=o(k)$
we improve the running time compared to Hutch++, though Hutch++ guarantees
are relative and not additive.

However, one can argue that using parallel computing with $k$ compute
nodes can trivially reduce the cost of Hutch++ to $O(n^{2}\epsilon^{-1}\log(1/\delta))$,
and this is a fairer comparison since our algorithm uses $O(k$) qubits.
In fact, when input data is given in an unstructured dense classical
manner, it is hard to achieve any quantum advantage since encoding
the data will immediately incur an exponential cost in the number
of qubits. The classical algorithm for encoding a vector of size $2^{N}$
as the amplitudes of a $N$-qubit circuit (the algorithm from~\cite{shende2005synthesis},
which is implemented in various quantum computing frameworks) requires
$O(2^{N})$ gate complexity and depth. Yet, there is a tradeoff between
the number of qubits and depth: recently it was shown (\cite{araujo2021divide}),
that a quantum circuit with a depth of $O(N^{2})$ and $O(2^{N})$
qubits can effectively load a $2^{N}$-dimensional vector into a quantum
state (here the quantum system has ancilla qubits). Regardless, an
exponential cost is involved, and any exponential (or even polynomial)
reduction in cost (compared to classical) in any downstream algorithm
(as is the case for our algorithm) get nullified.

The main potential for our algorithm is when the input matrices are
huge matrices with compact quantum circuit description. In such scenarios
we have an encoding of the matrices in an exponentially sized Hilbert
space, but only pay costs that are (hopefully) linear or polynomial
in the number of qubits. Such situations can only occur for scenarios
where we can construct meaningful matrices directly on the quantum
state. The most likely candidates are Hamiltonians of quantum systems.
Another avenue worth exploring involves encoding a clever approximation
of matrix $\matA$, rather than encoding $\matA$ directly. One can
hope that such approximations could be generated using cost-effective
operations in the quantum domain, particularly those that have low
depth. Presumably, one can attempt to find such approximations that
are not necessarily cheaper in terms of matrix size, but more easily
encoded into the quantum state (depth wise), while preserving the
spectral properties of $\matA$ up to a small error. We leave exploring
these directions to future research.

We stress that the above discussion holds only for positive definite
matrices, while our algorithm is more general and can be even used
to estimate multivariate traces of non-symmetric input matrices.

\subsection{From multivariate traces to spectral sums}

Multivariate trace estimation can be used for spectral sum estimation,
via polynomial approximation. A well established approach is to estimate
$\Trace{f(\matA)}$ using $\Trace{p(\matA)}$ where $p(\cdot)$ is
a polynomial that approximates $f(\cdot)$. We do not consider how
$p(\cdot)$ is computed; we focus solely on how $\Trace{p(\matA)}$
can be computed via multivariate traces. 

We suggest two methods for this task. The first is to use matrix moments.
Suppose $p(x)=\sum_{k=0}^{L}c_{k}x^{k}$ is a $L$-degree polynomial.
Then, 
\[
\Trace{p(\matA)}=\sum_{k=0}^{L}c_{k}\Trace{\matA^{k}},
\]
Therefore, we can approximate $\Trace{\matA^{k}}$ for $k=0,\ldots,L$
using our algorithm\footnote{For odd degrees we add a state preparation circuit for the identity
matrix as input}. As long as we have access to enough qubits, we can computed them
in parallel. Using these trace values, the sum is computed classically.

The second approach translates $\Trace{p(\matA)}$ to a single multivariate
trace. We first rewrite $p(x)$ in a factored form: 
\[
p(x)=c_{L}(x-r_{1})(x-r_{2})\cdots(x-r_{L})
\]
The shifts $r_{1},\dots,r_{L}$ are the roots of $p(\cdot)$, and
as long as the polynomial has positive degree they exist (perhaps
with multiplicity in some of the roots). Now, 
\[
p(\matA)=c_{L}(\matA-r_{1}\matI)(\matA-r_{2}\matI)\cdots(\matA-r_{L}\matI)
\]
so $\Trace{p(\matA)}=c_{L}\MTrace{\matA-r_{1}\matI,\dots,\matA-r_{L}\matI}L$.

It is well known that the roots of a polynomial expressed as a linear
combination of monomials, i.e. $p(x)=\sum_{k=0}^{L}c_{k}x^{k}$, are
equal to the eigenvalues of the $L\times L$ companion matrix: 
\[
\begin{bmatrix}1 & 0 & \cdots & 0 & 0\\
0 & 1 & \cdots & 0 & 0\\
\vdots & \vdots & \ddots & \vdots & \vdots\\
0 & 0 & \cdots & 1 & 0\\
-c_{1}/c_{0} & -c_{2}/c_{0} & \cdots & -c_{L-1}/c_{0} & -c_{L}/c_{0}
\end{bmatrix}
\]
Indeed, MATLAB's roots function implements exactly this method. However,
this method has numerical stability issues, and there are better methods
based on expressing the polynomial in the Chebyshev basis. See \cite[Chapter 18]{trefethen2013approximation}
for a discussion.

The polynomial factorization based approach requires $L$ distinct
but related input matrix state preparation circuits: each one of them
is a preparation circuit of a $\matA$ shifted by some constant. This
is in contrast to the moments based method, which requires a single
preparation circuit (for $\matA$ only). Thus, this approach is especially
appealing in cases where we can generate circuits of shifted variants
of $\matA$ once we have some base preparation circuit $\matA$ itself.

\section{Conclusions}

In this work, we have presented a new quantum algorithm for approximating
multivariate traces. To that end we introduce the qMSLA framework.
qMSLA builds a set of fundamental matrix algebra building block operations
for circuits that encode matrices in quantum states. qMSLA facilitates
the seamless composition of basic operations to construct circuits,
as we demonstrate with our construction of circuits that encode multivariate
traces. This adaptability broadens its applicability to various linear
algebra problems.

We view the qMSLA-based approach as a promising paradigm for designing
quantum algorithm for QLA and QML via high level building blocks.
We envision several avenues for future expansion of qMSLA:
\begin{itemize}
\item \textbf{Enrich level 2 operations}: Expanding the existing set of
level 2 operations within qMSLA will further enhance its versatility
and enable the efficient construction of low-depth circuits for complex
matrix computations.
\item \textbf{Extended matrix state preparation circuits}: Developing an
extended version of the qMSLA framework that can handle ``garbage''
in ancilla qubits of inputs and outputs is crucial. Quantum computations
often involve intermediate results that are not essential for the
final outcome. A framework capable of efficiently dealing with such
garbage data will significantly improve the overall efficiency.
\end{itemize}

\subsection*{Acknowledgments.}

This research was supported by the US-Israel Binational Science Foundation
(Grant no. 2017698), Israel Science Foundation (Grant no. 1524/23)
and IBM Faculty Award. Liron Mor-Yosef acknowledges support by the
Milner Foundation and the Israel Council for Higher Education.

\bibliographystyle{plain}
\bibliography{references}

\appendix

\begin{center}
\textbf{\Huge{}Appendix}{\Huge\par}
\par\end{center}

The structure of the appendix is outlined as follows:
\begin{itemize}
\item In Appendix~\ref{sec:L0-qMSLA}, we provide a detailed description
of level 0 qMSLA operations. These fundamental circuit-level operations
serve as building blocks for higher levels of qMSLA operations.
\item Appendix~\ref{sec:Pseudo-Code-of-level1} gives implementation details
for level 1 qMSLA. In particular, Appendix~\ref{sec:Pseudo-Code-of-level1}
defines a data structure for describing matrix state preparation circuits
and provides pseudocodes for all level 1 operations.
\item Appendix~\ref{sec:optimized-mvtp} describes a more efficient variant
of \noun{MVTracePrep}.
\end{itemize}

\section{\label{sec:L0-qMSLA}Level 0 qMSLA}

In order to describe the implementation of qMSLA in a concise yet
complete, framework independent manner, we baseqMSLA's implementation
on a set of basic circuit level operations. We group these operations
as ``level 0 qMSLA''. However, we stress that these operations are
at the circuit level, and do not operate on state preparation circuits.
In the next appendix, we explain how to implement all level 1 qMSLA
operations in terms of level 0 operations. We emphasize that while
our description of level 0 operations is independent of various quantum
computing frameworks (e.g. QISKIT, Q\#, PennyLane, etc.), nearly all
operations are present in such frameworks under various names. There
are a couple of operations that are not always present (circuit transpose
and conjugate). For these we provide implementation details.

As a design choice, all qMSLA operation we describe (here, and in
subsequent sections) are assumed to be non-destructive; they accept
circuit(s) and return a new circuit, without overwriting inputs. Of
course, in the context of an algorithm that is implemented using a
sequence of qMSLA operations, it is possible to use destructive operations
for more efficient implementation. We view this as a type of complier
optimization, which we leave for future research.

Table~\ref{tab:qmsla-table-l0} provides a summary of all level 0
qMSLA operations. In the following subsections, we provide additional
details.

\begin{table}[tph]
\begin{centering}
{\footnotesize{}}%
\begin{tabular}{|c|c|c|c|c|c|c|}
\hline 
{\footnotesize{}Input} & {\footnotesize{}Output} & {\footnotesize{}$q(\text{output})$} & {\footnotesize{}$g(\text{output})$} & {\footnotesize{}$d(\text{output})$} & {\footnotesize{}Operation Name} & {\footnotesize{}Subsection}\tabularnewline
\hline 
\hline 
{\footnotesize{}$q$} & {\footnotesize{}${\cal I}_{q}$} & {\footnotesize{}$q$} & {\footnotesize{}$0$} & {\footnotesize{}$0$} & {\footnotesize{}$\textrm{qmsla.qc\_empty}$} & {\footnotesize{}\ref{subsec:create-new-circuits}}\tabularnewline
\hline 
{\footnotesize{}${\cal W},{\cal Q}$} & {\footnotesize{}$\mathcal{Q}\otimes\mathcal{W}$} & {\footnotesize{}$q({\cal W})+q({\cal Q})$} & {\footnotesize{}$g({\cal W})+g({\cal Q})$} & {\footnotesize{}$\max(d({\cal W}),d(Q))$} & {\footnotesize{}$\textrm{qmsla.qc\_tensor}$} & {\footnotesize{}\ref{subsec:create-new-circuits}}\tabularnewline
\hline 
{\footnotesize{}${\cal W},{\cal Q},\sigma$} & \multirow{1}{*}{{\footnotesize{}$\mathcal{S}_{\sigma^{-1}}\mathcal{Q}\mathcal{S}_{\sigma}\mathcal{W}$}} & {\footnotesize{}$q({\cal W})$} & {\footnotesize{}$g({\cal W})+g({\cal Q})$} & {\footnotesize{}$d({\cal W})+d(Q)$} & {\footnotesize{}$\textrm{qmsla.qc\_compose}$} & {\footnotesize{}\ref{subsec:create-new-circuits}}\tabularnewline
\hline 
{\footnotesize{}${\cal {\cal Q}}$, ${\cal G}$,$\delta$} & {\footnotesize{}${\cal G}_{\delta}{\cal {\cal Q}}$} & {\footnotesize{}$q({\cal Q})$} & {\footnotesize{}$g({\cal Q})+1$} & {\footnotesize{}$d({\cal Q})+1$} & {\footnotesize{}$\texttt{\textrm{qmsla.qc\_add\_gate}}$} & {\footnotesize{}\ref{subsec:add-gate}}\tabularnewline
\hline 
{\footnotesize{}${\cal U}$} & {\footnotesize{}${\cal U^{\T}}$} & {\footnotesize{}$q({\cal U})$} & {\footnotesize{}$g({\cal U})$} & {\footnotesize{}$d({\cal U})$} & {\footnotesize{}$\textrm{qmsla.qc\_transpose}$} & {\footnotesize{}\ref{subsec:transpose-conjugate-inverse}}\tabularnewline
\hline 
{\footnotesize{}${\cal U}$} & {\footnotesize{}$\overline{{\cal U}}$} & {\footnotesize{}$q({\cal U})$} & {\footnotesize{}$g({\cal U})$} & {\footnotesize{}$d({\cal U})$} & {\footnotesize{}$\textrm{qmsla.qc\_conjugate}$} & {\footnotesize{}\ref{subsec:transpose-conjugate-inverse}}\tabularnewline
\hline 
{\footnotesize{}${\cal U}$} & {\footnotesize{}${\cal U}^{-1}$} & {\footnotesize{}$q({\cal U})$} & {\footnotesize{}$g({\cal U})$} & {\footnotesize{}$d({\cal U})$} & {\footnotesize{}$\textrm{qmsla.qc\_inverse}$} & {\footnotesize{}\ref{subsec:transpose-conjugate-inverse}}\tabularnewline
\hline 
\hline 
{\footnotesize{}${\cal W},\sigma$} & {\footnotesize{}$\mathcal{S}_{\sigma}\mathcal{W}$}\tablefootnote{Equality here holds only on the ground state. That is, if ${\cal U}$
is the output, we only require that ${\cal U}\Ket 0_{q({\cal W})}={\cal S}_{\sigma}{\cal W}\Ket 0_{q({\cal W})}$.} & {\footnotesize{}$q({\cal W})$} & {\footnotesize{}$g({\cal W})$} & {\footnotesize{}$d({\cal W})$} & {\footnotesize{}$\textrm{qmsla.qc\_permute\_bits}$} & {\footnotesize{}\ref{subsec:premute_bits}}\tabularnewline
\hline 
\end{tabular}{\footnotesize\par}
\par\end{centering}
{\tiny{}\caption{\label{tab:qmsla-table-l0}Level 0 qMSLA operations. These are basic
circuit level operations that allow us to implement higher level qMSLA
operations. The classical cost for these operations remains consistently
$O(g(output))$. In the table, ${\cal U},{\cal W},{\cal Q}$ are quantum
circuits, ${\cal G}$ is a quantum gate, $\sigma$ is a permutation,
and $\delta$ is a list of qubit indices.}
}{\tiny\par}
\end{table}

\subsection{\label{subsec:create-new-circuits}Creating new circuits:\protect \\
${\cal U}\gets\textrm{\textrm{qmsla.qc\_empty}(q)}$, ${\cal U}\gets\textrm{\texttt{\textrm{qmsla.qc\_tensor}}}(\mathcal{Q},{\cal W})$,\protect \\
 ${\cal U}\gets\textrm{qmsla.qc\_compose}({\cal W},Q,\sigma)$\protect \\
}

The first set of level 0 operations concern creating a new circuit.
Circuits can be created in three ways: initiating an empty circuit,
taking the tensor product of two circuit, and composing two circuits.

\paragraph{${\cal U}\gets\textrm{qmsla.qc\_empty}(q)$:}

Create a new empty (no gates) circuit ${\cal U}$ on $q$ qubits.
The unitary matrix associated with the circuit is the identity matrix,
i.e. $\matM({\cal U})=\matI_{2^{q}}$.

\paragraph{${\cal U}\gets\textrm{\texttt{\textrm{qmsla.qc\_tensor}}}(\mathcal{Q},{\cal W})$:}

Create new circuit ${\cal U}$ on $q(Q)+q({\cal W})$ qubits. Apply
${\cal Q}$ to the $q({\cal Q})$ MSB qubits, and ${\cal W}$ to the
$q({\cal W})$ LSB qubits. The corresponding unitary is the Kronecker
(aka tensor) product of the unitaries, i.e. $\text{\ensuremath{\cirmat{{\cal U}}}=\ensuremath{\cirmat{{\cal Q}}}\ensuremath{\ensuremath{\otimes\cirmat{{\cal W}}}}}$.

\paragraph{${\cal U}\gets\textrm{qmsla.qc\_compose}({\cal W},Q,\sigma)$:}

Create a new circuit ${\cal U}$ on $\max(q({\cal Q}),q({\cal W}))$
qubits which is made of first applying ${\cal W}$ (on the MSB $q({\cal W})$
qubits), and then applying ${\cal Q}$. However, for $i=1,\dots,q({\cal W})$
qubit $\sigma(i)$ in ${\cal W}$'s output is wired to qubit $i$
in ${\cal Q}$. Mathematically, the operation is ${\cal U}=\mathcal{S}_{\sigma^{-1}}\mathcal{Q}\mathcal{S}_{\sigma}\mathcal{W}$.
However, no actual SWAP gates are used. We note that this is exactly
the ``compose'' operation in QISKIT.

\subsection{\label{subsec:add-gate}Adding a gate:\protect \\
${\cal U}\gets\textrm{qmsla.qc\_add\_gate(\ensuremath{\mathcal{Q}}, \ensuremath{{\cal G}}, \ensuremath{\delta})}$}

The $\textrm{qmsla.qc\_add\_gate}$\_function constructs a new circuit
by composing an existing circuit ${\cal Q}$ with the gate ${\cal G}$,
which is applied the qubits of ${\cal Q}$ whose indices appear in
the qubit list $\delta$. The gate is added at the end of the circuits,
on the output of ${\cal Q}$, where the wires in $\delta$ outgoing
from ${\cal Q}$ are connected to ${\cal G}$ using the ordering in
$\delta$. Mathematically, the operation is ${\cal U}={\cal G}_{\delta}{\cal Q}$,
where ${\cal G}_{\delta}$ denotes the circuit on $q({\cal Q})$ qubits
that consist solely of ${\cal G}$ applied to the qubits listed in
$\delta$.

\subsection{\label{subsec:transpose-conjugate-inverse}Transpose, conjugate,
and inverse of a circuit:\protect \\
${\cal U}^{\protect\T}\gets\textrm{qmsla.qc\_transpose}({\cal U})$,
$\overline{{\cal U}}\gets\textrm{qmsla.qc\_conjugate}({\cal U})$,\protect \\
${\cal U}^{\protect\conj}\gets\textrm{qmsla.qc\_inverse}({\cal U})$}

The operations $\textrm{qmsla.qc\_transpose}$, $\textrm{qmsla.qc\_conjugate}$,
$\textrm{qmsla.qc\_inverse}$ implement transpose of the unitary,
conjugate of the unitary, and the combination of both to obtain the
inverse (adjoint) of the unitary (respectively). To obtain the inverse/adjoint
circuit ${\cal U}^{-1}$ from a given circuit ${\cal U}$, reverse
the order of gates and conjugate each gate. To create a conjugated
circuit $\overline{{\cal U}}$ from a given circuit ${\cal U}$, conjugate
each gate without changing the order. To create a transposed circuit
${\cal U}^{\T}$ from a circuit ${\cal U}$, we simply invert and
conjugate. The inverse operation is typically implemented in most
quantum computing frameworks (e.g., the ``inverse'' function in
QISKIT), while the transpose and conjugate operations are relatively
straightforward to implement in most frameworks, based on the description
above. We summarize in the following lemma.
\begin{lem}
~
\begin{enumerate}
\item Given a classical description for a circuit ${\cal U}$ with $n$
gates we can create a classical description of a circuit ${\cal U}^{\T}$
for which $\matM({\cal U}^{\T})=\matM({\cal U})^{\T}$ in $O(n)$
operations. The depth and number of gates in the circuit are the same
as for ${\cal U}$.
\item Given a classical description for a circuit ${\cal U}$ with $n$
gates we can create a classical description of a circuit $\overline{{\cal U}}$
for which $\matM(\overline{{\cal U}})=\overline{\matM({\cal U})}$
in $O(n)$ operations. The depth and number of gates in the circuit
are the same as for ${\cal U}$.
\item Given a classical description for a circuit ${\cal U}$ with $n$
gates we can create a classical description of a circuit ${\cal U}^{\conj}$
for which $\matM({\cal U}^{\conj})=\matM({\cal U})^{\conj}$ in $O(n)$
operations. The depth and number of gates in the circuit are the same
as for ${\cal U}$.
\end{enumerate}
\end{lem}

\subsection{\label{subsec:premute_bits}Efficient qubit rearrangement:\protect \\
${\cal U}\gets\textrm{qmsla.qc\_premute\_bits}({\cal W},\sigma)$}

Given a circuit ${\cal W}$ and qubit permutation $\text{\ensuremath{\sigma}}$,
the goal of $\text{qmlsa.qc\_permute\_bits }$ is to return a circuit
that applies the qubit permutation $\sigma$ on the output of ${\cal W}$,
e.g. return ${\cal S}_{\sigma}{\cal W}$. However, for our purposes,
it suffices that this holds only when the input is the ground state.
That is, $\textrm{qmsla.qc\_premute\_bits}({\cal W},\sigma)$ returns
a circuit ${\cal U}$ such that ${\cal U}\Ket 0_{q({\cal W})}={\cal S}_{\sigma}{\cal W}\Ket 0_{q({\cal W})}$.
Since we do require equality for all input state, just for the ground
state, we can implement ${\cal U}$ without any SWAP gates, and with
the same depth as ${\cal U}.$ A pseudocode description appears is
provided in Algorithm~\ref{alg:qc_premute_bits}. We remark $\textrm{qmsla.qc\_premute\_bits}$
is different from other level 0 operations in that it utilizes level
0 qMSLA operations, yet it does not qualify for level 1 operation
due to its lack of clear linear algebraic interpretation.

\begin{algorithm}[tph]
\begin{algorithmic}[1]

\STATE \textbf{Input: }Classical description of circuit ${\cal W}$
and permutation $\sigma$

\STATE 

\STATE ${\cal W}^{-1}\gets\textrm{qmsla.qc\_inverse}({\cal W})$

\STATE ${\cal I}\gets\textrm{qmsla.qc\_empty}(q({\cal W}))$

\STATE ${\cal Q}\gets\textrm{qmsla.qc\_compose}({\cal I},{\cal W}^{-1},\sigma^{-1})$

\STATE $\mathcal{U}\gets\textrm{qmsla.qc\_inverse}({\cal Q})$

\RETURN $\mathcal{U}$

\end{algorithmic}

\caption{\label{alg:qc_premute_bits}$\textrm{qmsla.qc\_premute\_bits}$}
\end{algorithm}

Mathematically, we can verify that $\textrm{qmsla.qc\_premute\_bits}$
indeed achieves its desired goals:
\begin{align}
\mathcal{U}\Ket 0_{q} & ={\cal Q}^{-1}\Ket 0_{q}\nonumber \\
 & =(\calS_{\sigma_{1}}{\cal W}^{-1}\calS_{\sigma_{1}^{-1}}{\cal I})^{-1}\Ket 0_{q}\nonumber \\
 & =(\calS_{\sigma_{1}}({\cal S}_{\sigma_{1}}{\cal W})^{-1}{\cal I})^{-1}\Ket 0_{q}\nonumber \\
 & ={\cal I}^{-1}{\cal S}_{\sigma_{1}}{\cal W}\calS_{\sigma_{1}^{-1}}\Ket 0_{q}\nonumber \\
 & ={\cal S}_{\sigma_{1}}{\cal W}\Ket 0_{q}\label{eq:premute_bits}
\end{align}
where we used the fact that for any permutation $\sigma$ we have
that $\calS_{\sigma}\Ket 0_{q}=\Ket 0_{q}$. This demonstrates that
the output circuit ${\cal U}$ applies the permutation circuit ${\cal S}_{\sigma}$
to the original circuit ${\cal Q}$ on the ground state, as intended.

\paragraph{Relationship with \textsc{EliminatePermutations.}}

qMSLA operation $\textrm{qmsla.qc\_premute\_bits}$ is key in applying
the \textsc{EliminatePermutations }process\textsc{. }Suppose that
for some output state preparation circuit ${\cal U}_{\matX}$ we have
identified a circuit ${\cal W}$ and two permutations $\sigma_{1},\sigma_{2}$
such that $\mathcal{U}_{\matX}={\cal S}_{\sigma_{1}}{\cal W}{\cal S}_{\sigma_{2}}$.
Then ${\cal U}_{\matX}^{'}=\text{qc\_permute\_bits}({\cal W},\sigma_{1})$
is also a state preparation circuit for $\matX$. This can be verified
as follows:
\begin{align*}
\mathcal{U}_{\matX}^{'}\Ket 0_{q} & ={\cal S}_{\sigma_{1}}{\cal Q}\Ket 0_{q}\\
 & ={\cal S}_{\sigma_{1}}{\cal Q}{\cal S}_{\sigma_{2}}\Ket 0_{q}\\
 & =\mathcal{U}_{\matX}\Ket 0_{q}
\end{align*}
In main text of the paper, wherever we write \textsc{EliminatePermutations}
(in circuit diagrams and text itself) we mean that when preparing
that pseudocode for Appendix \ref{sec:Pseudo-Code-of-level1} we have
identified the permutations and applied $\text{qc\_permute\_bits}$.

\section{\label{sec:Pseudo-Code-of-level1}Implementation details for level
1 qMSLA}

In this section, we give implementation details (pseudocode) for all
level 1 qMSLA operations.

Pseudocodes are provided in Algorithms \ref{alg:Pseudo-code-of-operations}
and \ref{alg:Pseudo-code-of-operations-2}. However, to understand
the pseudocode we need first to discuss how matrix state preparation
circuits are represented therein. A state preparation circuit is a
structure with three fields: the circuit (${\cal U}$), the number
of rows in the prepared matrix ($m$), and number of columns ($n$).
In the pseudocode, a state preparation circuit is constructed using
the constructor $\text{MatrixStatePreperation}({\cal U},m,n)$. The
various fields of a state preparation structure ${\cal U}_{\matA}$
are accesses in the pseudocode as follows: ${\cal U}_{\matA}.m$ for
number of rows, ${\cal U}_{\matA}.n$ for the number of columns, and
${\cal U}_{\matA}$ for the circuit itself. Pseudocodes also use two
logical accessors: ${\cal U}_{\matA}.\text{creg}$ and ${\cal U}_{\matA}.\text{rreg}$.
${\cal U}_{\matA}.\text{creg}$ returns the list of qubits that correspond
to the column register, ${\cal U}_{\matA}.\text{rreg}$ returns the
list of qubits that correspond to the row register.

If qubits are provided as a numerical list, we use colon notation:
$0:n$ represents elements from $0$ to $n-1$ (inclusive), and $m:n$
represents elements from $m$ to $n-1$ (inclusive). Concatenating
lists is denoted by the plus operator. The $\text{shift}$ operator
shifts the indices in a list by a specified number.

\begin{algorithm}[tph]
\begin{algorithmic}[1]

\STATE \textbf{\uline{qmsla.identity(\mbox{$n$}):}}

\STATE $q\gets\log_{2}n$

\STATE ${\cal I}_{2q}\gets\textrm{\textrm{qmsla.qc\_empty(\ensuremath{2q})}}$

\STATE ${\cal U}\gets\textrm{\textrm{qmsla.qc\_add\_gate(\ensuremath{{\cal I}_{2q}}, \ensuremath{\calH}, \ensuremath{0:q})}}$\COMMENT{Add Hadamard gates on each qubit of the columns register}

\STATE ${\cal U}\gets\text{qmsla.qc\_add\_gate}({\cal U},\text{CNOT},\textrm{\textrm{\ensuremath{0:q}, \ensuremath{q:2q})}}$\COMMENT{Add CNOT gates on row qubit controlled by corresponding column qubits}

\RETURN $\textrm{\textrm{MatrixStatePreparation}}({\cal U},n,n)$

~

\setcounter{ALC@line}{0}

\STATE \textbf{\uline{qmsla.matrix(\mbox{${\cal U}$}):}}

\STATE $q\gets q({\cal U})$

\STATE ${\cal U}_{\matI_{2^{q}}}\gets\textrm{\textrm{qmsla.identity}(\ensuremath{2q({\cal U})})}$

\STATE ${\cal U}_{\matM({\cal U})}\gets\textrm{qmsla.qc\_compose(\ensuremath{{\cal U}_{\matI_{2^{q({\cal U})}}}},\,\ensuremath{{\cal U}_{\matI_{2^{q({\cal U})}}}}.\textrm{rreg})}$

\RETURN $\textrm{\textrm{MatrixStatePreparation}}({\cal U}_{\matM({\cal U})},2^{q},2^{q})$

~

\setcounter{ALC@line}{0}

\STATE \textbf{\uline{qmsla.conjugate(\mbox{${\cal U_{\matA}}$}):}}

\STATE ${\cal U}_{\overline{\matA}}\gets\textrm{qmsla.qc\_conjugate(\ensuremath{{\cal U}_{\matA}})}$

\RETURN $\textrm{\textrm{MatrixStatePreparation}}({\cal U}_{\overline{\matA}},{\cal U}_{\matA}.m,{\cal U}_{\matA}.n)$

~

\setcounter{ALC@line}{0}

\STATE \textbf{\uline{qmsla.transpose(\mbox{${\cal U}_{\matA}$}):}}

\STATE $\ensuremath{\sigma_{\T}}\gets\textrm{\ensuremath{{\cal U}_{\matA}}.creg+\ensuremath{{\cal U}_{\matA}}.rreg}$

\STATE ${\cal U}_{\matA^{\T}}\gets\textrm{\textrm{qmsla.qc\_permute\_bits}(\ensuremath{{\cal U}_{\matA}},\,\ensuremath{\sigma_{\T}})}$

\RETURN $\textrm{\textrm{MatrixStatePreparation}}({\cal U}_{\matA^{\T}},{\cal U}_{\matA}.n,{\cal U}_{\matA}.m)$

~

\setcounter{ALC@line}{0}

\STATE \textbf{\uline{qmsla.vec(\mbox{${\cal U}_{\matA}$}):}}

\RETURN $\textrm{\textrm{MatrixStatePreparation}}({\cal U}_{\matA},{\cal U}_{\matA}.m\cdot{\cal U}_{\matA}.n,1)$

\textbf{}

\end{algorithmic}

\caption{\label{alg:Pseudo-code-of-operations}Pseudocode for all level-1 qMSLA
operations (part 1)}
\end{algorithm}

\begin{algorithm}[tph]
\begin{algorithmic}[1]

\STATE \textbf{\uline{qmsla.pad\_zero\_columns(\mbox{${\cal U}_{\matA}$},
\mbox{$k$}):}}

\STATE ${\cal I}_{k+q({\cal U}_{\matA})}\gets\textrm{\textrm{qmsla.qc\_empty}}(k+q({\cal U}_{\matA}))$

\STATE ${\cal U}_{\matA\oplus0_{0\times(2^{k}-1)n}}\gets\textrm{\textrm{\textrm{qmsla.qc\_compose}(\ensuremath{{\cal I}_{k+q({\cal U}_{\matA})}},\,\ensuremath{{\cal U}_{\matA}},\,\ensuremath{k:(\log_{2}mn+k)})}}$

\RETURN $\textrm{\textrm{MatrixStatePreparation}}({\cal U}_{\matA\oplus0_{0\times(2^{k}-1)n}},{\cal U}_{\matA}.m,(2^{k}-1)\cdot{\cal U}_{\matA}.n)$

~

\setcounter{ALC@line}{0}

\STATE \textbf{\uline{qmsla.matrix\_vec(\mbox{${\cal U_{\matA}}$},
\mbox{${\cal U}_{\b}$}):}}

\STATE ${\cal U}_{\b}^{\T}\gets\textrm{\textrm{qmsla.qc\_transpose}}({\cal U}_{\b})$

\STATE ${\cal W}\gets\text{\textrm{qmsla.qc\_compose}(\ensuremath{{\cal U}_{\matA}},\,\ensuremath{{\cal U}_{\b}^{\T}},\,\ensuremath{{\cal U}_{\matA}}.\textrm{creg})})$

\STATE ${\cal U}_{\frac{\matA\b}{\FNorm{\matA}\TNorm{\b}}\oplus\vxi}\gets\textrm{qmsla.vec}(\text{MatrixStatePreperation}({\cal W},{\cal U}_{\matA}.m,{\cal U}_{\matA}.n))$

\RETURN ${\cal U}_{\frac{\matA\b}{\FNorm{\matA}\TNorm{\b}}\oplus\vxi}$

~

\setcounter{ALC@line}{0}

\STATE \textbf{\uline{qmsla.kronecker(\mbox{${\cal U}_{\matA}$},
\mbox{${\cal U}_{\matB}$}):}}

\STATE $\sigma_{\otimes}\gets\textrm{\ensuremath{{\cal U}_{\matA}}.rreg+\text{shift}(\ensuremath{{\cal U}_{\matB}}.rreg,\,\ensuremath{q_{\matA}})+\ensuremath{{\cal U}_{\matA}}.creg+\text{shift}(\ensuremath{{\cal U}_{\matB}}.creg},\,q_{\matA})$

\STATE ${\cal U}\gets\textrm{qmsla.qc\_tensor(\ensuremath{{\cal U}_{\matA}}, \ensuremath{{\cal U}_{\matB}})}$

\STATE ${\cal U}\gets\textrm{\textrm{qmsla.qc\_permute\_bits}}({\cal U},\sigma_{\otimes})$

\RETURN $\textrm{MatrixStatePreparation}({\cal U},{\cal U}_{\matA}.m\cdot{\cal U}_{\matB}.m,{\cal U}_{\matA}.n\cdot{\cal U}_{\matB}.n$)

~

\setcounter{ALC@line}{0}

\STATE \textbf{\uline{qmsla.kronecker(\mbox{${\cal U}_{\matA_{1}},\cdots,{\cal U}_{\matA_{k}}$}):}}\textbf{
}(Assume $k\ge2$)

\STATE ${\cal U}\gets\textrm{\textrm{qmsla.qc\_empty}}(\sum_{1}^{k}q({\cal U}_{i}))$

\FOR{$i=1$ {\bf to} $i=k$} 

\STATE ${\cal U}\gets\textrm{qmsla.qc\_compose(}{\cal U},{\cal U}_{\matA_{i}},\sum_{j=1}^{i-1}\log_{2}m_{j}n_{j}:\sum_{j=0}^{i}\log_{2}m_{j}n_{j})$

\ENDFOR 

\STATE {\scriptsize{}$\sigma_{\otimes}\gets\textrm{\ensuremath{{\cal U}_{\matA_{1}}}.rreg + \text{shift}(\ensuremath{{\cal U}_{\matA_{2}}}.rreg,\,\ensuremath{q_{\matA_{1}}}) +\ensuremath{\cdots}+ \text{shift}(\ensuremath{{\cal U}_{\matA_{k}}}.rreg,\,\ensuremath{\sum_{i=1}^{k-1}}\ensuremath{q_{\matA_{i}}}) +\ensuremath{{\cal U}_{\matA_{1}}}.creg + \text{shift}(\ensuremath{{\cal U}_{\matA_{2}}}.creg,\,\ensuremath{q_{\matA_{1}}}) +\ensuremath{\cdots}+\text{shift}(\ensuremath{{\cal U}_{\matA_{k}}}.creg,\,\ensuremath{\sum_{i=1}^{k-1}}\ensuremath{q_{\matA_{i}}})}$}{\scriptsize\par}

\STATE ${\cal U}_{\matA_{1}\otimes\matA_{2}\otimes\cdots\otimes\matA_{k}}\gets\textrm{\textrm{qmsla.qc\_permute\_bits}}({\cal U},\sigma_{\otimes})$

\RETURN $\text{MatrixStatePreperation}({\cal U}_{\matA_{1}\otimes\matA_{2}\otimes\cdots\otimes\matA_{k}},\prod_{i=1}^{k}{\cal U}_{\matA_{i}}.m,\prod_{i=1}^{k}{\cal U}_{\matA_{i}}.n)$

~

\textbf{}

\end{algorithmic}

\caption{\label{alg:Pseudo-code-of-operations-2}Pseudocode for all level-1
qMSLA operations (part 2)}
\end{algorithm}

\FloatBarrier

\section{\label{sec:optimized-mvtp}A more efficient \noun{MVTracePrep}}

In this section, we propose \textsc{\noun{MVTracePrep}}\textsc{Optimized}
- a more efficient variant of \noun{MVTracePrep} than the variant
in the main text. \noun{MVTracePrep} translates Theorem \ref{thm:main}
directly into a circuit via the series of qMSLA operations. This showcases
the appeal of qMSLA's approach. Each qMSLA operation is efficient,
with cost that is proportional to the size of the output (i.e. $O(g(\text{output}))$),
so the overall cost is is rather attractive: $O(\log k\sum_{i=2}^{2k}g_{\matA_{i}})$.
Nevertheless, careful inspection of the final output reveals that
a more efficient construction is possible, one that costs only $O(\sum_{i=2}^{2k}g_{\matA_{i}})$,
though it is no longer a straightforward translation of Theorem \ref{thm:main}.

The key observation is that if we do not apply \noun{EliminatePermutation}
at all, the output circuit ${\cal U}_{\vphi}$ has a very special
structure. First, there is a layer which consists of the tensor product
of a subset of the circuits $\{{\cal U}_{\matA_{i}}\}$. Then there
is a series of SWAPs that implements a permutation. Finally, there
is another layer which consists of the tensor product of another subset
of the circuits $\{{\cal U}_{\matA_{i}}\}$. See Figure~\ref{fig:main-alg}
for an illustration. Thus, to build the circuit, we can build the
first layer, build the last layer, compute the middle permutation,
and finally compose the first layer with the final layer with the
qubit permutation in a gate efficient manner via $\text{qmsla.qc\_compose}$.

\textsc{\noun{MVTracePrep}}\textsc{Optimized }(Algorithm\textsc{~\ref{alg:TraceProdPrepOptimized}})
does exactly this. The main challenge is in finding the qubit permutation
of the SWAP layer. \textsc{\noun{MVTracePrep}}\textsc{Optimized} finds
the permutation by tracking the execution of the various qMSLA operations
in \noun{MVTracePrep} logically, i.e. computing only the qubit permutation
they induce. See Algorithm~\ref{alg:GeneratePermutation} for pseudocode
for implementing this. \noun{MVTracePrepOptimized} uses Algorithm
~\ref{alg:GeneratePermutation} as a subroutine.

\begin{algorithm}[tph]
\begin{algorithmic}[1]

\STATE \textbf{Input: }List of reindexed qubit registers from 0 to
total number of qubits: $[\textrm{rreg}_{1},\textrm{creg}_{1},\cdots,\textrm{rreg}_{2k},\textrm{creg}_{2k}]$

\STATE $p\gets k+1$

\IF{$k$ is odd} 

\STATE $l_{\textrm{even}}\gets(k+1)/2$ and $l_{\textrm{odd}}\gets(k-1)/2$

\STATE $\tau_{\matF_{\textrm{even}}^{(1)}}\gets\textrm{\{creg=\textrm{rreg}\ensuremath{{}_{p}}, rreg=\ensuremath{\textrm{creg}_{p}}\}}$\COMMENT{Permutation for $\matE_{(k+1)/2,k}=\matA_{p}^{\T}$}

\STATE $\tau_{\matF_{\textrm{odd}}^{(1)}}\gets\textrm{\{creg=\ensuremath{\textrm{creg}_{p+1}}+\textrm{rreg}\ensuremath{{}_{p-1}}, rreg=\ensuremath{\ensuremath{\textrm{rreg}_{p+1}+\textrm{creg}_{p-1}}}\}}$\COMMENT{Permutation for $\matO_{i,k}=\matA_{p+1}\otimes\matA_{p-1}^{\T}$}

\ELSE 

\STATE $l_{\textrm{even}}\gets k/2$ and $l_{\textrm{odd}}\gets k/2$

\STATE $\tau_{\matF_{\textrm{even}}^{(1)}}\gets\textrm{\{creg=\ensuremath{\textrm{creg}_{p+1}}+\textrm{rreg}\ensuremath{{}_{p-1}}, rreg=\ensuremath{\ensuremath{\textrm{rreg}_{p+1}+\textrm{creg}_{p-1}}}\}}$\COMMENT{Permutation for $\matE_{i,k}=\matA_{p+1}\otimes\matA_{p-1}^{\T}$}

\STATE $\tau_{\matF_{\textrm{odd}}^{(1)}}\gets\textrm{\{creg=\textrm{rreg}\ensuremath{{}_{p}}, rreg=\ensuremath{\textrm{creg}_{p}}\}}$\COMMENT{Permutation for $\matO_{k/2,k}=\matA_{p}^{\T}$}

\ENDIF 

\FOR{$i=2$ {\bf to} $i=2k-p$} 

\IF{$p+r$ is odd} 

\STATE $\tau_{\matO_{\textrm{i,k}}}\gets\textrm{\{creg=\ensuremath{\textrm{creg}_{p+i}}+\textrm{rreg}\ensuremath{{}_{p-i}}, rreg=\ensuremath{\ensuremath{\textrm{rreg}_{p+i}+\textrm{creg}_{p-i}}}\}}$

\STATE ${\cal \tau}_{\matF_{\textrm{odd}}^{(i)}}\gets\textrm{\{creg=\ensuremath{\tau_{\matO_{\textrm{i,k}}}\textrm{.creg}}, rreg=\ensuremath{\tau_{\matO_{\textrm{i,k}}}\textrm{.rreg}} +\ensuremath{{\cal \tau}_{\matF_{\textrm{odd}}^{(i-1)}}}.creg +\ensuremath{{\cal \tau}_{\matF_{\textrm{odd}}^{(i-1)}}}.rreg\}}$

\ELSE 

\STATE $\tau_{\matE_{\textrm{i,k}}}\gets\textrm{\{creg=\ensuremath{\textrm{creg}_{p+i}}+\textrm{rreg}\ensuremath{{}_{p-i}}, rreg=\ensuremath{\ensuremath{\textrm{rreg}_{p+i}+\textrm{creg}_{p-i}}}\}}$

\STATE ${\cal \tau}_{\matF_{\textrm{even}}^{(i)}}\gets\textrm{\{creg=\ensuremath{\tau_{\matE_{\textrm{i,k}}}\textrm{.creg}}, rreg=\ensuremath{\tau_{\matE_{\textrm{i,k}}}\textrm{.rreg}} +\ensuremath{{\cal \tau}_{\matF_{\textrm{even}}^{(i-1)}}}.creg +\ensuremath{{\cal \tau}_{\matF_{\textrm{even}}^{(i-1)}}}.rreg\}}$

\ENDIF 

\ENDFOR 

\STATE $\sigma_{\textrm{even}}^{-1}\gets\textrm{\ensuremath{\matF_{\textrm{even}}^{(l_{\textrm{even}})}}.creg + \ensuremath{\matF_{\textrm{even}}^{(l_{\textrm{even}})}}.rreg}$

\STATE $\sigma_{\textrm{odd}}^{-1}\gets\textrm{ \ensuremath{\matF_{\textrm{odd}}^{(l_{\textrm{odd}})}}.creg + \ensuremath{\matF_{\textrm{odd}}^{(l_{\textrm{odd}})}}.rreg}$

\RETURN $\sigma\coloneqq\sigma_{\textrm{even}}\sigma_{\textrm{odd}}^{-1}$

\end{algorithmic}

\caption{\label{alg:GeneratePermutation}\noun{Generate}\textsc{\noun{MVTracePrep}}\noun{Permutation}}
\end{algorithm}

\begin{algorithm}[tph]
\begin{algorithmic}[1]

\STATE \textbf{Input: }Classical description of the circuits ${\cal U}_{\matA_{1}},\cdots{\cal U}_{\matA_{2k}}$where
$\matA_{i}\in\C^{m_{i}\times n_{i}}$

\STATE 

\STATE \COMMENT{Constructing ${\cal U}_{\psi}$} 

\STATE ${\cal U}_{\overline{\matA}_{1}}\gets\textrm{qmsla.conjugate}({\cal U}_{\matA_{1}})$

\STATE ${\cal U}_{\vpsi}\gets\textrm{qmsla.pad}(\textrm{qmsla.vec}({\cal U}_{\overline{\matA}_{1}}),\log_{2}n_{3}n_{4}\cdots n_{2k},0)$

\STATE 

\STATE \COMMENT{Constructing ${\cal U}_{\phi}$} 

\STATE  $\textrm{all\_regs}\gets[]$

\STATE  $\textrm{idx}=0$

\FOR{$i=1$ {\bf to} $i=2k$} 

\STATE  $\textrm{reg\_i}\gets\textrm{idx}+[{\cal U}_{\matA_{i}}.\textrm{rreg}+{\cal U}_{\matA_{i}}.\textrm{creg}]$
\COMMENT{concat qubits list and add idx for each element in the list}

\STATE  $\textrm{idx}\gets\textrm{idx}+\textrm{size(reg\_i)}$

\STATE  $\textrm{all\_regs}\gets\textrm{all\_regs.append(reg\_i)}$

\ENDFOR 

\STATE $\sigma\gets\textrm{GenerateMVTracePrepPermutation(all\_regs)}$

\STATE ${\cal U_{\textrm{even}}}\gets{\cal U}_{\matA_{2}}$

\STATE ${\cal U}_{\textrm{odd}}\gets{\cal U}_{\matA_{3}}^{-1}$

\FOR{$i=2$ {\bf to} $i=k$} 

\STATE ${\cal U_{\textrm{even}}}\gets\textrm{qmsla.qc\_tensor({\cal \ensuremath{{\cal U}_{\textrm{even}}}}, \ensuremath{{\cal U}_{\matA_{2i}}})}$

\IF{$2i$ < $2k-2i$} 

\STATE ${\cal U_{\textrm{even}}}\gets\textrm{qmsla.qc\_tensor({\cal \ensuremath{{\cal U}_{\textrm{even}}}}, \ensuremath{{\cal U}_{\matA_{2k-2i}}})}$

\ENDIF 

~

\STATE ${\cal U}_{\textrm{odd}}\gets\textrm{qmsla.qc\_tensor(\ensuremath{{\cal U}_{\textrm{odd}}}, \ensuremath{{\cal U}_{\matA_{2i+1}}^{-1}})}$

\IF{$2i$ < $2k-2i$} 

\STATE ${\cal U}_{\textrm{odd}}\gets\textrm{qmsla.qc\_tensor(\ensuremath{{\cal U}_{\textrm{odd}}}, \ensuremath{{\cal U}_{\matA_{2k-2i-1}}^{-1}})}$

\ENDIF 

\ENDFOR 

\STATE ${\cal U}_{\vphi}\gets\textrm{qmsla.compose(\ensuremath{{\cal U}_{\textrm{even}}}, \ensuremath{{\cal U}_{\textrm{odd}}}, \ensuremath{\sigma})}$

\RETURN ${\cal U}_{\vpsi},{\cal U}_{\vphi}$

\end{algorithmic}

\caption{\textsc{\label{alg:TraceProdPrepOptimized}}\textsc{\noun{MVTracePrep}}\textsc{Optimized}}
\end{algorithm}

\FloatBarrier

\end{document}